\documentclass[format=acmsmall, review=false, nonacm]{acmart}
\usepackage{booktabs} %
\usepackage[noend,ruled,linesnumbered]{algorithm2e} %

\SetAlFnt{\small}
\SetAlCapFnt{\small}
\SetAlCapNameFnt{\small}
\SetAlCapHSkip{0pt}
\IncMargin{-\parindent}

\SetKwInput{KwI}{Input}
\SetKwInput{KwO}{Output}
\SetKw{KwR}{return}
\SetKw{KwNo}{no}
\SetKw{KwOr}{or}
\SetKw{KwAnd}{and}

\SetCommentSty{mycommfont}
\SetKwComment{Comment}{$\triangleright$\ }{}

\usepackage{amsmath,graphicx,amsthm,dsfont,mathtools}
\usepackage{color,soul}
\usepackage{xspace}
\usepackage{enumerate}
\usepackage{multicol}
\usepackage{booktabs}
\usepackage{paralist}
\usepackage{graphicx}
\usepackage{latexsym}
\usepackage{subcaption}
\usepackage{xifthen}
\usepackage{multirow}

\usepackage[textsize=tiny,textwidth=1.5cm,linecolor=green!70!black, backgroundcolor=green!10, bordercolor=black]{todonotes}

\usepackage{cleveref}
\crefname{table}{Table}{Tables}
\crefname{figure}{Figure}{Figures}
\crefname{theorem}{Theorem}{Theorems}
\crefname{definition}{Definition}{Definitions}
\crefname{corollary}{Corollary}{Corollaries}
\crefname{observation}{Observation}{Observations}
\crefname{lemma}{Lemma}{Lemmas}
\crefname{example}{Example}{Examples}
\crefname{reduction}{Reduction}{Reductions}
\crefname{construction}{Construction}{Constructions}
\crefname{subsection}{Section}{Sections}
\crefname{section}{Section}{Sections}
\crefname{proposition}{Proposition}{Propositions}
\crefname{algorithm}{Algorithm}{Algorithms}
\crefname{drule}{Rule}{Rules}
\crefname{claim}{Claim}{Claims}
\crefname{clm}{Claim}{Claims}
\crefname{appendix}{Appendix}{Appendix}
\crefname{remark}{Remark}{Remark}

\newtheorem{theorem}{Theorem}
\newtheorem{definition}{Definition}
\newtheorem*{definition*}{Definition}

\newtheorem{corollary}{Corollary}
\newtheorem{clm}{Claim}

\newtheorem{proposition}{Proposition}
\newtheorem{observation}{Observation}
\newtheorem{example}{Example}

\newtheorem{remark}{Remark}
\newtheorem{claim}{Claim}[theorem]
\newenvironment{claimproof}[1]
{\begin{proof}
    \renewcommand{\qedsymbol}{\hfill~(end of the proof of~\cref{#1}~$\diamond$)}}
  {\end{proof}}

\newcommand{\wred}[1]{{\color{red!50!black}{#1}}}
\newcommand{\dg}[1]{{\color{green!50!black}{#1}}}

\newcommand{\RR}{\ensuremath{\mathcal{R}}}
\newcommand{\aaa}{\ensuremath{\mathcal{C}}}
\newcommand{\alt}{\ensuremath{c}}
\newcommand{\vvv}{\ensuremath{\mathcal{V}}}
\newcommand{\ppp}{\ensuremath{\mathcal{P}}}
\newcommand{\app}{\ensuremath{\mathsf{A}}}
\newcommand{\VV}{\ensuremath{\mathsf{V}}}
\newcommand{\kk}{\ensuremath{k}}
\newcommand{\vr}{\ensuremath{\mathsf{\rho}}}
\newcommand{\AV}{\text{\normalfont{AV}}}
\newcommand{\SAV}{\text{\normalfont{SAV}}}
\newcommand{\CC}{\text{\normalfont{CC}}}
\newcommand{\PAV}{\text{\normalfont{PAV}}}
\newcommand{\W}{\ensuremath{W}}

\newcommand{\nn}{\ensuremath{\hat{n}}}
\newcommand{\mm}{\ensuremath{\hat{m}}}

\newcommand{\TUcs}[1]{\ensuremath{\TU_{#1}\text{-core-stable}}}
\newcommand{\NTUcs}[1]{\ensuremath{\NTU_{#1}\text{-core-stable}}}

\newcommand{\TUcore}[1]{\ensuremath{\TU_{#1}\text{-core}}}
\newcommand{\NTUcore}[1]{\ensuremath{\NTU_{#1}\text{-core}}}

\newcommand{\NTUblocks}[1]{\ensuremath{\NTU_{#1}\text{-blocks}}}

\newcommand{\TUblocking}[1]{\ensuremath{\TU_{#1}\text{-blocking}}}
\newcommand{\NTUblocking}[1]{\ensuremath{\NTU_{#1}\text{-blocking}}}

\newcommand{\TUutil}[1]{\ensuremath{\TU_{#1}\text{-}\util}}
\newcommand{\NTUutil}[1]{\ensuremath{\NTU_{#1}\text{-}\util}}
\newcommand{\votone}{\ensuremath{v^1}}
\newcommand{\setaltone}{\ensuremath{\mathsf{s}^1}}
\newcommand{\vottwo}{\ensuremath{v^2}}
\newcommand{\setalttwo}{\ensuremath{\mathsf{s}^2}}

\newcommand{\dummyalt}{\ensuremath{\alt_d}}

\newcommand{\true}{\ensuremath{\mathsf{T}}}
\newcommand{\false}{\ensuremath{\mathsf{F}}}

\usepackage{pifont}
\newcommand{\yes}{\dg{\ding{51}}}
\newcommand{\no}{\wred{\ding{55}}}
\newcommand{\open}{open}
\newcommand{\votsep}{voter-separable}
\newcommand{\altsep}{alternative-separable}
\newcommand{\totsep}{totally separable}
\newcommand{\totsepity}{total separability}
\newcommand{\cval}[1]{\ensuremath{\TU_{#1}\text{-}\chi}}          %
\newcommand{\cvalvr}{\cval{\vr}}%

\DeclareMathOperator*{\argmax}{arg\,max}

\newcommand{\coalsize}{\ensuremath{\frac{|\vvv|}{\kk}}}
\newcommand{\share}{\ensuremath{\lfloor\frac{|\vvv'|\kk}{|\vvv|}\rfloor}}
\newcommand{\shareS}{\ensuremath{\lfloor\frac{|S|\kk}{|\vvv|}\rfloor}}

\newcommand{\myemph}[1]{{\color{green!40!black}\emph{#1}}}
\renewcommand{\myemph}[1]{{\emph{#1}}}
\newcommand{\mypara}[1]{\noindent\textbf{#1}}
\newcommand{\myparait}[1]{\noindent\emph{#1}}

\newcommand{\NP}{\textsf{NP}}
\newcommand{\NPh}{\wred{\NP-h}}
\newcommand{\NPc}{\wred{\NP-c}}
\newcommand{\NPhh}{\NP-hard}
\newcommand{\NPcc}{\NP-complete}

\newcommand{\coNP}{\textsf{coNP}}
\newcommand{\CoNP}{\textsf{CoNP}}

\newcommand{\coNPc}{\wred{\coNP-{c}}}
\newcommand{\coNPhh}{\coNP-hard}
\newcommand{\coNPcc}{\coNP-complete}

\newcommand{\DP}{\textsf{DP}}

\newcommand{\DPc}{\wred{\DP-c}}
\newcommand{\DPhh}{\DP-hard}
\newcommand{\DPcc}{\DP-complete}

\newcommand{\PP}{\textsf{P}}
\newcommand{\poly}{\dg{Poly time}}
\newcommand{\hard}{{Hard:}}

\newcommand{\ThetaTwoP}{\ensuremath{\Theta_2^{\text{P}}}}
\newcommand{\SigmaTwoP}{\ensuremath{\Sigma_2^{\text{P}}}}

\newcommand{\mw}{MW}
\newcommand{\MWGame}{\mw\ game}

\newcommand{\TU}{{\textsf{TU}}}
\newcommand{\NTU}{\textsf{NTU}}
\newcommand{\alloc}{\ensuremath{\boldsymbol{\alpha}}}

\newcommand{\util}{\ensuremath{\mu}}

\newcommand{\wvector}{\ensuremath{\boldsymbol{\lambda}}}
\newcommand{\xx}{\ensuremath{\boldsymbol{x}}}

\newcommand{\weight}{\ensuremath{\lambda}}

\newcommand{\score}[1]{\ensuremath{\mathsf{sc}_{#1}}}

\newcommand{\scoreTU}[1]{\score{#1}}
\newcommand{\scoreNTU}[1]{\score{#1}}
\newcommand{\decprob}[3]{%

  \smallskip   
  \begin{minipage}{0.94\linewidth}%
    \textsc{#1}\\
    \textbf{Input:} #2\\
    \textbf{Question:} #3
  \end{minipage}%
  \smallskip

  \par
}

\newcommand{\BicliqueL}{\textsc{Balanced Complete Bipartite Subgraph}}
\newcommand{\Biclique}{\textsc{Biclique}}

\newcommand{\MWGCore}[3]{\textsc{#2\text{-}#1\text{-}#3}}

\newcommand{\CE}{\textsc{Core-NonEmpty}}
\newcommand{\CV}{\textsc{Core-Membership}}

\newcommand{\MWGTUCexist}[1]{\MWGCore{\TU}{#1}{\CE}}

\newcommand{\MWGNTUCexist}[1]{\MWGCore{\NTU}{#1}{\CE}}

\newcommand{\MWGTUCverif}[1]{\MWGCore{\TU}{#1}{\CV}}

\newcommand{\MWGNTUCverif}[1]{\MWGCore{\NTU}{#1}{\CV}}

\newcommand{\XTCDP}{\textsc{RX3C-UNRX3C}}

\usepackage{thm-restate}

\usepackage{etoolbox} %

\newcommand{\mytitle}{Multi-Winner Voting Games in \TU\ and \NTU: When is the Core Always Non-Empty?}

\setcitestyle{acmnumeric}

\title[Multi-Winner Voting Games]{\mytitle}

\author{Jiehua Chen}
\author{Christian Hatschka}
\authorsaddresses{}
\begin{abstract}
  Multi-winner approval voting selects a size-$k$ committee that aggregates voters' approval preferences over a set of alternatives.
  A central question is coalitional stability: No coalition should be able to pick a committee---of size at most its proportional share---under which every coalition member has strictly more approved alternatives. 
  This notion, introduced by Aziz et al.\ (2017) as \emph{core-stable committees}, is naturally interpreted as a core notion with non-transferable utility. 
  Existence in full generality remains open. %

  We introduce \emph{multi-winner voting games}, a cooperative-game framework that unifies prior work and supports a systematic study
  of two utility-transfer models across different voting rules. 
  Players are voters.
  Each coalition has a proportional seat cap and may only propose \emph{admissible} committees up to that size.
  Fixing a multi-winner rule, each admissible committee induces a utility vector for the members of the coalition.

  In the \emph{transferable utility} (\TU) model, a coalition may redistribute the total utility of an admissible committee among its members.
  In the \emph{non-transferable utility} (\NTU) model, a coalition may only use utility vectors that are realized directly by some admissible committee.
  The \emph{core} consists of utility vectors feasible for the grand coalition that are not \emph{blocked} by any coalition.
  A coalition is \emph{blocking} if it can propose an admissible committee that makes all its members strictly better off, directly in \NTU\ and after redistribution in \TU.
  When instantiated with the standard \PAV/approval utility, the \NTU-core is equivalent to the core-stable committee concept studied in prior work.
  To our knowledge, the \TU-core for multi-winner voting has not been previously studied.

  We analyze core existence and computation for four prominent rules: Approval Voting (\AV), Satisfaction Approval Voting (\SAV), Chamberlin--Courant (\CC), Proportional Approval Voting (\PAV). 
  For \AV\ and \SAV, we show that the \TU-core is always non-empty. We also give a polynomial-time algorithm that exploits the linearity of these rules to compute a utility vector in the core. We show that the algorithm applies to a larger class of voting rules that even includes some voting rules for linear preferences.
  For \CC\ and \PAV, the \TU-core can be empty and finding a core utility vector is computationally hard.
  Under \NTU, the games induced by \AV, \SAV, and \PAV\ are core-equivalent via a monotone transformation of utilities, so they capture the same stability object.
  For \CC, we prove that the \NTU-core is always non-empty and give a simple greedy algorithm that computes a utility vector in the \NTU-core.
  Finally, we show that \emph{core membership}, deciding whether a given utility vector lies in the core, is computationally hard in most settings.

\end{abstract}

\begin{document}

\begin{titlepage}

  \maketitle

  \vspace{1cm}
  \setcounter{tocdepth}{1} %
  \tableofcontents

\end{titlepage}

\allowdisplaybreaks
\section{Introduction}\label{sec:intro}

In multi-winner approval voting, we are given a finite set~$\aaa$ of alternatives and a finite set~$\vvv$ of voters.
Each voter~$v_i \in \vvv$ is associated with an approval set~$\app_i \subseteq \aaa$; the alternatives in~$\app_i$ are approved by $v_i$ and the remaining alternatives are not approved.
Given a target size~$\kk$, the goal is to select a committee of exactly~$\kk$ alternatives that aggregates voters' approval preferences according to a specified voting rule.

A widely studied class consists of \emph{score-based} multi-winner rules.
Under such a rule, each voter assigns a score to each committee, and a size-$\kk$ \emph{winning} committee maximizes the total score (summed over voters).
Four prominent examples are Approval Voting~(\AV), Satisfaction Approval Voting~(\SAV), Chamberlin--Courant~(\CC), and Proportional Approval Voting~(\PAV).
Under \AV, the score of~$v_i$ for a committee is the number of alternatives in the committee that $v_i$ approves of. 
Under \SAV, the score of~$v_i$ for a committee is the number of alternatives in the committee that $v_i$ approves of divided by the total number of alternatives $v_i$ approves of (the score is $0$ if no alternatives are approved). 
Under \CC, the score of $v_i$~is binary: it is one if the committee contains at least one alternative in~$\app_i$, and zero otherwise.
Under \PAV, the score of $v_i$~depends on the number of approved alternatives in the committee via the harmonic sum. 

Much of the multi-winner literature evaluates rules via \emph{representativeness} and \emph{proportionality} axioms.
For example, \CC\ and \PAV\ satisfy \myemph{justified representation} (JR)~\cite{AzizBCEFW17}.
Informally, JR says that any group large enough to ``deserve'' one seat and agreeing on an alternative must contain at least one voter
who approves a selected alternative.
\PAV\ even satisfies the stronger axiom \myemph{extended justified representation} (EJR). %
Representation axioms such as JR and EJR specify \emph{what} proportionality guarantees a committee should satisfy.
In this work, we focus instead on \emph{stability}, which is strictly stronger than JR/EJR:
an entitled coalition should not be able to switch to a proportional committee that gives \emph{every} member strictly higher utility, even when the coalition's gains come from covering different approved alternatives rather than a single common one.

A standard way to formalize such coalitional stability under proportional entitlements is via the \myemph{core} for committee selection, introduced by \citet{AzizBCEFW17}.
Intuitively, a coalition~$\vvv' \subseteq \vvv$ is entitled to propose a committee of size proportional to its share. 
A size-$\kk$ committee~$\W$ is \myemph{core-stable} if there is no coalition~$\vvv'$ and no committee~$\W'$ with $|\W'|\le \share$ such that every voter in~$\vvv'$ strictly improves her \AV-score\footnote{\citet{AzizBCEFW17} defined the core using \PAV-scores; as discussed in \cref{obs:AVPAVequiv}, the resulting core coincides with the \AV-score (and \SAV-score) core.
  We use \AV-scores because they are standard in subsequent work and simpler to compute.\label{fn:core-AV-PAV}} when moving from~$\W$ to~$\W'$.
This yields a cooperative-game interpretation with non-transferable utilities (\NTU), where proportional entitlements restrict feasible deviations.
A basic open question posed by \citet{AzizBCEFW17} is whether this core is always non-empty.

In this work, we generalize this committee-core viewpoint by defining \myemph{multi-winner approval voting games} under both non-transferable~(\NTU) and transferable utilities~(\TU).
The voters are the players, and a coalition~$\vvv' \subseteq \vvv$ may claim any \myemph{admissible} committee~$\W' \subseteq \aaa$ with $|\W'|\le \share$.
Fixing a voting rule, each admissible committee induces voter utilities according to the rule, and we study the resulting cooperative-game core as a set of \emph{core-stable} utility vectors.

Under the \TU\ model, only a coalition's \emph{total} induced utility matters, and it may be redistributed arbitrarily among its members.
A utility vector is \TU-core-stable if it is induced by some size-$\kk$ committee and no coalition~$\vvv'$ can block it,
i.e., $\vvv'$ should not be able to claim an admissible committee~$\W'$ whose induced \emph{total} score for~$\vvv'$ exceeds the sum of utilities assigned to the voters in~$\vvv'$.

Under the \NTU\ model, utilities are not transferable: A coalition can only realize per-voter utilities that arise directly from some admissible committee.
A utility vector is \myemph{feasible} if it equals the per-voter score vector induced by some size-$\kk$ committee.
The \NTU-core consists of all feasible utility vectors that are not blocked by any coalition~$\vvv'$, i.e., there is no admissible committee~$\W'$ such that every voter in~$\vvv'$ obtains a strictly larger induced score from~$\W'$ than her utility.
Since every size-$\kk$ committee induces a feasible utility vector, the \NTU-core under \AV\ and \SAV\ coincides with the committee core of \citet{AzizBCEFW17}.

In many environments, proportional `entitlements' are not rigid per-voter guarantees but coalition-level resources that can be reallocated internally after the outcome is chosen.
Interpreting utility as a proxy for surplus (e.g., influence, access, funding, or agenda benefits) makes \TU\ natural: A coalition that can coordinate on a deviation can also write side agreements that compensate members who would otherwise lose from the deviation.
In this interpretation, \TU-core stability is a no-profitable-renegotiation condition subject to proportional seat caps.

To illustrate this distinction, we give a small example where winning committees need not be core-stable, and that allowing transfers (\TU) can restore stability.
\begin{example}\label{ex1}
  Let $\aaa=\{\alt_1,\alt_2,\alt_3,\alt_4,\alt_5\}$ and $\vvv=\{v_1,\dots,v_6\}$ with
  $\app_1=\app_2=\app_3=\app_4=\{\alt_1,\alt_2,\alt_3\}$ and $\app_5=\app_6=\{\alt_4,\alt_5\}$.
  Let $\kk=3$.
  Then $\W=\{\alt_1,\alt_2,\alt_3\}$ is the unique \AV-winning committee, with total \AV-score~$12$,
  inducing utilities $(3,3,3,3,0,0)$.
  However, it is not core-stable (in the \NTU\ sense): the coalition~$\{v_5,v_6\}$ can propose~$\{\alt_4\}$ (or $\{\alt_5\}$) and strictly improve both members.
  Thus every core-stable committee must contain~$\alt_4$ or $\alt_5$ and two of $\{\alt_1,\alt_2,\alt_3\}$, inducing $(2,2,2,2,1,1)$.
  In this instance, no \AV-winning committee is core-stable (under \NTU). The same reasoning leads to no \SAV-winning committee being core-stable (under \NTU), as $\W=\{\alt_1,\alt_2,\alt_3\}$ is the unique \SAV-winning committee.
  This phenomenon also occurs for \PAV; see, e.g., \citet[Section~5.2]{AzizBCEFW17}.
  For \CC, any size-$3$ committee containing at least one of $\{\alt_4,\alt_5\}$ (and hence at least one of~$\{\alt_1, \alt_2, \alt_3\}$) is winning and induces an all-$1$ vector, which is certainly in the core. 
  In fact, we will show that every JR-satisfying committee induces a core-stable utility vector; see \cref{prop:CCJR}.

  In the \TU\ model, an equal split of the overall \AV-score $12$ yields $(2,2,2,2,2,2)$, which lies in the \TU-core for \AV. Similarly, an equal split of the overall \SAV-score $4$ yields $(\frac{2}{3},\frac{2}{3},\frac{2}{3},\frac{2}{3},\frac{2}{3},\frac{2}{3})$, which lies in the \TU-core for \SAV.
  For \PAV, a \TU-core utility vector assigns $1.5$ to the first four voters and $1$ to the remaining two.
  For \CC, assigning $1$ to every voter yields a \TU-core utility function.
\end{example}

\paragraph*{Our contributions.}
We study the core of multi-winner approval voting games under both \NTU\ and \TU\ for four prominent score-based rules: \AV, \SAV, \CC, and \PAV.
Our results compare core \emph{existence}, \emph{structure}, and \emph{computability} across the two models and four rules.
\mypara{Structural and constructive results.}
\begin{inparaenum}[(1)]
  \item \textbf{\TU-\AV/\SAV:} We prove that the \TU-core is always non-empty and give a polynomial-time algorithm (cf.\ \cref{alg:TU-core-totsep}) to compute a core utility vector. 
  The algorithm exploits the additivity of these rules by decomposing the total score into a coalition-independent baseline tied to the weakest winning alternative and an alternative-specific surplus that is payable only by voters that gain utility from these alternatives. 
  We also provide an alternative non-emptiness proof via the Bondareva--Shapley theorem~\cite{bondareva1963,shapley1967}. Both results hold for a broad class of voting rules in which each alternative contributes independently to each voter's score, and the approach generalizes even to rules over linear preferences such as the multi-winner Borda rule.
  \item \textbf{\TU-\CC/\PAV:} We show that the \TU-core can be empty for \CC\ and \PAV.
  Thus, transferability alone does not guarantee stability.
  Additionally, for each $\vr \in\{\AV,\SAV,\PAV,\CC\}$, we show that the Shapley value need not yield an element in the \TU-core, even when the \TU-core is non-empty.
  \item \textbf{\NTU-\CC:} We show that \NTU-core under \CC\ is ``equivalent'' to justified representation~(JR).
  Combining this equivalence with a known greedy JR procedure (\cref{alg:NTU-core-CC}) yields a polynomial-time method to compute an \NTU-core utility vector.
  \item \textbf{\NTU-\AV/\SAV/\PAV:} We make explicit a simple correspondence between the \NTU-cores induced by \AV, \SAV,and \PAV\ (cf.\ \cref{obs:AVPAVequiv}); core non-emptiness for these games remains open though. 
\end{inparaenum}

Taken together, these results show that the two axes of our framework---\TU\ versus \NTU\ and the choice of scoring rule---interact non-trivially:
Neither axis uniformly determines core non-emptiness. \AV\ and \SAV\ are well-behaved under \TU\ but open under \NTU; \CC\ is the reverse.
\PAV\ offers no immediate guarantees in either model.
The full $2\times 4$ case distinction is thus necessary to map the stability landscape.
Moreover, the pattern is not arbitrary:
Under \TU, feasibility ties core elements to winning committees, so rule-specific optimization structure directly governs stability---\AV\ and \SAV's linearity yields universal existence, while the nonlinearities of \CC\ and \PAV\ permit emptiness. The same mechanism as for \AV\ and \SAV\ extends to any totally separable rule (cf. \cref{def:totsep}) — including rules over linear preferences such as the multi-winner Borda.
Under \NTU, core-stable outcomes need not arise from winning committees (cf.\ \cref{ex1}),
decoupling stability from the rule.
The one positive \NTU\ result---for \CC---relies  on an equivalence with JR, a property defined independently of the scoring rules. 
This suggests that the \TU\ model is the natural setting in which the choice of voting rule substantively shapes coalitional stability.

\mypara{Computational results.}
We study \CE\ and \CV\ for both \TU\ and \NTU-games, and establish tight complexity classifications in most cases (summarized in \cref{table:summary}).

\begin{table}
  \centering
  \caption{Overview of structural and computational results.
    Each cell reports whether the core is guaranteed to be non-empty and the complexity of verifying whether a given utility vector lies in the core.
    Constructive polynomial-time algorithms exist only for \TU–\AV, \TU-\SAV\ (\cref{alg:TU-core-totsep}), and \NTU–\CC\ (\cref{alg:NTU-core-CC}, via the equivalence with JR). Entries marked ``open'' correspond to the central open question of Aziz et al.~\cite{AzizBCEFW17}; by \cref{obs:AVPAVequiv}, the \NTU\ cases for \AV, \SAV, and \PAV\ are equivalent.
    [TN] refers to Theorem N, [PN] to Proposition N.
    Upper bounds marked ``\poly'' have constructive algorithms. 
    The complexity upper bounds (\DP\ and \ThetaTwoP) follow from \cref{thm:complexity-upperbounds}; see \cref{def:BeyondNP}.}

  {
    \begin{tabular}{@{}c c@{\;\;}c c c@{\,}c c c@{\;}c c c@{\;}c@{}} %
      \toprule
& \multicolumn{2}{c}{\TU-Core} & & \multicolumn{2}{c}{\TU-Core} & &  \multicolumn{2}{c}{\NTU-Core} & &  \multicolumn{2}{c}{\NTU-Core} \\
    Rule &  \multicolumn{2}{c}{Non-Emptiness} & &  \multicolumn{2}{c}{Membership} && \multicolumn{2}{c}{Non-Emptiness} & &  \multicolumn{2}{c}{Membership}\\ \midrule
      \multirow{2}{*}{\AV} & Always (\yes) & [T\ref{thm:AV-TU-Core-Non-Empty}] && \hard\  &  & & \multirow{2}{*}{\open} & \multirow{2}{*}{\open} & & \hard &   \\
 & \poly   &  [T\ref{thm:TU-AV-poly}] && \coNPc  & [T\ref{thm:TU-CV-AV}] & & & & & \DPc &   [T\ref{thm:NTU-CV-AV-PAV}]  \\\midrule
 \multirow{2}{*}{\SAV} & Always (\yes) & [T\ref{thm:AV-TU-Core-Non-Empty}] && \hard\  &  & & \multirow{2}{*}{\open} & \multirow{2}{*}{\open} & & \hard &   \\
 & \poly   &  [T\ref{thm:TU-AV-poly}] && \coNPc  & [T\ref{thm:TU-CV-SAV}] & & & & & \DPc &   [T\ref{thm:NTU-CV-AV-PAV}]  \\\midrule
       \multirow{2}{*}{\CC} & Can fail (\no)   & [P\ref{prop:PAVCCTUCNE}] && \hard  &  & & Always~(\yes) &  [T\ref{thm:CCNTUCE}] & & {\hard} &  \\
     &  \NPh, in $\ThetaTwoP$   &  [T\ref{thm:TU-CE}] && \DPc  &  [T\ref{thm:TU-CV-CC}] & &  \poly &  [T\ref{thm:CCNTUCE}] & &\NPc &  [T\ref{thm:NTU-CV-CC}]  \\\midrule
      \multirow{2}{*}{\PAV}    & Can fail (\no)  & [P\ref{prop:PAVCCTUCNE}] && \hard &  & & \multirow{2}{*}{\open} & \multirow{2}{*}{\open} & & \hard &   \\
       & \NPh, in $\ThetaTwoP$   & [T\ref{thm:TU-CE}] && \NPh, in \DP\  &  [P\ref{prop:TU-CV-PAV}] & &  & & & \DPc &  [T\ref{thm:NTU-CV-AV-PAV}]  \\
    \bottomrule
    \end{tabular}}\label{table:summary}
\end{table}

\paragraph*{Related work.}
General background on cooperative games and the core can be found in the monograph of Peleg and Sudh{\"o}lter~\cite{Peleg1983Introduction} (covering both \TU\ and \NTU) and in the book of Chalkiadakis et al.~\cite{ChalkiadakisCompAspects} (emphasizing computational aspects, primarily for \TU-games).
Our setting is an algorithmic cooperative game in which feasible deviations are constrained by coalition-specific proportional entitlements, and utilities are induced by a multi-winner scoring rule.

A broad line of work on algorithmic cooperative games studies \emph{combinatorial optimization games}~\cite{BHH2001ORgames,Deng2009COG}, where coalition values are defined by optimizing over a combinatorial structure induced by the coalition.
Deng et al.~\cite{deng1999algorithmic} and subsequent work establish core existence and complexity results for such TU games.
Our \TU\ model is similar in spirit in that coalition power is captured by an optimization objective, but differs in the feasibility constraint (committees bounded by~$\share$ rather than substructures) and in the utility source (scores induced by multi-winner rules rather than a fixed combinatorial objective).
Weighted voting games~\cite{ChalkiadakisCompAspects} constitute another classical family; Elkind et al.~\cite{Elkindwvg09} study the core there and show that non-emptiness is decidable in polynomial time, in contrast to the hardness phenomena we obtain for several of our settings.

The concept of core-stable committees under the \PAV\ rule was introduced by \citet{AzizBCEFW17}, who interpret proportional entitlements as constraints on deviations and pose non-emptiness as a basic open problem.
This stability concept corresponds to our \NTU-core under \PAV, equivalently, under \AV\ or \SAV; see \cref{obs:AVPAVequiv}. 
Subsequent work~\cite{JiangMW20,PierczynskiS22,LacknerSkownron23ABC,BrillGPSW24,Peters25} define core with respect to \AV\ since the two resulting core concepts coincide. %
This line of work has since focused on
\begin{inparaenum}
  \item approximations to core stability~\cite{JiangMW20},
  \item auditing/verification problems~\cite{Munagala22}, and
  \item identifying domains or parameter regimes where the core is guaranteed to exist and can be found efficiently~\cite{ChengJMW20,PierczynskiS22,Peters25,CharikarLRV025}.
\end{inparaenum}
Our contribution is orthogonal: rather than restricting preferences or moving to randomized outcomes, we broaden the model by allowing both \NTU\ and \TU\ interpretations of the same underlying rule-induced utilities and analyze how existence and computational properties change across rules and utility models.
In particular, we %
provide existence and algorithmic results for \AV\ and \SAV\ in the \TU\ model and for \CC\ in the \NTU\ model.
Recently, Brill et al.~\cite{BrillGPSW24} show that it is \coNPcc\ to determine whether a given committee is core-stable under \NTU-\AV. Per \cref{obs:AVPAVequiv}, this also holds for \NTU-\SAV\ and \NTU-\PAV.
We study the related but harder problem of determining \NTU-core membership, which is beyond \NP; more precisely it is \DPcc~(cf.~\cref{thm:NTU-CV-AV-PAV}).

Related core notions have also been studied in participatory budgeting (PB), which generalizes multi-winner selection by allocating a budget across projects.
Fain et al.~\cite{FainGM16} initiate the study of (approximate) core outcomes in continuous PB and relate them to market equilibria, while discrete PB variants have been studied algorithmically and from an approximation perspective~\cite{FainM018,JiangMW20,Munagala22,SongThanh25}.
Recent work shows that core non-emptiness can fail for broad classes of PB utilities~\cite{Maly25PBcoreempty}.
While PB and our setting share the key idea of proportional resource entitlements for coalitions, the feasible outcomes and utility aggregation differ: PB allocates divisible (or budgeted) resources, whereas we study discrete committees and utilities induced by multi-winner scoring rules.

Beyond cooperative deviations by voter coalitions, other game-theoretic models for multi-winner selection treat strategic behavior of candidates or projects.
Obraztsova et al.~\cite{ObraztsovaPEG20} introduce multi-winner candidacy games, and Faliszewski et al.~\cite{FaliszewskiPBgames} study project submission games; both model equilibrium outcomes of strategic entry rather than coalition deviations in a cooperative game.
Haret et al.~\cite{HaretK0S24}  propose a non-cooperative budgeting game in which outcomes arise from budget-allocation dynamics and prove core non-emptiness for restricted (tree-structured) domains. In contrast,  we directly formulate \TU\ and \NTU\ cooperative games over committees, establishing, among others, non-emptiness of the \TU-\AV/\SAV-core and \NTU-\CC-core for unrestricted domains.

\paragraph*{Paper outline.}

In \cref{sec:prelim}, we recall definitions from multi-winner approval voting.
There, we also introduce our multi-winner voting games in \TU\ and \NTU\ and their related computational problems.
In \cref{sec:structural}, we present our main structural results.
In \cref{sec:complexity}, we establish complexity classifications for \CE\ and \CV. 
We conclude in \cref{sec:conclude} with future research directions.

\section{Preliminaries}\label{sec:prelim}%
Given a non-negative integer~$t\in \mathds{N}$, let $[t]$ denote the set~$\{1,\ldots, t\}$.

\subsection{Multi-winner approval voting}
An \myemph{approval preference profile} (in short \myemph{profile}) is a triple~$\ppp=(\aaa, \vvv, \RR)$,
where $\aaa$ denotes a set~$\aaa$ of $m$ alternatives with $\aaa=\{\alt_1,\dots,\alt_m\}$,
$\vvv$ a set of $n$ voters with $\vvv=\{v_1,\dots, v_n\}$,
and $\RR$ a collection~$\RR=(\app_1,\dots, \app_n)$ of subsets (\myemph{approval ballots}) of $\aaa$,
with $\app_i$ consisting of the alternatives that voter~$v_i$ approves of, $i\in[n]$. We also write $\app(v)$ if the index of the voter is omitted.

For each alternative~$\alt_j\in \aaa$, we use \myemph{$\VV(\alt_j)$} to denote the set of voters who approve of~$\alt_j$.
A subset~$\vvv'\subseteq \vvv$ of voters is called a \myemph{coalition}.

An instance of \myemph{multi-winner approval voting} (in short \mw) is a pair~$I=(\ppp,\kk)$,
where $\ppp$ denotes the approval preference profile and $\kk\in \mathds{N}$ denotes the number of alternatives from~$\aaa$ that shall be selected.
The task is to select a \myemph{committee} (i.e., a subset of alternatives) of $\kk$ alternatives from~$\aaa$.
To complete this task, we use voting rules that are based on scores. 
Roughly speaking, a score-based voting rule~$\vr$ assigns to each committee~$\W$ a non-negative value, called \myemph{score}, which measures the satisfaction of the voters towards~$\W$.
Then, $\vr$ selects a size-$\kk$ committee among all size-$\kk$ committees that maximizes the score. 

We consider four typical score-based voting rules: Approval Voting~(\AV), Satisfaction Approval Voting~(\SAV), Chamberlin--Courant~(\CC),
and Proportional Approval Voting~(\PAV).
For each~$\vr\in \{\AV,\SAV,\CC,\PAV\}$, the score of a committee is the sum of the voters' individual scores.
Formally, given a coalition $\vvv'\subseteq \vvv$ and a committee $\W\subseteq \aaa$, define:
\begin{description}
  \item[Approval Voting (\AV):]
  \myemph{$\score{\AV}(\vvv',\W)$} $= \sum_{v_i\in \vvv'} |\app_i\cap \W|$.
    \item[Satisfaction Approval Voting (\SAV):]
  \myemph{$\score{\SAV}(\vvv',\W)$} $= \begin{cases}
  	\sum_{v_i\in \vvv'} \frac{|\app_i\cap \W|}{|\app_i|} & \app_i\neq\emptyset \\
  	0 & \, \text{else}
  \end{cases} $.
  \item[Chamberlin--Courant (\CC):]
  \myemph{$\score{\CC}(\vvv',\W)$} $= \sum_{v_i\in \vvv'} \min\bigl(|\app_i\cap \W|,1\bigr)$.
  \item[Proportional Approval Voting (\PAV):]
  \myemph{$\score{\PAV}(\vvv',\W)$} $= \sum_{v_i\in \vvv'} \sum_{z=1}^{|\app_i\cap \W|}\frac{1}{z}$.
\end{description}
For readability, we write $\score{\vr}(v_i,\W)$ instead of $\score{\vr}(\vvv',\W)$ when $\vvv'=\{v_i\}$.

Finally, for each voter coalition~$\vvv'\subseteq \vvv$, 
we call a committee~$\W'\subseteq \aaa$ \myemph{share-admissible} (in short \myemph{admissible}) for~$\vvv'$ if $|\W'|$ is at most proportional to the share of $\vvv'$, i.e., $|\W'| \le \share$. 
Then, let 
\begin{align}\label{def:score-V'}
  \score{\vr}(\vvv')
  \;\coloneqq\;
  \max_{\W'\subseteq \aaa\;:\; |\W'|\le \share}
  \score{\vr}(\vvv',\W')
\end{align}
denote the maximum score that an admissible committee can achieve for~$\vvv'$.
Moreover, let \myemph{$\vr(\ppp,\kk)$} $ = \displaystyle\argmax_{\W\subseteq \aaa\;:\; |\W|=\kk} \score{\vr}(\vvv,\W)$
denote the set of all winning committees under rule $\vr$.

For more details on multi-winner approval voting we refer to a recent book by Lackner and Skowron~\cite{LacknerSkownron23ABC}.

\subsection{Multi-winner voting games}
Let $\vr$ denote a score-based voting rule.
A \myemph{$\vr$-multi-winner voting game} consists of an \mw-instance $(\ppp,\kk)$ with voter set $\vvv$ in $\ppp$,
where the voters are the players.
The utility of each voter coalition is derived from the maximum score that an admissible committee can achieve for it under $\vr$.
The goal is to find a \myemph{utility function~$\alloc \colon \vvv\to \mathds{R}_{\geq0}$} (aka.\ \myemph{utility vector})
that assigns to each voter~$v_i$ a utility value~$\alloc(v_i)$ such that no subset of voters has an incentive to deviate from the solution.

We propose two types of multi-winner voting games: one with \myemph{transferable utilities} (\TU) and the other with
\myemph{non-transferable utilities}~(\NTU).
In an \NTU-game, a coalition can only realize utilities that arise from some committee (each voter receives her own score).
In a \TU-game, only the total coalition utility matters and can be redistributed arbitrarily among coalition members.

\begin{definition}[\MWGame\ in \TU\ and \NTU]\label{defi:games}
  The input of an \MWGame\ is an \mw-instance~$(\ppp,\kk)$.

  \mypara{The \TU-game} induced by $(\ppp,\kk)$ and $\vr$ is the characteristic-function game $(\vvv,\cvalvr)$ where,
  for every coalition~$\vvv'\subseteq \vvv$,
  \begin{align*} %
    \myemph{\cvalvr(\vvv')}
    \;\coloneq\; %
    \max_{\W'\subseteq \aaa\;:\; |\W'|\le \share}
    \score{\vr}(\vvv',\W'); \text{ note that }\cvalvr(\vvv') = \score{\vr}(\vvv') \text{ as defined in } \eqref{def:score-V'}. 
  \end{align*}

  \mypara{The \NTU-game} induced by~$(\ppp,\kk)$ and $\vr$ consists of for each coalition~$\vvv'\subseteq \vvv$,
  the set of \myemph{feasible} utility functions with 
  \[
    \myemph{\NTUutil{\vr}(\vvv')}
    \coloneq
    \Big\{
    \alloc'\colon \vvv'\to \mathds{R}_{\ge 0}\ \Big|\
    \exists \W'\subseteq \aaa \text{ with }|\W'|\le \share
    \text{ s.t. } \alloc'(v)=\score{\vr}(v,\W')\ \forall v\in \vvv'
    \Big\}.
  \]

\end{definition}

\begin{remark}[TU games as a subclass of NTU games]\label{remark:TU-sub-NTU}
  It is well-known that \NTU-games generalize \TU-games~\cite{ChalkiadakisCompAspects}.
  In particular, every \TU-game $(\vvv,\cvalvr)$ can be represented as an \NTU-game by defining, for each coalition $\vvv'\subseteq \vvv$,
  the feasible set
  \[
    \TUutil{\vr}(\vvv') \;=\; \Big\{\alloc'\colon \vvv'\to\mathds{R}_{\ge 0}\ \Big|\ \sum_{v\in \vvv'} \alloc'(v)\le \cvalvr(\vvv')\Big\}.
  \]
\end{remark}

Next, we define blocking coalitions and cores for the two games. 
\begin{definition}[Blocking, witness, and core]\label{defi:blocking-core}
  Let~$(\ppp,\kk)$ be an \mw-instance with $\ppp=(\aaa,\vvv)$,
  $\alloc\colon \vvv\to \mathds{R}_{\ge 0}$ a utility function,
  and $\vvv'\subseteq \vvv$ a voter coalition.

  $\vvv'$ is called \myemph{$\TUblocking{\vr}$} for $\alloc$ if
  \begin{align}\label{eq:TU-blocking-sum}
    \sum_{v\in \vvv'} \alloc(v) \;<\; \cvalvr(\vvv').
  \end{align}

  $\vvv'$ is called \myemph{$\NTUblocking{\vr}$} for $\alloc$ if there exists
  a utility function $\alloc'\in \NTUutil{\vr}(\vvv')$ such that
  \begin{align}\label{eq:blocking-componentwise}
    \alloc(v) < \alloc'(v)\quad \text{for all } v\in \vvv'.
  \end{align}
  \noindent  A \myemph{blocking witness} for $\vvv'$ (against $\alloc$) is additional data that certifies that $\vvv'$ is blocking.
  \begin{compactitem}[--]
    \item 
    {A committee~$\W'$ is a \myemph{witness} for $\vvv'$ to be~\myemph{$\TUblocking{\vr}$} for~$\alloc$ if
      $|\W'|\le \share$ and
      $\sum_{v\in \vvv'}\alloc(v)<\score{\vr}(\vvv',\W')$. Note that this implies that~\eqref{eq:TU-blocking-sum} holds.}
    \item
    {A committee~$\W'$ is a \myemph{witness} for $\vvv'$  to be~\myemph{$\NTUblocking{\vr}$} for~$\alloc$  if
      $|\W'|\le \share$ and 
      $\alloc(v) < \score{\vr}(v,\W')$ for all $v\in \vvv'$. Note that this implies that~\eqref{eq:blocking-componentwise} holds.}
  \end{compactitem}

  We say that $\alloc$ is \myemph{$\TUcs{\vr}$} if $\sum_{v\in \vvv}\alloc(v)=\cvalvr(\vvv)$, and
  no coalition is $\TUblocking{\vr}$ for~$\alloc$.
  It is \myemph{$\NTUcs{\vr}$} if $\alloc \in \NTUutil{\vr}(\vvv)$, and no coalition is $\NTUblocking{\vr}$ for~$\alloc$.

  The \myemph{$\NTUcore{\vr}$} (resp.\ $\TUcore{\vr}$) is the set of all {$\NTUcs{\vr}$} (resp.\ {$\TUcs{\vr}$}) utility functions.
\end{definition}

\begin{remark}
  The \NTU-game of \cref{defi:games} should not be confused with the \NTU\ representation of the \TU-game (\cref{remark:TU-sub-NTU}).
  In the \TU-representation, any non-negative utility vector summing to at most $\cvalvr(\vvv')$ is feasible for a coalition~$\vvv'$,
  whereas in the \NTU-game only utility vectors directly realized by some admissible committee are feasible.
  In particular, $\NTUutil{\vr}(\vvv') \subset \TUutil{\vr}(\vvv')$ in general, and core non-emptiness in one model does not imply it in the other.
\end{remark}
  
We illustrate the \TU\ and \NTU-cores on two small instances (with $k=2$) for rules $\vr\in\{\AV,\SAV,\PAV,\CC\}$; here we view the utility function $\alloc\colon \vvv\to\mathds{R}_{\ge 0}$ as a vector.
\begin{example}[\TU]\label{ex:TUgames}
  Let $\aaa=\{\alt_1,\alt_2,\alt_3\}$, $\vvv=\{v_1,v_2,v_3,v_4\}$, and
  \[
    \app_1=\app_2=\{\alt_1\},\ 
    \app_3=\{\alt_1, \alt_2\},\ 
    \app_4=\{\alt_2, \alt_3\}
  \]
  with $\kk=2$.
  In the \TU\ model, $\alloc$ is core-stable if and only if it is feasible
  $\bigl(\sum_{i=1}^4 \alloc(v_i)=\max_{|\W|=2}\scoreTU{\vr}(\vvv,\W)\bigr)$
  and satisfies the coalition inequalities
  $\sum_{v_i\in \vvv'}\alloc_i \ge \max_{|\W'|\le \lfloor 2|\vvv'|/4\rfloor}\score{\vr}(\vvv',\W')$ for all $\vvv'\subseteq\vvv$.
  Since singleton coalitions have seat cap~$0$, the binding constraints come from
  size-$2$ and -$3$ coalitions (seat cap~$1$).
  For \AV, \CC, and \PAV, these reduce to the same set of coalition inequalities on
  $(\alloc(v_i))_{i\in[4]}$; only the feasibility for the grand coalition differs. For \SAV\ we get a separate set of inequalities.
  The coalition inequalities for \AV, \CC, and \PAV\ are
  \begin{alignat*}{3}
\alloc(v_i) \geq~ 0, \text{ for all } i \in [4],\quad  
\alloc(v_i)+\alloc(v_j)\geq~& 2,  \text{ for all } i, j \in [3] \text{ with } i \neq j, \nonumber\\
\alloc(v_1)+\alloc(v_4) \geq 1, ~~ \alloc(v_2)+\alloc(v_4) \geq~ 1,~~ \alloc(v_3)+\alloc(v_4) \geq~& 2, \quad
\alloc(v_1)+\alloc(v_2)+\alloc(v_3)\geq3.
\end{alignat*}

\AV: $\W_1=\{\alt_1,\alt_2\}$ is the unique winning committee with score~$5$,
so $\sum_{i\in [4]}\alloc_i = 5$. 
The $\TUcore{\AV}$ is
$\alloc=(1+\epsilon_1,\,1+\epsilon_2,\,1+\epsilon_3,\,2-(\epsilon_1+\epsilon_2+\epsilon_3))$
where
$0\leq\epsilon_1+\epsilon_2+\epsilon_3\leq 2$,\;
$0\leq\epsilon_1+\epsilon_2\leq 1$,\;
$0\leq\epsilon_1+\epsilon_3\leq 2$,\;
$0\leq\epsilon_2+\epsilon_3\leq 2$.

\CC: The winning committees are $\W_1$ and $\W_2=\{\alt_1,\alt_3\}$, each with
score~$4$, so $\sum_{i\in [4]}\alloc_i = 4$.
The $\TUcore{\CC}$ is
$\alloc=(1+\epsilon_1,\,1+\epsilon_2,\,1+\epsilon_3,\,1-(\epsilon_1+\epsilon_2+\epsilon_3))$
where
$0\leq\epsilon_1+\epsilon_2+\epsilon_3\leq 1$,\;
$\epsilon_1+\epsilon_2=0$,\;
$0\leq\epsilon_1+\epsilon_3\leq 1$,\;
$0\leq\epsilon_2+\epsilon_3\leq 1$.

\PAV: $\W_1$ is the unique winning committee with score~$4.5$, so $\sum_{i\in [4]}\alloc_i = 4.5$.
The $\TUcore{\PAV}$ is
$\alloc=(1+\epsilon_1,\,1+\epsilon_2,\,1+\epsilon_3,\,1.5-(\epsilon_1+\epsilon_2+\epsilon_3))$
where
$0\leq\epsilon_1+\epsilon_2+\epsilon_3\leq 1.5$,\;
$0\leq\epsilon_1+\epsilon_2\leq 0.5$,\;
$0\leq\epsilon_1+\epsilon_3\leq 1.5$,\;
$0\leq\epsilon_2+\epsilon_3\leq 1.5$.

For \SAV, we get the following coalition inequalities:
  \begin{alignat*}{3}
	\alloc(v_i) \geq~& 0, \text{ for all } i \in [4],\\  
	\alloc(v_i)+\alloc(v_3)\geq~& 1.5,  \text{ for all } i \in [2], \nonumber\\
	\alloc(v_i)+\alloc(v_4)\geq~& 1,  \text{ for all } i \in [3], \nonumber\\
	\alloc(v_1)+\alloc(v_2) \geq~& 2, \quad
	\alloc(v_1)+\alloc(v_2)+\alloc(v_3)\geq2.5.
\end{alignat*}

\SAV: $\W_1$ is the unique winning committee with score~$3.5$, so $\sum_{i\in [4]}\alloc_i = 3.5$. The $\TUcore{\SAV}$ is
$\alloc=(1+\epsilon_1,\,1+\epsilon_2,\,1+\epsilon_3,\,0.5-(\epsilon_1+\epsilon_2+\epsilon_3))$
where
$-0.5\leq\epsilon_1+\epsilon_2+\epsilon_3\leq 0.5$,\;
$0\leq\epsilon_1+\epsilon_2\leq 0.5$,\;
$-0.5\leq\epsilon_1+\epsilon_3\leq 0.5$,\;
$-0.5\leq\epsilon_2+\epsilon_3\leq 0.5$.
\end{example}

\begin{example}[\NTU]\label{ex:NTUgames}
Let $\aaa=\{\alt_1,\alt_2,\alt_3,\alt_4\}$, $\vvv=\{v_1,\ldots,v_6\}$, and
\[
\app_1=\{\alt_1\},\quad
\app_2=\app_3=\{\alt_1,\alt_4\},\quad
\app_4=\{\alt_2\},\quad
\app_5=\app_6=\{\alt_2,\alt_3,\alt_4\}.
\]
Let $\kk=2$.
Coalitions of size at most~$2$ have seat cap~$\lfloor 2\cdot 2/6\rfloor = 0$ and cannot block.
Coalitions of size $3$ to $5$ may select a single alternative.

\mypara{Necessary conditions.}
If $\W\cap\{\alt_1,\alt_4\}=\emptyset$, then $v_1,v_2,v_3$ each receive utility~$0$
and the coalition $\{v_1,v_2,v_3\}$ blocks via $\{\alt_1\}$.
Symmetrically, $\W\cap\{\alt_2,\alt_3,\alt_4\}=\emptyset$ lets $\{v_4,v_5,v_6\}$ block
via~$\{\alt_2\}$.
Hence every core-stable committee must intersect both
$\{\alt_1,\alt_4\}$ and $\{\alt_2,\alt_3,\alt_4\}$.

\mypara{Sufficient conditions.}
Suppose $\W$ meets both conditions.
A coalition of size $3$--$5$ deviating to $\{\alt_j\}$ requires every member
to currently have utility~$0$ and to approve~$\alt_j$.
For $\alt_1$: Voters $v_2,v_3$ also approve $\alt_4$, so at most $v_1$ has utility~$0$.
For $\alt_2$: Voters $v_5,v_6$ also approve $\alt_3,\alt_4$, so at most $v_4$ has utility~$0$.
For $\alt_3$: Only two approvers exist.
For $\alt_4$: Both conditions together ensure at least four voters ($v_2,v_3$ via the first, $v_5,v_6$ via the second) already have utility $\geq 1$.
In every case fewer than three zero-utility approvers exist.
The grand coalition (with seat cap~$2$) cannot block either:
strict improvement for $v_1$ requires %
$\alt_1 \notin \W$ and for $v_4$ requires
$\alt_2\notin \W$, as otherwise $v_1$ and $v_4$ cannot improve.
This gives us $\W=\{\alt_3,\alt_4\}$; but then voters $v_5,v_6$ who already have
utility~$2$ are not strictly improved.

Therefore the $\NTUcore{\AV}$ consists of the utility vectors induced by
$\{\alt_1,\alt_2\},\ \{\alt_1,\alt_3\},\ \{\alt_1,\alt_4\},\ \{\alt_2,\alt_4\},\ \{\alt_3,\alt_4\}$.
Since any blocking deviation selects a single alternative, the same five
committees form the \NTU-core under \SAV, \CC, and~\PAV.

\end{example}

\subsection{Central problems}\label{subsec:problems}
In this paper, we study two standard computational problems for cores: core non-emptiness and core membership verification.
Let $\vr$ denote a score-based multi-winner voting rule.

\decprob{\MWGTUCexist{$\vr$} (\text{resp.} \MWGNTUCexist{$\vr$})}
{An \mw-instance~$(\ppp,\kk)$.}
{Is the $\TUcore{\vr}$ (resp.\ $\NTUcore{\vr}$) of $(\ppp,\kk)$ non-empty?}

\decprob{\MWGTUCverif{$\vr$} (\text{resp.} \MWGNTUCverif{$\vr$})}
{An \mw-instance~$(\ppp,\kk)$ and a utility function~$\alloc$}
{Is $\alloc$ in the $\TUcore{\vr}$ (resp.\ $\NTUcore{\vr}$) of $(\ppp,\kk)$?}

\section{Non-Emptiness and Structure of the Core}\label{sec:structural}%
We now turn to our main structural results, organized by utility model.

\subsection{\TU: Existence and construction}\label{sub:structure-TU}
We start with the \TU\ interpretation, where a coalition's induced total score can be redistributed arbitrarily among its members.
Our first result shows that, under \AV\ and \SAV, the resulting \TU-game always has a non-empty core. We are actually able to show an even stronger result: For any voting rule where the total committee score is the sum of individual voter scores and each alternative's contribution is independent of which others are selected, the resulting \TU-game always has a non-empty core.
We note that the \TU-game under these voting rules is not convex since the \TU characteristic function is not submodular (cf.\ \cref{ex:TUgames} for a non-submodular instance),
so core non-emptiness does not follow from standard sufficient conditions for convex games~\cite{Shapley1971}.
Instead, we require a direct argument.
We present two complementary perspectives.
The Bondareva–Shapley argument~\cite{bondareva1963,shapley1967} provides a conceptual certificate of non-emptiness via the fractional-committee relaxation, illuminating why the core is always non-empty: Linearity of these voting rules ensures that fractional committees cannot outperform integral ones.
For these rules direct combinatorial construction (\cref{alg:TU-core-totsep}) additionally provides an explicit core element in polynomial time, which the Bondareva–Shapley proof alone does not guarantee.
We then show that this positive picture is specific to \AV\ and \SAV: For \CC\ and \PAV\ the \TU-core can be empty, and even when the core exists, canonical solution concepts such as the Shapley value need not produce a core element.

We next recall the Bondareva--Shapley Theorem~\cite{bondareva1963,shapley1967}.
\begin{proposition}[Bondareva-Shapley Theorem\footnote{The Bondareva-Shapley Theorem is often stated with weights indexed only by non-empty
    coalitions, i.e., $\wvector=(\weight_S)_{\emptyset\neq S\subseteq \vvv}$.
    We include $\weight_\emptyset$ purely for notational convenience (so we can write $\sum_{S\subseteq N}$), and this is without loss since $\nu(\emptyset)=0$ in our setting and $\emptyset$ never appears in the balancing constraints~\eqref{eq:BS_balanced}.}~\cite{bondareva1963,shapley1967}]\label{prop:BS}
  Let $(\vvv,\chi)$ be a TU-game. The core of $(\vvv, \chi)$ is non-empty if and only if
  for  every weight vector~$\wvector=(\weight_{S})_{S\subseteq \vvv}$
  with $\weight_{S}\ge 0$, $S\subseteq \vvv$, satisfying
  \vspace{-2ex}
  \begin{align}\label{eq:BS_balanced}
    \sum_{S\subseteq \vvv\colon v_i\in S} \weight_{S} = 1
    \qquad\text{for all } v_i\in \vvv,
  \end{align}
  \vspace{-1ex}
  \begin{align}\label{eq:BS_ineq}
    \text{  the following inequality holds: }
     \sum_{S\subseteq \vvv} \weight_{S}\, \chi(S) \le \chi(\vvv). & \qquad \qquad \qquad \qquad
  \end{align}
\end{proposition}

\noindent In the following, we call a weight vector~$\wvector$ \myemph{balanced} if it is non-negative %
and satisfies~\eqref{eq:BS_balanced}.
To show core non-emptiness, we will use \cref{prop:BS} above and show that no balanced weight vector can collectively obtain more value than the grand coalition~$\vvv$.
The crucial idea is that every balanced vector yields a fractional committee,
and the left-hand side of~\eqref{eq:BS_ineq} can be upper-bounded by the optimum of a linear program over fractional committees.
This linear program is maximized by selecting the top-$\kk$ integral committee, which corresponds to a winning committee under the considered rules. 
First we define two terms that characterize the voting rules for which we show that the core is always non-empty.
\begin{definition}\label{def:totsep}
We say a score-based voting rule $\vr$ is \myemph{\votsep} if $\score{\vr}(\vvv',\W)=\sum_{v_i\in\vvv'}\score{\vr}(v_i,\W)$ for all $\vvv'\subseteq\vvv$ and $\W\subseteq\aaa$. 
We say a score-based voting rule $\vr$ is \myemph{\altsep} if $\score{\vr}(\vvv',\W)=\sum_{\alt_i\in\W}\score{\vr}(\vvv',\{\alt_i\})$ for all $\vvv'\subseteq\vvv$ and $\W\subseteq\aaa$. 
We say a score-based voting rule $\vr$ is \myemph{\totsep} if it is \votsep\ and \altsep.
\end{definition}
We can immediately observe that both \AV\ and \SAV\ are \totsep.
\begin{observation}\label{obs:AVSAVtotsep}
\AV\ and \SAV\ are \totsep.
\end{observation}
\begin{restatable}{theorem}{thmAVTUCoreNonEmpty}\label{thm:AV-TU-Core-Non-Empty}
$\TUcore{\vr}$ is always non-empty for all \totsep\ voting rules $\vr$ that satisfy $\scoreTU{\vr}(v,\{\alt\})\geq0$ for all $v\in\vvv$ and $\alt\in\aaa$. Therefore, $\TUcore{\AV}$ and $\TUcore{\SAV}$ are always non-empty.
\end{restatable}
  \begin{proof}
    Let $(\ppp,\kk)$ be an \mw-instance with $\ppp=(\aaa,\vvv)$. Let $\vr$ be a \totsep\ voting rule.
  For every voter coalition~$S\subseteq \vvv$, let
  \begin{align*}
  \chi(S)\coloneq \cval{\vr}(S) = \max_{\W'\subseteq \aaa\colon |\W'|\le \shareS}\score{\vr}(S,\W')\text{, for }\vr\in\{\AV,\SAV\}.
  \end{align*}
  Note that for $\vr=\AV$, $\chi(S)=\max_{\W'\subseteq \aaa\colon |\W'|\le \shareS}\sum_{v_i\in S}|\app_i\cap \W'|$, and for $\vr=\SAV$, $\chi(S)=\max_{\W'\subseteq \aaa\colon |\W'|\le \shareS}\sum_{v_i\in S}\frac{|\app_i\cap \W'|}{|\app_i|}$.
  Then, by \cref{prop:BS}, it suffices to show that all \emph{balanced} weight vectors satisfy inequality~\eqref{eq:BS_ineq}. %
  Let $\wvector=(\weight_{S})_{S\subseteq \vvv}$ be an arbitrary balanced weight vector satisfying~\eqref{eq:BS_balanced}. %
  For each voter coalition~$S\subseteq \vvv$, choose a committee~\myemph{$\W_{S}$} that maximizes the score for~$S$, i.e.,
  \begin{align}
    \chi(S) = \score{\vr}(S,\W_S) \text{ and } |\W_S| \le \shareS.\label{def:W_s}
  \end{align}
  Note that $\W_S$ depends on the voting rule, e.g., $\W_S$ can differ based on whether \AV, \SAV, or some other rule is the underlying voting rule.  
  Then, we claim that the weighted sum of the sizes of the committees~$\W_{S}$ is bounded by~$\kk$.
  \begin{restatable}[]{claim}{claimweightedcommittees}\label{claim:weighted-committees}
    $\displaystyle\sum_{S\subseteq \vvv}\weight_S |\W_S|\leq \kk$.
  \end{restatable}
\begin{claimproof}{claim:weighted-committees}
    Using \eqref{eq:BS_balanced} we obtain the incidence-counting identity
    \begin{align}\label{eq:sum_weighted_sizes}
      \sum_{S\subseteq \vvv}\weight_S|S|
      =
      \sum_{S\subseteq V}\weight_S\sum_{v_i\in S}1
      =
      \sum_{v_i\in \vvv}\sum_{S\subseteq \vvv\colon v_i\in S}\weight_S
      \stackrel{\eqref{eq:BS_balanced}}{=}
      \sum_{v_i\in \vvv}1
      =
      |\vvv|.
    \end{align}
    Combining~\eqref{def:W_s} with \eqref{eq:sum_weighted_sizes} gives
    \begin{align*}
      \sum_{S\subseteq \vvv}\weight_S|\W_S| \le  \sum_{S\subseteq \vvv}\weight_S \lfloor\frac{\kk}{|\vvv|}|S|\rfloor \le  \frac{\kk}{|\vvv|} \sum_{S\subseteq \vvv} \weight_{S}|S|  = k, \text{as desired.}
    \end{align*}
\end{claimproof}
  Next, we show that the weighted sum of the scores of the committees~$\W_S$ is upper-bounded by the weighted sum of the scores of the associated fractional committee~$\xx$, where
  for each~$\alt_j\in \aaa$, let \myemph{$\xx(\alt_j)$} $\displaystyle\coloneqq \min\Bigl\{1,\sum_{S\subseteq \vvv\colon \alt_j\in \W_S}\weight_S\Bigr\}$.
  \begin{claim} %
    \label{claim:fractional-committee}
    \begin{compactenum}
      \item\label{fractional-k}  $\sum_{\alt_j\in \aaa} \xx(\alt_j) \le \kk$.
      \item\label{fractional-sum} $\displaystyle\sum_{S\subseteq \vvv} \weight_{S}\cdot \score{\vr}(S,\W_S) \le \sum_{\alt_j\in \aaa} \xx(\alt_j)\cdot\score{\vr}(\vvv,\{\alt_j\})$, where \myemph{$\xx(\alt_j)$} $\displaystyle\coloneqq \min\Bigl\{1,\ \sum_{S\subseteq \vvv\colon \alt_j\in \W_S}\weight_S\Bigr\}$.
    \end{compactenum}
    \end{claim}
    \begin{claimproof}{claim:fractional-committee}
    The first statement means that $\xx$ as defined above yields a fractional committee of size at most~$\kk$.
    By the definition of~$\xx$ and \cref{claim:weighted-committees}, %
    we obtain \begin{align*}
      \sum_{\alt_j\in \aaa} \xx(\alt_j) \stackrel{\text{def}}{\le} \sum_{\alt_j\in \aaa} \sum_{\stackrel{S\subseteq\vvv\colon}{\alt_j\in \W_S}}\weight_S
      = \sum_{
      S\subseteq \vvv} \sum_{\alt_j\in \W_S}\weight_S
      = \sum_{S\subseteq \vvv}\weight_S|\W_S| \stackrel{\small\text{\cref{claim:weighted-committees}}}{\le} \kk, \text{ showing the first statement.}
    \end{align*}
    Then, by the definition of the score under \totsep\ voting rules and by counting the summation differently, we obtain 
    \begin{align}
      \sum_{S\subseteq \vvv} \weight_{S}\cdot \score{\vr}(S,\W_S)
      = \sum_{S\subseteq \vvv} \weight_{S} \sum_{v_i\in S}\sum_{\alt_j\in \W_S}\score{\vr}(v_i, \{\alt_j\})
      = \sum_{\alt_j\in \aaa} \sum_{v_i\in \vvv} \score{\vr}(v_i, \{\alt_j\})\sum_{\mathclap{\stackrel{S\subseteq \vvv\colon}{v_i\in S\wedge \alt_j\in \W_S}}} \weight_S.\label{eq:sum-score-weights}
    \end{align}
    For each alternative~$\alt_j\in \aaa$ and voter~$v_i\in \vvv$, %
    we infer the following.
    
    \noindent $\displaystyle\sum_{\mathclap{\stackrel{S\subseteq \vvv\colon}{v_i\in S\wedge \alt_j\in \W_S}}} \weight_S \le \sum_{S\subseteq \vvv\colon v_i\in S} \weight_S\stackrel{\eqref{eq:BS_balanced}}{=} 1$,
    and
    $\displaystyle\sum_{\mathclap{\stackrel{S\subseteq \vvv\colon}{v_i\in S\wedge \alt_j\in \W_S}}} \weight_S
    \le \sum_{\mathclap{\substack{S\subseteq \vvv\colon\\ \alt_j\in \W_{S}}}}\weight_S$.~
    Hence,
      $\displaystyle\sum_{\mathclap{\stackrel{S\subseteq \vvv\colon}{v_i\in S\wedge \alt_j\in \W_S}}} \weight_S \le \min(1, \sum_{S\subseteq \vvv\colon \alt_j \in \W_S}\weight_S)$.
    Plugging the inequality into \eqref{eq:sum-score-weights},
    we obtain
    \begin{align*}
      \sum_{S\subseteq \vvv} \weight_{S}\cdot \score{\vr}(S,\W_S) \le& \sum_{\alt_j\in \aaa} \sum_{v_i\in \vvv}  \score{\vr}(v_i, \{\alt_j\})\min(1, \sum_{S\subseteq \vvv\colon \alt_j \in \W_S}\weight_S) \\
      = & \sum_{\alt_j\in \aaa}\sum_{v_i\in \vvv} \score{\vr}(v_i, \{\alt_j\})\xx(\alt_j) =\sum_{\alt_j\in \aaa} \xx(\alt_j)\cdot\score{\vr}(\vvv,\{\alt_j\})
    \end{align*}
  \end{claimproof} 
  Finally, consider the following linear program

  {\centering
    $\max\Bigl\{\sum_{\alt_j\in \aaa} y(\alt_j) \cdot \score{\vr}(\vvv,\{\alt_j\})\colon 0\le y(\alt_{j}) \le 1, \forall \alt_j\in \aaa \wedge \sum_{\alt_j\in \aaa} y(\alt_j)\le \kk\Bigl\}$.\par
  }
  
  Since  $\score{\vr}(\vvv,\{\alt_j\})\ge 0$ holds for all~$\alt_j\in \aaa$, the linear program is maximized by setting~$y(\alt_{j})=1$ for the $\kk$ alternatives with the largest score~$\score{\vr}(\vvv,\{\alt_j\})$ and $y(\alt_{j})=0$ otherwise.
  Since $\xx$ is a feasible solution of the linear program (see \cref{claim:fractional-committee}\eqref{fractional-k}), 
  we infer

  \smallskip
  {\centering %
    $\sum_{\alt_j\in \aaa} \xx(\alt_j) \cdot \score{\vr}(\vvv,\{\alt_j\}) \le \max_{\W\subseteq \aaa\colon |\W|=\kk}\sum_{\alt_j\in \W}\score{\vr}(\vvv,\{\alt_j\}) = \cval{\vr}(\vvv)$. 
    \par}
  \smallskip
    
  Combining the above with \cref{claim:fractional-committee}\eqref{fractional-sum}
    gives
    
   \smallskip
  {\centering %
    $\sum_{S\subseteq \vvv} \weight_{S} \chi(S) = \sum_{S\subseteq \vvv} \weight_{S}\cdot \score{\vr}(S,\W_S) \le  \sum_{\alt_j\in \aaa} \xx(\alt_j)\score{\vr}(\vvv,\{\alt_j\}) \le  \cval{\vr}(\vvv)$.
    \par}
    \smallskip
    
  \noindent Thus \eqref{eq:BS_ineq} holds for every balanced weight vector, and \cref{prop:BS} implies $\TUcore{\vr}\neq \emptyset$. %
  \end{proof}

The LP-based algorithm derived from the core non-emptiness proof of \cref{thm:AV-TU-Core-Non-Empty} may run in exponential time. 
Next, we explicitly give a utility function in polynomial time and show that it is in the core; see \cref{alg:TU-core-totsep}.
Despite its compact form, our algorithm must satisfy core constraints for all coalitions and all admissible deviations.
The baseline-plus-surplus split is designed so that any coalition's attainable score under its seat cap is upper-bounded by the utilities it is assigned. 
More precisely, our algorithm exploits the linearity of \totsep\ voting rules: starting from a winning committee~$\W$, we distribute
\begin{compactenum}
  \item a common baseline based on the weakest alternative in~$\W$ to all voters (see line~\ref{baseline1}), and
  \item the remaining surplus of each alternative only among its approving voters (see lines~\ref{surplus1}--\ref{surplus2});
\end{compactenum}
we prove that no coalition can block the resulting allocation.

Intuitively, the $\kk^{\text{th}}$ best alternative~$\alt_{\kk}$ serves as a ``price unit'': any coalition with seat cap $s$ can extract at most $s|\VV(\alt_{\kk})|$ baseline approvals.
Hence, any improvement beyond this baseline must come from selecting alternatives that beat $\alt_{\kk}$, whose excess approvals we pay only to voters who can actually realize them (the approvers).
    
\begin{algorithm}[t!] 
	\caption{Compute a $\TUcs{\vr}$ utility function for \totsep\ $\vr$.}\label{alg:TU-core-totsep} 
	\KwI{A profile~$\ppp=(\aaa,\vvv,\RR)$ with $\scoreTU{\vr}(v,\{\alt\})$ for all $v\in\vvv, \alt\in\aaa$ and $\kk\in [|\aaa|]$.}
	
	Sort and rename the alternatives such that $\scoreTU{\vr}(\vvv,\{\alt_1\}) \ge \dots \ge \scoreTU{\vr}(\vvv,\{\alt_m\})$, breaking ties arbitrarily.
	
	\smallskip
	\Comment{Distribution of the surplus}\vspace{-1mm}
	\ForEach{$(j,v_i)\in [\kk-1]\times \vvv$}{\label{surplus1}%
		\leIf{$\scoreTU{\vr}(v_i,\{\alt_j\})\neq0$}{$\alloc_j(v_i) \leftarrow \frac{\scoreTU{\vr}(\vvv,\{\alt_j\}) - \scoreTU{\vr}(\vvv,\{\alt_k\})}{\scoreTU{\vr}(\vvv,\{\alt_j\})}\cdot\scoreTU{\vr}(v_i,\{\alt_j\})$}{$\alloc_j(v_i) \leftarrow 0$}\label{surplus2}
	}
	\medskip
	
	\Comment{Distribution of the baseline score}\vspace{-1mm}
	\lForEach{$v_i \in \vvv$}{\label{baseline1}%
		$\alloc_{\kk}(v_i) \leftarrow \frac{\kk\,\scoreTU{\vr}(\vvv,\{\alt_k\})}{|\vvv|}$
	}
	\Return{$\sum_{j\in [\kk]}\alloc_j$}
\end{algorithm}

\begin{restatable}{theorem}{thmTUAVpoly}\label{thm:TU-AV-poly}
	\cref{alg:TU-core-totsep} computes a $\TUcs{\vr}$ utility function in polynomial time for every \totsep\ voting rule $\vr$ that satisfies $\scoreTU{\vr}(v,\{\alt\})\geq0$ for all $v\in\vvv$ and $\alt\in\aaa$.
\end{restatable}
\begin{proof}
	Let $\vr$ be a \totsep\ voting rule.
	Clearly, \cref{alg:TU-core-totsep} runs in polynomial time: sorting alternatives by their \vr-scores and computing the $\kk$ utility functions can be done in polynomial time if the scores $\scoreTU{\vr}(v,\{c\})$ are given for each voter-alternative pair.
	Let $\alt_1,\dots,\alt_m$ be as sorted in \cref{alg:TU-core-totsep}.
	By definition, $\W=\{\alt_1,\dots,\alt_{\kk}\}$ is a winning committee under \vr.
	It remains to show that the returned function~$\alloc$ is in the \TU-core.
	
	We start by showing that the utility function is feasible for the grand coalition, i.e., $\sum_{v_i\in \vvv}\alloc(v_i) = \cval{\vr}(\vvv)= \sum_{j\in [\kk]}\scoreTU{\vr}(\vvv,\{\alt_j\})$:
	\begin{align*}
		&\sum_{v_i\in \vvv}\alloc(v_i) = \sum_{j\in [\kk]} \sum_{v_i\in \vvv} \alloc_j(v_i)\\ 
		& = \sum_{j\in [\kk-1]} \sum_{\substack{v_i : \\ \scoreTU{\vr}(v_i,\{\alt_j\}) > 0}}\frac{\scoreTU{\vr}(\vvv,\{\alt_j\}) - \scoreTU{\vr}(\vvv,\{\alt_k\})}{\scoreTU{\vr}(\vvv,\{\alt_j\})}\cdot\scoreTU{\vr}(v_i,\{\alt_j\})
		\;+\; \sum_{v_i\in \vvv} \frac{\kk\,\scoreTU{\vr}(\vvv,\{\alt_k\})}{|\vvv|} \\
		&= \sum_{j\in [\kk-1]}\bigl(\scoreTU{\vr}(\vvv,\{\alt_j\}) - \scoreTU{\vr}(\vvv,\{\alt_k\})\bigr) \;+\; \kk\scoreTU{\vr}(\vvv,\{\alt_k\})= \sum_{j\in [\kk]}\scoreTU{\vr}(\vvv,\{\alt_j\}), \text{ as desired. }
	\end{align*}
	Now, we show that no coalition can block~$\alloc$.
	Fix an arbitrary coalition~$\vvv'\subseteq \vvv$ and fix an admissible committee~$\W'\subseteq \aaa$ for $\vvv'$ with $|\W'|\le \share$.
	We will show that $\score{\vr}(\vvv', \W') \le \sum_{v_i\in \vvv'} \alloc(v_i)$, which implies that $\vvv'$ cannot block~$\alloc$.
	First, split the alternatives from~$\W'$ into alternatives from $\{\alt_1,\dots,\alt_{\kk-1}\}$ that induce a surplus and the rest, and define
	$\W_1=\{\alt_j\in\{\alt_1,\dots,\alt_{\kk-1}\}\cap\W'\mid \scoreTU{\vr}(\vvv,\{\alt_j\})>\scoreTU{\vr}(\vvv,\{\alt_k\})\}$ and $\W_2 = \W'\setminus\W_1$.
	\allowdisplaybreaks
	Then 
	\begin{align}
		\score{\vr}(\vvv',\W') = \sum_{\alt_j\in \W'}\scoreTU{\vr}(\vvv',\{\alt_j\}) =~ & 
		\sum_{\alt_j\in \W_1}\scoreTU{\vr}(\vvv',\{\alt_j\}) +
		\sum_{\alt\in \W_2}\scoreTU{\vr}(\vvv',\{\alt\}) \nonumber \\
		\le~&   \sum_{\alt_j\in \W_1}\scoreTU{\vr}(\vvv',\{\alt_j\}) + |\W_2|\cdot \scoreTU{\vr}(\vvv,\{\alt_k\}),\label{TU-core-AV-LHS}
	\end{align}
	where the last inequality holds because $\W=\{\alt_1,\dots,\alt_{\kk}\}$ consists of the $\kk$ alternatives with highest \vr-scores, and hence
	$\scoreTU{\vr}(\vvv,\{\alt\})\le \scoreTU{\vr}(\vvv,\{\alt_k\})$ for every $\alt\notin \W_1$; moreover, $\scoreTU{\vr}(\vvv',\{\alt\})\le \scoreTU{\vr}(\vvv,\{\alt\})$.
	
	Next, we lower-bound $\sum_{v_i\in \vvv'}\alloc(v_i)$.
	\begin{claim}\label{claim:lower-bound}
		$\displaystyle\sum_{v_i\in\vvv'}\alloc(v_i)
		\;\geq\;\sum_{\alt_j\in \W_1}\scoreTU{\vr}(\vvv',\{\alt_j\}) + |\W_2|\cdot \scoreTU{\vr}(\vvv,\{\alt_k\})$
	\end{claim}
	
	\begin{claimproof}{claim:lower-bound}
		By definition of~$\alloc$,
		\allowdisplaybreaks
		\begin{align*}  
			&\sum_{v_i\in \vvv'}\alloc(v_i)
			= \sum_{v_i\in \vvv'}\sum_{j\in [\kk]}\alloc_j(v_i) \ge \sum_{\alt_j\in \W_1}\sum_{v_i\in \vvv'}\alloc_j(v_i) + \sum_{v_i\in \vvv'}\alloc_{\kk}(v_i)\\
			&= \sum_{\alt_j\in \W_1}\left(\sum_{\substack{v_i \in \vvv': \\ \scoreTU{\vr}(v_i,\{\alt_j\}) > 0}}\frac{\scoreTU{\vr}(\vvv,\{\alt_j\}) - \scoreTU{\vr}(\vvv,\{\alt_k\})}{\scoreTU{\vr}(\vvv,\{\alt_j\})}\cdot\scoreTU{\vr}(v_i,\{\alt_j\})\right)
			+ \sum_{v_i \in \vvv'}  \frac{\kk\,\scoreTU{\vr}(\vvv,\{\alt_k\})}{|\vvv|}\\
			&= \sum_{\alt_j\in \W_1}\left(\sum_{\substack{v_i \in \vvv': \\ \scoreTU{\vr}(v_i,\{\alt_j\}) > 0}}\frac{\scoreTU{\vr}(\vvv,\{\alt_j\}) - \scoreTU{\vr}(\vvv,\{\alt_k\})}{\scoreTU{\vr}(\vvv,\{\alt_j\})}\cdot\scoreTU{\vr}(v_i,\{\alt_j\})\right)
			+ \frac{|\vvv'|\kk\,\scoreTU{\vr}(\vvv,\{\alt_k\})}{|\vvv|}\\
			&\geq \sum_{\alt_j\in \W_1}\left(\sum_{\substack{v_i \in \vvv': \\ \scoreTU{\vr}(v_i,\{\alt_j\}) > 0}}\frac{\scoreTU{\vr}(\vvv,\{\alt_j\}) - \scoreTU{\vr}(\vvv,\{\alt_k\})}{\scoreTU{\vr}(\vvv,\{\alt_j\})}\cdot\scoreTU{\vr}(v_i,\{\alt_j\})\right)
			+ (|\W_1|+ |\W_2|)\scoreTU{\vr}(\vvv,\{\alt_k\})\\
			&= \sum_{\alt_j\in \W_1}\left(\scoreTU{\vr}(\vvv,\{\alt_k\})+\sum_{\substack{v_i \in \vvv': \\ \scoreTU{\vr}(v_i,\{\alt_j\}) > 0}}\frac{\scoreTU{\vr}(\vvv,\{\alt_j\}) - \scoreTU{\vr}(\vvv,\{\alt_k\})}{\scoreTU{\vr}(\vvv,\{\alt_j\})}\cdot\scoreTU{\vr}(v_i,\{\alt_j\})\right)
			+ |\W_2|\scoreTU{\vr}(\vvv,\{\alt_k\})\\
			&= \sum_{\alt_j\in \W_1}\left(\scoreTU{\vr}(\vvv,\{\alt_k\})+\frac{\scoreTU{\vr}(\vvv,\{\alt_j\}) - \scoreTU{\vr}(\vvv,\{\alt_k\})}{\scoreTU{\vr}(\vvv,\{\alt_j\})}\cdot(\scoreTU{\vr}(\vvv,\{\alt_j\})-\sum_{v_i\notin\vvv'}\scoreTU{\vr}(v_i,\{\alt_j\}))\right)\\
			&\phantom{oo}+ |\W_2|\scoreTU{\vr}(\vvv,\{\alt_k\})\\
			&\geq \sum_{\alt_j\in \W_1}\left(\scoreTU{\vr}(\vvv,\{\alt_j\})-\frac{\scoreTU{\vr}(\vvv,\{\alt_j\})}{\scoreTU{\vr}(\vvv,\{\alt_j\})}\cdot\sum_{v_i\notin\vvv'}\scoreTU{\vr}(v_i,\{\alt_j\})\right)\\
			&\phantom{oo}+ |\W_2|\scoreTU{\vr}(\vvv,\{\alt_k\})\\
			&=\sum_{\alt_j\in \W_1}\scoreTU{\vr}(\vvv',\{\alt_j\}) + |\W_2|\cdot \scoreTU{\vr}(\vvv,\{\alt_k\})
		\end{align*}
		The first inequality holds since $|\W'|=|\W_1|+|\W_2|\le \share$.

		For the fifth line, we use $\scoreTU{\vr}(\vvv,\{\alt_j\})=\sum_{v_i\in\vvv'}\scoreTU{\vr}(v_i,\{\alt_j\})+\sum_{v_i\notin\vvv'}\scoreTU{\vr}(v_i,\{\alt_j\})$, i.e., the \votsep\ property. In the sixth line, we use that $\scoreTU{\vr}(\vvv,\{\alt_k\})\geq0$.
	\end{claimproof}
	Combining \eqref{TU-core-AV-LHS} with \cref{claim:lower-bound} yields
	$\sum_{v_i\in \vvv'}\alloc(v_i) \ge \score{\vr}(\vvv',\W')$, as desired. 
	Since $\vvv'\subseteq \vvv$ was arbitrary, no coalition blocks~$\alloc$, hence $\alloc$ lies in the \TU-core.
\end{proof}

\begin{example}\label{ex:AVTUalg}
  Recall the instance from \cref{ex:TUgames} with $\kk=2$.
  For \AV, the winning committee is $\W=\{\alt_1,\alt_2\}$
  with $\score{\AV}(\vvv,\{\alt_1\}) = 3$ and $\score{\AV}(\vvv,\{\alt_2\})=2$.
  \cref{alg:TU-core-totsep} gives every voter~$\alloc_2(v_i)=2\cdot 2 / 4 =1$.
  The surplus of~$\alt_1$ over~$\alt_2$ (line \ref{surplus2}) is shared equally among~$\alt_1$'s three approving voters~$\alloc_1(v_i)=(3-2)/3$ for all~$i\in [3]$, while $\alloc_1(v_4)=0$ since she does not approve~$\alt_1$.
  Summing the two components yields $\alloc=(1+1/3, 1+1/3, 1+1/3, 1+0)=(4/3,4/3,4/3,1)$.
  For \SAV, the winning committee is $\W=\{\alt_1,\alt_2\}$
  with $\score{\SAV}(\vvv,\{\alt_1\}) = 2.5$ and $\score{\SAV}(\vvv,\{\alt_2\})=1$.
  \cref{alg:TU-core-totsep} gives every voter~$\alloc_2(v_i)=2\cdot 1 / 4 =0.5$.
  
  The surplus of~$\alt_1$ over~$\alt_2$ (line \ref{surplus2}) is shared among~$\alt_1$'s three approving voters according to the score they assign to $\alt_1$~$\alloc_1(v_1)=\alloc_1(v_2)=(2.5-1)/2.5\cdot 1=0.6$ and $\alloc_1(v_3)=(2.5-1)/2.5\cdot 0.5=0.3$, while $\alloc_1(v_4)=0$ since she does not approve~$\alt_1$.
  Summing the two components yields $\alloc=(0.5+0.6 ,0.5+0.6,0.5+0.3,0.5)=(1.1,1.1,0.8,0.5)$.
\end{example}
The preceding results rely on the additivity of \totsep\ voting rules, which allows us to compare any ``deviating'' committee to a baseline alternative and to distribute ``surplus'' only among approvers. Note that \cref{alg:TU-core-totsep} and \cref{thm:AV-TU-Core-Non-Empty} also work for voting rules over linear preferences that are \totsep. This includes the multi-winner variant of the Borda rule.
For \CC\ and \PAV, the voter-level scoring functions introduce complementarities and nonlinearities that break this structure, and core stability may fail even with transfers.

\begin{restatable}[]{proposition}{propPAVCCTUCNE}\label{prop:PAVCCTUCNE}
  \TUcore{\CC} and \TUcore{\PAV} can be empty. 
\end{restatable}
\begin{proof}
  Consider the instance~$I=(\ppp,\kk)$ with voters~$\vvv=\{v_1,\dots,v_6\}$, alternatives~$\aaa=\{\alt_1,\alt_2,\alt_3,\alt_4\}$, and $\kk=2$.
  The approval preferences of the voters are as follows: $\app_1=\{\alt_1,\alt_2\}$,
  $\app_2=\{\alt_1,\alt_4\}$,
  $\app_3=\{\alt_1,\alt_3\}$,
  $\app_4=\{\alt_2,\alt_3\}$,
  $\app_5=\{\alt_2,\alt_4\}$,
  $\app_6=\{\alt_3, \alt_4\}$.

  One can verify that for any size-$2$ committee the \CC-score is at most five,
  while the \PAV-score is at most $5.5$,
  i.e., $\cval{\CC}(\vvv)=5$ and  $\cval{\PAV}(\vvv)= 5.5$. 
  However, every alternative is approved by exactly three voters.
  Since $|\vvv|/\kk=3$, the corresponding three voters can claim one seat and obtain a total score of three (under \CC\ as well as \PAV).
  These give four linear inequalities:
  \begin{align*}
    \alloc_1+\alloc_2+\alloc_3\geq 3,\quad    \alloc_1+\alloc_4+\alloc_5\geq 3, \quad
    \alloc_3+\alloc_4+\alloc_6\geq 3,\quad   \alloc_2+\alloc_5+\alloc_6\geq 3.
  \end{align*}
  Combining them yields $\sum_{i\in [6]}\alloc_i \ge 6$, which is not feasible given that the grand coalition can only achieve a score of at most $5$ (resp. $5.5$).
\end{proof}

Even when the \TU-core is non-empty, core stability is not guaranteed by standard fairness-driven %
solution concepts.
The next proposition formalizes this point for our setting: the Shapley value can assign too little total payoff to some entitled coalition, enabling a blocking deviation.

\begin{restatable}[]{proposition}{propShapley}\label{prop:Shapley}
  The utility function defined by the Shapley value is not always in $\TUcore{\vr}$ for any $\vr\in \{\AV, \SAV, \CC, \PAV\}$,
  even if the core is not empty. %
\end{restatable}
\begin{proof}
We first consider \AV.
Consider the following profile~$\ppp=(\aaa, \vvv, \RR)$ with 
$\aaa=\{\alt_1,\dots,\alt_5\}$, $\vvv=\{v_1,\dots,v_6\}$, and
$\app_1=\app_2=\app_3=\{\alt_1,\alt_2\}$, $\app_4=\app_5=\app_6=\{\alt_3,\alt_4,\alt_5\}$. Let $\kk=4$.

\mypara{Characteristic function.}
Every alternative has three approving voters: $|\VV(\alt_j)|=3$, $j\in [5]$.
Seat caps are $\lfloor|\vvv'|\cdot 4/6\rfloor$, giving cap $0,1,2,2,3,4$ for
$|\vvv'|=1,\dots,6$.
By symmetry, coalition values depend only on the group composition~$(g_1,g_2)$ (number of voters from each group).
The non-trivial values are:

\begin{tabular}{@{}c|c@{\;}c@{\;}c@{\;}c@{\;}c@{\;}c@{\;}c@{\;}c@{\;}c@{\;}c@{\;}c@{\;}c@{\;}c@{\;}c@{}}
  \toprule
$(g_1,g_2)$ & $(2,0)$ & $(1,1)$ & $(0,2)$ & $(3,0)$ & $(2,1)$ & $(1,2)$ & $(0,3)$ & $(3,1)$ & $(2,2)$ & $(1,3)$ & $(3,2)$ & $(2,3)$ & $(3,3)$ \\
\hline
$\chi$ & 2 & 1 & 2 & 6 & 4 & 4 & 6 & 6 & 4 & 6 & 8 & 9 & 12\\\bottomrule
\end{tabular}
    
\noindent The grand coalition value is $\chi(3,3)=12$ (e.g., $\{\alt_1,\alt_2,\alt_3,\alt_4\}$ gives $3\cdot 2+3\cdot 2=12$).
    
\mypara{Shapley value.} For voter~$v_i$,
\begin{align*}
      \phi_i = \sum_{\vvv'\subseteq \vvv\setminus \{v_i\}} \frac{|\vvv'|! (|\vvv|- |\vvv'|-1)!}{|\vvv|!}\cdot \bigl(\chi(\vvv'\cup \{v_i\})-\chi(\vvv')\bigr).
\end{align*}

By symmetry $\phi_1=\phi_2=\phi_3$ and $\phi_4=\phi_5=\phi_6$, with $\phi_1+\phi_4=4$.

Grouping $v_1$'s marginal contributions by $(a,b)$, the number of members of $S\subseteq\mathcal{V}\setminus\{v_1\}$ from each group
($a\in\{0,1,2\}$, $b\in\{0,1,2,3\}$), with multiplicity $\binom{2}{a}\binom{3}{b}$ and Shapley weight $\tfrac{(a{+}b)!\,(5{-}a{-}b)!}{6!}$, the nine \emph{non-zero} terms are from~$(0,1)$, $(0,2)$,
    $(1,0)$, $(1,1)$, $(1,3)$, $(2,0)$, $(2,1)$, $(2,2)$, $(2,3)$. 
This yields Shapley value~$\phi_1=\tfrac{19}{10}$ (we omit the routine calculation).

\smallskip
   
\mypara{Blocking.}
Coalition $\{v_1,v_2,v_3\}$ has seat cap~$2$ and picks $\{\alt_1,\alt_2\}$, achieving score~$6$, while
$\phi_1+\phi_2+\phi_3=\tfrac{57}{10}<6$.
The \TUcore{\AV} is non-empty by \cref{thm:AV-TU-Core-Non-Empty}; \cref{alg:TU-core-totsep} yields $\boldsymbol{\alpha}=(2,2,2,2,2,2)$.

\medskip
We next consider \SAV.
Consider the following profile~$\ppp=(\aaa, \vvv, \RR)$ with 
$\aaa=\{\alt_1,\dots,\alt_6\}$, $\vvv=\{v_1,\dots,v_4\}$, and
$\app_1=\{\alt_1\}$,
$\app_2=\app_3=\app_4=\{\alt_2,\ldots,\alt_6\}$. Let $\kk=2$.

\mypara{Characteristic function.}
Seat caps are $\lfloor|\vvv'|\cdot 2/4\rfloor$, giving cap $0,1,1,2$ for
$|\vvv'|=1,\dots,4$.
By symmetry, coalition values depend only on the group composition~$(g_1,g_2)$ (number of voters from each group).
The non-trivial values are:

\begin{tabular}{@{}c|c@{\;}c@{\;}c@{\;}c@{\;}c@{}}
	\toprule
	$(g_1,g_2)$ & $(1,1)$ & $(0,2)$ & $(1,2)$ & $(0,3)$ & $(1,3)$\\
	\hline
	$\chi$ & $1$ & $\frac{2}{5}$ & $1$ & $\frac{3}{5}$ & $\frac{8}{5}$\\\bottomrule
\end{tabular}

\noindent The grand coalition value is $\chi(1,3)=\frac{8}{5}$ (e.g., $\{c_1,c_2\}$ gives $1+\frac{3}{5}=\frac{8}{5}$).

\mypara{Shapley value.} 
By symmetry $\phi_2=\phi_3=\phi_4$, with $\phi_1+\phi_2+\phi_3+\phi_4=\frac{8}{5}$.

Grouping $v_1$'s marginal contributions by $(0,a)$, the number of members of $S\subseteq\mathcal{V}\setminus\{v_1\}$ from the second group ($a\in\{1,2,3\}$), with multiplicity $\binom{3}{a}$ and Shapley weight $\tfrac{(a)!\,(3{-}a)!}{4!}$, the three \emph{non-zero} terms are from~$(0,1)$, $(0,2)$,
$(0,3)$. 
This yields Shapley value~$\phi_1=\tfrac{65}{100}$ (we omit the routine calculation), which implies $\phi_2=\phi_3=\phi_4=\frac{95}{300}$.

\smallskip

\mypara{Blocking.}
Coalition $\{v_1,v_2\}$ has seat cap~$1$ and picks $\{\alt_1\}$, achieving score~$1$, while
$\phi_1+\phi_2=\tfrac{290}{300}<1$.
The \TUcore{\SAV} is non-empty by \cref{thm:AV-TU-Core-Non-Empty}; $\boldsymbol{\alpha}=(1,\frac{1}{5},\frac{1}{5},\frac{1}{5})$ is in the core.

\medskip
\noindent We now consider \CC. 
Consider the following profile~$\ppp=(\aaa, \vvv, \RR)$ with $\aaa=\{\alt_1,\alt_2, \alt_3,\alt_4\}$, $\vvv=\{v_1, v_2,v_3,v_4\}$,
and $\app_1=\{\alt_1,\alt_2\}$, $\app_2=\{\alt_1,\alt_3\}$, $\app_3=\{\alt_2,\alt_3\}$, $\app_4=\{\alt_4\}$.
Let~$\kk=2$. 	

\mypara{Characteristic function.} Every alternative except for $\alt_4$ has two approving voters: $|\VV(\alt_j)|=2$, $j\in[3]$ and $|\VV(\alt_4)|=1$. Seat caps are $\lfloor|\vvv'|\cdot\frac{2}{4}\rfloor$, giving cap $0,1,1,2$ for $|\vvv'|=1,\ldots,4$. By symmetry, coalition values depend only on the group composition $(g_1,g_2)$, where $g_1$ is the number of voters from $\{v_1,v_2,v_3\}$ and $g_2$ is the number of voters from $\{v_4\}$. The non-trivial values are:

\begin{tabular}{@{}c|c@{\;}c@{\;}c@{\;}c@{\;}c@{}}
  \toprule
$(g_1,g_2)$ & $(2,0)$ & $(1,1)$ & $(3,0)$ & $(2,1)$ & $(3,1)$ \\
\hline
$\chi$ & 2 & 1 & 2 & 2 & 3\\\bottomrule
\end{tabular}

The grand coalition value is $\chi(3,1)=3$ (e.g.,$\{\alt_1,\alt_2\}$ gives 3).

\mypara{Shapley value.}
By symmetry $\phi_1=\phi_2=\phi_3$. Grouping $v_1$'s marginal contributions by $(a,b)$ the number of members of $S\subseteq\vvv\setminus\{v_1\}$ from each group ($a\in\{1,2\}$, $b\in\{0,1\}$), with multiplicity $\binom{2}{a}\binom{1}{b}$ and Shapley weight $\frac{(a+b)!(3-a-b)!}{4!}$, the four \emph{non-zero} terms are from $(1,0),(0,1),(1,1),(2,1)$. This yields Shapley value $\phi_1=\frac{5}{6}$ (we omit the routine calculation).

\mypara{Blocking.} Coalition $\{v_1,v_2\}$ has seat cap $1$ and picks $\{\alt_1\}$, achieving score $2$, while $\phi_1+\phi_2=\frac{10}{6}<2$. The \TU-\CC-core is non-empty, as $\alloc=(1,1,1,0)$ is in the core.

\medskip
\noindent Finally, we consider \PAV: Consider the same profile, as for \CC.

\mypara{Characteristic function.} The characteristic function is mostly identical to the \CC-case with the small difference of $\chi(3,1)=3.5$. This gives us the following non-trivial values:

\begin{tabular}{@{}c|c@{\;}c@{\;}c@{\;}c@{\;}c@{}}
  \toprule
$(g_1,g_2)$ & $(2,0)$ & $(1,1)$ & $(3,0)$ & $(2,1)$ & $(3,1)$ \\
\hline
$\chi$ & 2 & 1 & 2 & 2 & 3.5\\\bottomrule
\end{tabular}

\mypara{Shapley value.} Using the same methods and arguments as for \CC, we get the Shapley value $\phi_1=\frac{23}{24}$. 

\mypara{Blocking.} Coalition $\{v_1,v_2\}$ has seat cap $1$ and picks $\{\alt_1\}$, achieving score $2$, while $\phi_1+\phi_2=\frac{23}{12}<2$. The \TU-\PAV-core is non-empty, as $\alloc=(1,1,1,0.5)$ is in the core.

\end{proof}

In summary, the \totsep\ rules are the rules among the four for which the \TU-core is unconditionally non-empty; we next ask how this picture changes when transfers are removed.

\subsection{\NTU: Equivalences and the \CC\ core}\label{sub:structure-NTU}
We now turn to the \NTU\ interpretation, which matches the original committee-core model~\cite{AzizBCEFW17}
in which coalitions cannot transfer utility, and blocking requires a deviation that strictly improves each member's realized score.
This makes the core sensitive to the feasible \emph{utility profiles} induced by committees, rather than only to total attainable score.
We first clarify the \AV-\SAV-\PAV\ equivalence (\cref{obs:AVPAVequiv}) and then focus on \CC.

\begin{restatable}[]{observation}{obsAVPAVequiv}\label{obs:AVPAVequiv}
  Let $(\ppp,\kk)$ be an \mw-instance with $\ppp=(\aaa,\vvv,\RR)$.
  For each coalition~$\vvv'\subseteq \vvv$,
  the function $f_{\vvv'}\colon \NTUutil{\AV}(\vvv')\to \NTUutil{\PAV}(\vvv')$ defined by
    $f_{\vvv'}(\alloc)(v)=\sum_{z\in [\alloc(v)]} 1/z$, for all $v\in\vvv'$,
  where the sum is interpreted as $0$ if $\alloc(v)=0$ is a bijection, and both $f_{\vvv'}$ and $f_{\vvv'}^{-1}$ are computable in polynomial time. 
  For each coalition~$\vvv'\subseteq \vvv$,
  the function $g_{\vvv'}\colon \NTUutil{\AV}(\vvv')\to \NTUutil{\SAV}(\vvv')$ defined by
  $g_{\vvv'}(\alloc)(v)=\begin{cases}
  	\frac{\alloc(v)}{|\app(v)|} & \app(v)\neq\emptyset \\
  	0 & \, \text{else}
  \end{cases}$,  is a bijection, and both $g_{\vvv'}$ and $g_{\vvv'}^{-1}$ are computable in polynomial time. 
  Consequently, for each coalition~$\vvv'\subseteq \vvv$, the function $h_{\vvv'}\colon \NTUutil{\SAV}(\vvv')\to \NTUutil{\PAV}(\vvv')$ defined by $h_{\vvv'}=f_{\vvv'}\circ g_{\vvv'}^{-1}$ is a bijection, and both $h_{\vvv'}$ and $h_{\vvv'}^{-1}$ are computable in polynomial time. 
  Moreover, for every $\alloc\in \NTUutil{\AV}(\vvv)$ and every coalition $\vvv'\subseteq\vvv$,
  {\centering
  	$\vvv' \text{ is } \NTU\text{-}\AV\text{-blocking for }\alloc
  	\text{ if and only if it is } \NTU\text{-}\PAV\text{-blocking for } f_{\vvv'}(\alloc)$\\
  	and\\
  	$\vvv' \text{ is } \NTU\text{-}\AV\text{-blocking for }\alloc
  	\text{ if and only if it is } \NTU\text{-}\SAV\text{-blocking for } g_{\vvv'}(\alloc)$.
  	\par}
   As a consequence, for every $\alloc\in \NTUutil{\SAV}(\vvv)$ and every coalition $\vvv'\subseteq\vvv$,
    {\centering
    $\vvv' \text{ is } \NTU\text{-}\SAV\text{-blocking for }\alloc
    \text{ if and only if it is } \NTU\text{-}\PAV\text{-blocking for } h_{\vvv'}(\alloc)$.
    \par}
\end{restatable}
\begin{proof}
  We first show the statements regarding the connection between \AV\ and \PAV:
  To simplify the notation, for each fixed $x\in\{0,\dots,\kk\}$ let \myemph{$\phi(x)$} $\coloneqq \sum_{z=1}^{x}\frac{1}{z}$ (with $\phi(0)=0$) denote the harmonic sum over~$x$.
  Note that $f_{\vvv'}(\alloc)(v)=\phi(\alloc(v))$ for all $v\in\vvv'$.
  
  We first show that $f$ is a bijection.
  Let $\alloc\in\NTUutil{\AV}(\vvv')$ be a feasible utility function for~$\vvv'$ under \AV.
  By feasibility, there exists a committee~$\W$ of size $\kk$ such that
  $\alloc(v)=\scoreNTU{\AV}(v,\W)$ for all $v\in\vvv'$.
  By the definition of \PAV,
  \[
    \scoreNTU{\PAV}(v,\W)=\phi(\alloc(v))=f_{\vvv'}(\alloc)(v),
  \]
  hence $f_{\vvv'}(\alloc)\in\NTUutil{\PAV}(\vvv')$.

  Conversely, let $\alloc'\in\NTUutil{\PAV}(\vvv')$ and let $\W$ witness feasibility, i.e.,
  $\alloc'(v)=\scoreNTU{\PAV}(v,\W)=\phi(\scoreNTU{\AV}(v,W))$ for all $v\in\vvv'$.
  Define $\alloc(v)=\scoreNTU{\AV}(v,\W)\in\{0,\dots,\kk\}$.
  Then $\alloc\in\NTUutil{\AV}(\vvv')$ and $f_{\vvv'}(\alloc)=\alloc'$.

  Since $\phi$ is strictly increasing on $\{0,\dots,\kk\}$, it is injective; therefore $f_{\vvv'}$ is a bijection.
  Polynomial-time computability follows by precomputing $\phi(0),\dots,\phi(\kk)$ and inverting via table lookup.

  Now, we show the blocking equivalence.
  Let $\alloc\in\NTUutil{\AV}(\vvv')$ and write $\alloc'=f_{\vvv'}(\alloc)$.
  Assume that $\vvv'$ is \NTUblocking{\AV} $\alloc$.
  By definition, let $\W'$ be an admissible committee with $|\W'|\le \share$ such that for all $v\in\vvv'$, $\scoreNTU{\AV}(v,\W')>\alloc(v)$.
  Applying $\phi$ and using strict monotonicity, for all $v\in\vvv'$ we obtain
  $\phi(\scoreNTU{\AV}(v,\W'))>\phi(\alloc(v))$.
  Using $\scoreNTU{\PAV}(v,\W')=\phi(\scoreNTU{\AV}(v,\W'))$ and $\alloc'(v)=\phi(\alloc(v))$,
  this is equivalent to $\scoreNTU{\PAV}(v,\W')>\alloc'(v)$ for all $v\in\vvv'$,
  hence $\vvv'$ is \NTUblocking{\PAV} $\alloc'$.

  The reverse implication is analogous (or follows by applying the same argument to $f_{\vvv'}^{-1}$).
  We now show the statements regarding the connection between \AV\ and \SAV:
  We first show that $g$ is a bijection. 
  Let $\alloc\in\NTUutil{\AV}(\vvv')$ be a feasible utility function for~$\vvv'$ under \AV.
  By feasibility, there exists a committee~$\W$ of size $\kk$ such that
  $\alloc(v)=\scoreNTU{\AV}(v,\W)$ for all $v\in\vvv'$.
    By the definition of \SAV,
  \[
  \scoreNTU{\SAV}(v,\W)=\begin{cases}
  \frac{|\app(v)\cap \W|}{|\app(v)|} & \app(v)\neq\emptyset \\
  	0 & \, \text{else}
  \end{cases}=g_{\vvv'}(\alloc)(v)
  \]
    hence $g_{\vvv'}(\alloc)\in\NTUutil{\SAV}(\vvv')$.
    
      Conversely, let $\alloc'\in\NTUutil{\SAV}(\vvv')$ and let $\W$ witness feasibility, i.e.,
    $\alloc'(v)=\scoreNTU{\SAV}(v,\W)=\begin{cases}
    	\frac{|\app(v)\cap \W|}{|\app(v)|} & \app(v)\neq\emptyset \\
    	0 & \, \text{else}
    \end{cases}$ for all $v\in\vvv'$.
    Define $\alloc(v)=\scoreNTU{\AV}(v,\W)\in\{0,\dots,\kk\}$.
    Then $\alloc\in\NTUutil{\AV}(\vvv')$ and $g_{\vvv'}(\alloc)=\alloc'$.
    Since the value of $\frac{|\app(v)\cap \W|}{|\app(v)|}$ is strictly increasing with increasing $|\app(v)\cap \W|$, it is injective; therefore $g_{\vvv'}$ is a bijection. Polynomial-time computability follows directly by definition of the scores, as we just divide (or multiply) by $|\app(v)|$.
    
      Now, we show the blocking equivalence.
    Let $\alloc\in\NTUutil{\AV}(\vvv')$ and write $\alloc'=g_{\vvv'}(\alloc)$.
    Assume that $\vvv'$ is \NTUblocking{\AV} $\alloc$.
    By definition, let $\W'$ be an admissible committee with $|\W'|\le \share$ such that for all $v\in\vvv'$, $\scoreNTU{\AV}(v,\W')>\alloc(v)$. As all voters must improve we know that for all $v\in\vvv'$ $\app(v)\neq\emptyset$.
    
    By dividing by $|\app(v)|$, for all $v\in\vvv'$ we obtain
    $\frac{\scoreNTU{\AV}(v,\W')}{|\app(v)|}>\frac{\alloc(v)}{|\app(v)|}$.
    Using $\scoreNTU{\SAV}(v,\W')=\frac{\scoreNTU{\AV}(v,\W')}{|\app(v)|}$ and $\alloc'(v)=\frac{\alloc(v)}{|\app(v)|}$,
    this is equivalent to $\scoreNTU{\SAV}(v,\W')>\alloc'(v)$ for all $v\in\vvv'$,
    hence $\vvv'$ is \NTUblocking{\SAV} $\alloc'$.
      The reverse implication is analogous (or follows by applying the same argument to $g_{\vvv'}^{-1}$).
      
      Finally the statement for $h$ and the blocking equivalence follows directly by chaining the two above arguments.
\end{proof}

Observation~\ref{obs:AVPAVequiv} shows that, in the \NTU\ model, the \AV- \SAV- and \PAV-induced cores are equivalent up to a monotone transformation of voter utilities.
Accordingly, the open non-emptiness question for the committee core can be phrased under either scoring rule.
In contrast, for \CC\ we can resolve existence unconditionally by exploiting a direct correspondence between core stability and representation.
For completeness, we recall the following definition.

\begin{definition}[Justified Representation (JR)\cite{AzizBCEFW17}]\label{def:JR}
  Let $(\ppp,\kk)$ denote an \mw-instance with $\ppp=(\aaa,\vvv,\RR)$.
  A committee~$\W$ with $|\W|=\kk$ satisfies \myemph{justified representation} (\myemph{JR}) if for each cohesive voter coalition~$\vvv'\subseteq \vvv$, it holds that $\W\cap \bigcup_{v_i\in\vvv'} \app_i \neq \emptyset$.
  $\vvv'$ is \myemph{cohesive} if $|\vvv'|\ge \lceil \coalsize \rceil$ and $\bigcap_{v_i\in \vvv'} \app_i \neq \emptyset$.
\end{definition}

For \CC, feasibility and blocking simplify because each voter's utility is binary.
A coalition can block only by moving from utility $0$ to utility $1$ for every member, i.e., by selecting a committee that``covers'' each member with at least one approved alternative.
Under proportional entitlements, this is closely aligned with the JR axiom. %
The next proposition makes this correspondence exact: core stability is equivalent to JR of the witnessing committee.

\begin{restatable}{proposition}{propCCJR}\label{prop:CCJR}
  Let $\alloc$ be an \NTU-utility function feasible for~$\vvv$ under \CC, and let $\W$ be a size-$\kk$ committee witness, i.e.,
  $\alloc(v_i) = \scoreNTU{\CC}(v_i, \W)$ holds for all $v_i \in \vvv$.
  Then, $\alloc \in \NTUcore{\CC}$ if and only if $\W$ satisfies JR.
\end{restatable}
\begin{proof}
    Let $\alloc$ and $\W$ be as in the statement.
  Recall that for \CC\ we have $\scoreNTU{\CC}(v_i,\W)\in\{0,1\}$, and
  $\scoreNTU{\CC}(v_i,\W)=1$ if and only if $\W\cap \app_i\neq\emptyset$.
  To show the statement, we consider two cases: $\W$ satisfies JR and it does not.
  
  Assume that $\W$ satisfies JR and suppose, for contradiction, that $\alloc\notin \NTUcore{\CC}$.
  Then there exist a voter coalition~$\vvv'\subseteq \vvv$ and an admissible committee~$\W'\subseteq \aaa$ with $|\W'|\le \share$ such that
  $\scoreNTU{\CC}(v_i,\W')>\alloc(v_i)$ for all $v_i\in\vvv'$.
  Since both terms lie in $\{0,1\}$, we have $\scoreNTU{\CC}(v_i,\W')=1$ and $\alloc(v_i)=0$ for all $v_i\in\vvv'$.
  Hence each $v_i\in\vvv'$ approves at least one alternative in $\W'$, so by pigeonhole there exists $\alt\in\W'$ approved by at least
  $\lceil \frac{|\vvv'|}{|\W'|} \rceil$~voters in $\vvv'$.
  Let $\vvv''\coloneq \{v_i\in\vvv' \colon \alt\in\app_i\}$; then
  $|\vvv''|\ge \lceil \frac{|\vvv'|}{|\W'|} \rceil\ge \lceil \frac{|\vvv|}{\kk} \rceil$, where the last inequality uses admissibility~$|\W'|\le \share$.
  Thus $\alt\in \cap_{v_i\in\vvv''}\app_i$ and $|\vvv''|\ge \lceil \frac{|\vvv|}{\kk} \rceil$, confirming that $\vvv''$ is cohesive. 
  Moreover, $\alloc(v_i)=0$ for all $v_i\in\vvv''$ implies $\W\cap \cup_{v_i\in\vvv''}\app_i=\emptyset$,
  contradicting JR.

  Now, assume that $\W$ violates JR, witnessed by $\vvv'\subseteq\vvv$ and $\alt\in\cap_{v_i\in\vvv'}\app_i$ with $|\vvv'|\ge \lceil \frac{|\vvv|}{\kk} \rceil$ and $\W\cap\app_i=\emptyset$ for all $v_i\in\vvv'$.
  Then, $\alloc(v_i)=\score{\CC}(v_i, \W)=0$ for all $v_i\in\vvv'$, while for $\W'=\{\alt\}$ we have $\scoreNTU{\CC}(v_i,\W')=1>\alloc(v_i)$ for all $v_i\in\vvv'$.
  Moreover, $\W'$ is admissible for~$\vvv'$ since $|\vvv'|\ge \lceil \frac{|\vvv|}{\kk} \rceil$ implies that $\share  \ge 1 = |\W'| $. 
  Hence, $\vvv'$ is \NTUblocking{\CC} $\alloc$, witnessing that $\alloc\notin \NTUcore{\CC}$.
  
  Therefore, $\alloc\in\NTUcore{\CC}$ if and only if $\W$ satisfies JR.
\end{proof}

\cref{prop:CCJR} immediately implies the following. %
\begin{corollary}\label{cor:JR->NTU-CC-Core}
  Every JR-satisfying voting rule finds a utility function in the $\NTUcore{\CC}$.
\end{corollary}
With the correspondence in \cref{prop:CCJR}, it suffices to compute a size-$\kk$ committee $\W$ that satisfies JR, since the induced utility function then belongs to $\NTUcore{\CC}$.
It is known that \CC, \PAV, and a greedy variant of \CC~(aka.\ Greedy\AV~\cite{AzizBCEFW17}) satisfy JR.
However, both \CC\ and \PAV\ are computationally difficult while Greedy\AV\ runs in polynomial time~\citep{AzizBCEFW17}.
For completeness, we present the greedy procedure in
\cref{alg:NTU-core-CC} and give a self-contained proof of the JR property.

\begin{algorithm}[t!]
  \caption{Compute a $\NTUcs{\CC}$ utility function.}\label{alg:NTU-core-CC}
  \KwI{A profile~$\ppp=(\aaa,\vvv,\RR)$ with $\RR=(\app_i)_{v_i\in \vvv}$ and $\kk \in [|\aaa|]$.}

  $\vvv'\leftarrow \vvv$;  $\W \leftarrow \emptyset$;\label{alg:NTU-core-CC:init1}

  Initialize $\alloc\colon \vvv \to \mathds{N}$ with $\alloc(v_i) = 0$ for all~$v_i\in \vvv$;\label{alg:NTU-core-CC:init2}
   
  \While{$\vvv'\neq \emptyset$ \KwAnd $|\W| < \kk$}{
    Let $\alt\leftarrow \argmax_{\alt_j\in \aaa\setminus \W}|\VV(\alt_j)|$\label{alg:NTU-core-CC:argmax} \tcp*{Recall that $\VV(\alt_j)$ denotes the set of voters that approve of~$\alt_j$}

   $\W\leftarrow \W \cup \{\alt\}$;   $\vvv'\leftarrow \vvv'\setminus \VV(\alt)$\label{alg:NTU-core-CC:W}

    \lForEach{$\alt_j \in \aaa\setminus \W$}{%
      $\VV(\alt_j)\leftarrow \VV(\alt_j)\setminus \VV(\alt)$
    }\label{alg:NTU-core-CC:V(c)}
    
    \lForEach{$v_i\in \VV(\alt)$}{$\alloc(v_i) \leftarrow 1$}\label{alg:NTU-core-CC:alloc}
  }

  \lIf{$|\W| < \kk$}{%
    Pick $\kk-|\W|$ arbitrary alternatives from~$\aaa\setminus \W$ and add them to~$\W$
  }

  \Return{$(\alloc,\W)$}
\end{algorithm}

\begin{restatable}[]{theorem}{thmCCNTUCE}\label{thm:CCNTUCE}
  \cref{alg:NTU-core-CC} computes a utility function in $\NTUcore{\CC}$ in polynomial time,
  and hence $\NTUcore{\CC}$ is always non-empty.
\end{restatable}
\begin{proof}
 Clearly, \cref{alg:NTU-core-CC} runs in polynomial time.
 Let $(\alloc,\W)$ be the output. By construction, $|\W|=\kk$ and
 $\alloc(v_i)=1$ if and only if $\W\cap \app_i\neq\emptyset$.
 Thus $\W$ witnesses feasibility of $\alloc$ for $\CC$.
 
 By \cref{prop:CCJR}, it suffices to show that $\W$ satisfies JR.
 Assume for contradiction that $\W$ violates JR. Then there exists a cohesive coalition~$\vvv^\star\subseteq \vvv$ with $|\vvv^\star|\ge \lceil \frac{|\vvv|}{\kk} \rceil$ and some alternative
 $\alt^\star\in \bigcap_{v_i\in \vvv^\star}\app_i$ such that
 $\W\cap \bigcup_{v_i\in \vvv^\star}\app_i=\emptyset$.
 In particular, no voter in $\vvv^\star$ is ever removed from $\vvv'$ during the algorithm,
 because a voter is removed only when she approves a chosen alternative.

 Consider an arbitrary iteration~$\ell\in\{1,\dots,\kk\}$ of the \textbf{while}-loop (so $|\W|<\kk$).
 At the beginning of iteration~$\ell$, all voters in $\vvv^\star$ are still in $\vvv'$,
 and all of them approve~$\alt^\star$. Hence $\alt^\star$ is approved by at least
 $|\vvv^\star|\ge |\vvv|/\kk$ voters in the current~$\vvv'$.
 Since the algorithm chooses an alternative $\alt$ maximizing $|\VV(\alt)|$ over $\aaa\setminus \W$ and since $\alt^{\star}$ was never chosen,
 it follows that the selected $\alt$ in iteration~$\ell$ is approved by at least $\lceil \frac{|\vvv|}{\kk}\rceil$ currently non-deleted voters, and these voters are removed from $\vvv'$ at the end of the iteration.
 This implies that at the end of the~$k^{\text{th}}$ iteration, all voters are removed, a contradiction to the assumption that $|\vvv^{\star}| \ge \lceil \frac{|\vvv|}{\kk} \rceil$ voters remain at the end of the algorithm. %
\end{proof}

We illustrate how to execute \cref{alg:NTU-core-CC}:

\begin{example}\label{ex:NTU-core-CC}
  Consider the instance~$(\ppp,\kk)$ from~\cref{ex:NTUgames} with $\kk=2$. %
  Initially, %
  alternative~$\alt_4$ has the most approvers~$|\VV(\alt_4)|=4$,
  so the algorithm selects~$\alt_4$ (line~\ref{alg:NTU-core-CC:argmax}).
  It then removes the covered voters~$\VV(\alt_4)=\{v_2,v_3,v_5,v_6\}$ from consideration and sets $\alloc(v_i) = 1$ for each of them (line \ref{alg:NTU-core-CC:W}--\ref{alg:NTU-core-CC:alloc}). 
  Among the remaining voters~$\{v_1, v_4\}$, alternatives $\alt_1$ and $\alt_2$ each have one approver,
  so the algorithm selects one of them --- say $\alt_1$ --- completing the committee $\W=\{\alt_1,\alt_4\}$.
  The resulting utility vector is $\alloc=(1,1,1,0,1,1)$;
  analogously, choosing $\alt_2$ yields $(0,1,1,1,1,1)$.
  In either case, exactly one voter receives utility~$0$, and the output lies in the \NTUcore{\CC} by \cref{thm:CCNTUCE}.
  Note that neither committee is \CC-winning: the unique winning committee $\W'=\{\alt_1,\alt_2\}$
  achieves $\score{\CC}(v,\W')=1$~for all $v\in \vvv$.
  Since a \CC-winning committee satisfies JR, by \cref{cor:JR->NTU-CC-Core}, $\W'$ also induces a utility vector in the core. 
\end{example}

\section{Computational Issues}\label{sec:complexity}%

The preceding results settle existence for \TU-\AV, \TU-\SAV, and \NTU-\CC,
and leave \AV/\SAV/\PAV\ under \NTU\ open.
A natural follow-up is computational: even when stable outcomes exist, can they be found or verified efficiently?
This question requires tools from computational complexity, which we briefly recall before stating our results.

\begin{definition}[Computational complexity classes: \DP, \ThetaTwoP, \SigmaTwoP]\label{def:BeyondNP}
In order to define these classes, we utilize so-called \myemph{\NP-oracles}.
An \myemph{\NP-oracle} is a black box algorithm that can decide in \emph{constant} time whether a given instance of a specific \NP\ problem is a yes-instance or no-instance~\cite[Chapter 3]{arora2009computational}. Note that, the restriction to a specific \NP\ problem is not actually necessary, as any \NP-hard problem can be used to model any other \NP-hard problem.

The complexity class \myemph{$\DP$}~\cite{PY84} (D standing for ``difference'') contains all problems that are the differences of two \NP\ problems.
In other words, a problem $\Pi$ is contained in $\DP$ if there exists a mapping $f=(f_1,f_2)$ mapping the problem to two problems $\Pi_1,\Pi_2$ in \NP\ such that an instance~$I$ of $\Pi$ is a yes-instance if and only if $f_1(I)$ is a yes-instance of $\Pi_1$ and $f_2(I)$ is a no-instance of~$\Pi_2$~\cite[Chapter 17]{papadimitrioubook}.
The complexity class~\myemph{$\ThetaTwoP$} (aka.\ $\PP^{\NP[\log]}$ and $\PP^{\NP}_{\mid\mid}$) contains all problems that can be solved by a \emph{polynomial-time deterministic algorithm} using \emph{logarithmically} many \NP-oracle calls~\cite{WagnerBoundeQueries1990,HSV05Kemeny}.
The complexity class \myemph{$\SigmaTwoP$} contains all problems that can be solved by a \emph{polynomial-time non-deterministic algorithm} with \emph{polynomially} many \NP-oracle calls~\cite[Chapter 5]{arora2009computational}. Note that it holds that
{$\PP\subseteq \NP \subseteq \DP \subseteq \ThetaTwoP  \subseteq\SigmaTwoP$.} It is generally assumed that all inclusion relations are strict. 
\end{definition}

For readers less familiar with the polynomial hierarchy: \DP\ captures `\NP-and-\coNP type' conditions (here: feasibility plus absence of a blocking witness);
\ThetaTwoP\ captures problems solvable in polynomial time with only logarithmically many \NP-oracle calls (here: we can binary-search the optimal grand-coalition score while checking feasibility of the induced inequality system); and \SigmaTwoP\ corresponds to an $\exists\,\forall$-structure typical of `there exists a feasible outcome such that no coalition has a blocking deviation.'
    
With these classes in hand, we can state precise complexity bounds for our core problems.
At a high level, our upper bounds follow a common pattern: non-membership in the core is certified by an explicit blocking coalition together with an admissible committee witnessing strict improvement (or, in the \TU\ case, a violation of a coalition inequality).
The more delicate containments (notably the $\ThetaTwoP$ upper bounds for existence) rely on separating the task of computing the grand-coalition optimum score from the feasibility of the resulting system of core inequalities.
\begin{restatable}[]{theorem}{thmcomplexityupperbounds}\label{thm:complexity-upperbounds}
  \begin{compactenum}
    \item\label{upper:TU-verif} \MWGTUCverif{$\AV$} and \MWGTUCverif{$\SAV$} are in $\coNP$, while \MWGTUCverif{$\PAV$} and   \MWGTUCverif{$\CC$} are in $\DP$.
    \item\label{upper:NTU-verif} For all~$\vr\in \{\AV, \SAV, \PAV, \CC\}$, \MWGNTUCverif{$\vr$} is in $\DP$.
    \item\label{upper:TU-exist}  For all~$\vr\in \{\AV, \SAV, \PAV, \CC\}$, \MWGTUCexist{$\vr$} is in $\ThetaTwoP$,
    and \MWGNTUCexist{$\vr$} is in~$\SigmaTwoP$.
  \end{compactenum}
\end{restatable}
\begin{proof}
  \noindent \emph{Statement~\eqref{upper:TU-verif}:} 
  A utility function~$\alloc$ is in the $\TUcore{\vr}$ if
    \begin{inparaenum}[(a)]
      \item\label{coNP:AV-cond1}there exists $W\subseteq\aaa$ with $|\W|=k$ such that
      $\sum_{v\in \vvv}\alloc(v)=\scoreTU{\vr}(\vvv,\W)$, and
      \item\label{coNP:AV-cond2}there do not exist a coalition~$\vvv'\subseteq \vvv$ and a witness committee~$\W'$ of size
      $|\W'|\le \share$ such that
      $\sum_{v\in \vvv'}\alloc(v) < \scoreTU{\vr}(\vvv',\W')$.
    \end{inparaenum}

    Note that condition~\eqref{coNP:AV-cond1} together with condition~\eqref{coNP:AV-cond2} is equivalent to checking whether
   (c) $\sum_{v\in \vvv}\alloc(v)=\cval{\vr}(\vvv)$: if $\sum_{v\in \vvv}\alloc(v)<\cval{\vr}(\vvv)$, then the grand coalition~$\vvv$ blocks.

     For $\vr\in\{\AV,\SAV\}$, a pair $(\vvv',\W')$ serves as a polynomial-size certificate for violating condition~\eqref{coNP:AV-cond2},
    and can be verified in polynomial time and condition~(c) can be checked in polynomial time.
    Hence the complement of \MWGTUCverif{$\vr$} is in $\NP$, and therefore
    \MWGTUCverif{$\vr$} is in $\coNP$. 

    For $\vr\in\{\CC,\PAV\}$, \MWGTUCverif{$\vr$} is in $\DP$.
    Indeed, verifying condition~(\ref{coNP:AV-cond1}) is in $\NP$, with a polynomial-size certificate given by a committee~$W$
    witnessing $\sum_{v\in \vvv}\alloc(v)=\scoreTU{\vr}(\vvv,\W)$.
    Verifying condition~(\ref{coNP:AV-cond2}) is in $\coNP$, since a polynomial-size certificate for its violation is a pair $(\vvv',\W')$, where $\vvv'$ is $\TUblocking{\vr}$ and $\W'$ witnesses it.
    By Papadimitriou~\cite{papadimitrioubook}, a language $P$ is in $\DP$ if there exist $\NP$-languages $Q_1,Q_2$ such that
    $P=Q_1\cap \overline{Q_2}$.
    Taking $Q_1$ to be condition~\eqref{coNP:AV-cond1} and $Q_2$ to be the complement of condition~\eqref{coNP:AV-cond2} yields
    \MWGTUCverif{$\vr$}$\in\DP$ for $\vr\in\{\CC,\PAV\}$.

    \smallskip
    \noindent \emph{Statement~\eqref{upper:NTU-verif}:}
    A utility function~$\alloc$ is in the $\NTUcore{\vr}$ if
    \begin{inparaenum}[(a)]
      \item\label{NTUcoreverif:cond1} there exists $W\subseteq\aaa$ with $|\W|=k$ such that
      $\alloc(v_i)=\scoreNTU{\vr}(v_i,W)$ for all $v_i\in\vvv$, and
      \item\label{NTUcoreverif:cond2} there do not exist a coalition~$\vvv'\subseteq \vvv$ and a witness committee~$\W'$ of size
      $|\W'|\le \share$ such that
      $\alloc(v_i) < \scoreNTU{\vr}(v_i,\W')$ for all $v_i\in\vvv'$.
    \end{inparaenum}

    For all~$\vr\in\{\AV,\SAV,\PAV,\CC\}$, \MWGNTUCverif{$\vr$} is in $\DP$:
    verifying condition~\eqref{NTUcoreverif:cond1} is in $\NP$ with certificate $W$, and violating condition~\eqref{NTUcoreverif:cond2} is
    in $\NP$ with certificate $(\vvv',\W')$. The same $\DP$ characterization argument as in Statement~\ref{upper:TU-verif} applies.

                 \smallskip
    \noindent \emph{Statement~\eqref{upper:TU-exist}:}
     We first prove $\ThetaTwoP$-containment for \MWGTUCexist{$\vr$} for $\vr\in\{\AV,\SAV,\PAV,\CC\}$ by adapting the approach of Chen et al.~\cite[Theorem 2]{PCOGpaper}.
     The key idea is to
     \begin{compactenum}[(i)]
       \item first compute the score of a winning committee~$s^\star=\scoreTU{\vr}(\vvv)$ using logarithmically many $\NP$-oracle calls via binary search, and then
       \item solve a \emph{hinted} version of the \CE\ problem that additionally receives~$s^\star$ as part of the input.
     \end{compactenum}
     
     \mypara{Step 1: computing $s^\star$ with logarithmically many $\NP$-queries.}
     Let $M$ be a polynomially bounded upper bound on $s^\star$ (e.g., $M=|\vvv|\kk$ for $\AV$ and $M=|\vvv|$ for $\CC$; for $\PAV$ we scale scores to integers).
     For any threshold $t\in\{0,\dots,M\}$, the committee score problem 
     \[
       \exists \W\subseteq\aaa,\ |\W|=\kk\colon \scoreTU{\vr}(\vvv,\W)\ge t.
     \]
     is in $\NP$ (certificate: $\W$), hence we can compute $s^\star$ by binary search using $O(\log M)$ $\NP$-oracle calls.
     
     \smallskip
     \mypara{Step 2: A hinted version of core existence is in $\coNP$ via ellipsoid-based certificate.}

     Define the hinted problem:
     \decprob{Hinted\text{-}\MWGTUCexist{$\vr$}}
     {An \mw-instance~$(\ppp,\kk)$ and an integer~$s^{\star}$ such that $s^{\star}= \scoreTU{\vr}(\vvv)$.}
     {Is the $\TUcore{\vr}$ of $(\ppp,\kk)$ non-empty?}

   \noindent  We express non-emptiness of the $\TUcore{\vr}$ as feasibility of a linear program (LP) with variables $\{\alloc(v)\}_{v\in\vvv}$:
   \begin{align}
     \sum_{v\in\vvv}\alloc(v) &= s^\star, \label{eq:hinted-eff}\\
     \sum_{v\in\vvv'}\alloc(v) &\ge \scoreTU{\vr}(\vvv',\W'), 
                                 \qquad\text{for all } \vvv'\subseteq\vvv \text{ and all } \W'\subseteq\aaa \text{ with } |\W'|\le \share. \label{eq:hinted-block}
   \end{align}
   This LP has exponentially many constraints of type~\eqref{eq:hinted-block}.

   \begin{clm}\label{cl:hinted-conp}
     \textsc{Hinted\text{-}\MWGTUCexist{$\vr$}} is in $\coNP$ for $\vr\in\{\AV,\SAV,\PAV,\CC\}$.
   \end{clm}
   
   \begin{claimproof}{cl:hinted-conp}
     We show that the complement of the problem (i.e., the hinted core is empty) has a polynomial-size certificate verifiable in polynomial time,
     based on the ellipsoid method~\cite{GLS1988}. 
     Since $\sum_{v\in\vvv}\alloc(v)=s^\star\le M$ and $\alloc(v)\ge 0$, we may restrict to $\alloc\in[0,M]^n$ (where $n\coloneq |\vvv|$),
     and thus start the ellipsoid method from any explicit ellipsoid containing this box.
     
     If the hinted core is empty, then the LP is infeasible. 
     Now, imagine we run the ellipsoid method with an arbitrary separation oracle~$Q$ (it need not be efficient) on the LP. 
     Note that since the hinted core is empty, for each queried center point~(i.e., $\alloc$), the oracle~$Q$ will answer that either \eqref{eq:hinted-eff} is violated or
     there exists a blocking witness~$(\vvv', \W')$ violating~\eqref{eq:hinted-block}.
     Because for a no-instance, the ellipsoid method makes only polynomially many $Q$-oracle calls~\cite{Schrijver99book} (in $n$ and the encoding length of $M$, but independent of the number of constraints), we obtain a polynomial-length \emph{ellipsoid certificate} $\hat{Q}$ that lists, for each queried center point~(i.e., $\alloc$), a violated constraint~$(\vvv', \W')$.
     Each listed violated constraint is polynomial-time checkable:
     If $\vvv'=\vvv$, then we check \eqref{eq:hinted-eff} is violated;
     otherwise we check in polynomial time whether $(\vvv',\W')$ is blocking, i.e., 
     $|\W'|\le \share$ and $\sum_{v\in\vvv'}\alloc(v) < \scoreTU{\vr}(\vvv',\W')$.
     Hence, infeasibility (i.e., empty hinted core) is in $\NP$, and therefore the hinted non-emptiness problem is in $\coNP$.
   \end{claimproof}

   \mypara{Step 3: concluding $\ThetaTwoP$.}
     After we have determined the optimal score $s^\star=\scoreTU{\vr}(\vvv)$ using logarithmically many $\NP$-oracle calls,
     we make one additional $\NP$-oracle call to the complement of \textsc{Hinted\text{-}\MWGTUCexist{$\vr$}} on input $(\ppp,\kk,s^\star)$, which is in $\NP$ by Claim~\ref{cl:hinted-conp}. 
     If that call answers ``yes'' (the hinted core is empty), we reject; otherwise we accept.
     In total we use $O(\log M)+1$ many $\NP$-oracle calls, and therefore \MWGTUCexist{$\vr$} is in $\ThetaTwoP$.
     \medskip
     
     It remains to show $\SigmaTwoP$-containment for \MWGNTUCexist{$\vr$}.
     We do so by giving a non-deterministic polynomial-time algorithm with access to an \NP-oracle.
     
     Let $I=(\ppp,\kk)$ be an instance.
     Notice that $I$ is a yes-instance if and only if there exists a size-$\kk$ committee $\W$ such that the induced allocation
     \[
       \alloc(v_i)\;=\;\scoreNTU{\vr}(v_i,\W)\qquad\text{for all } v_i\in\vvv
     \]
     is in the core, i.e., admits no blocking voter coalition.
     
     Our algorithm proceeds as follows.
     \begin{enumerate}
       \item Non-deterministically guess a size-$\kk$ committee $\W$, and compute the induced allocation $\alloc$ by setting
       $\alloc(v_i)=\scoreNTU{\vr}(v_i,\W)$ for all $v_i\in\vvv$.
       \item Query an \NP-oracle $Q_{\mathrm{block}}$ that decides the complement of \CV\ on~$(I,\alloc)$, i.e., whether $\alloc$ is blocked by some coalition. 
       Accept if and only if $Q_{\mathrm{block}}$ answers \textsc{no}.
     \end{enumerate}
     
     \mypara{Oracle.}
     The oracle $Q_{\mathrm{block}}$ is an \NP-oracle: a {yes}-certificate consists of a coalition~$\vvv'\subseteq \vvv$,
     and the additional witness required by the definition of blocking under $\vr$,
     and the verification that $\vvv'$ blocks the induced allocation can be carried out in polynomial time under~$\vr$.

     \mypara{Correctness.}
     If $I$ is a yes-instance, then there exists a size-$\kk$ committee $\W^\star$ whose induced allocation $\alloc^\star$ is in the core.
     On the computation branch that guesses $\W^\star$, there is no blocking coalition for $\alloc^\star$,
     hence $(\ppp,\W^\star)\notin Q_{\mathrm{block}}$ and the oracle answers \textsc{no}; thus the algorithm accepts.

     Conversely, if $I$ is a no-instance, then for every size-$\kk$ committee $\W$ the induced allocation is not in the core,
     so there exists a blocking coalition. Hence $(\ppp,\W)\in L_{\mathrm{block}}$ for every guessed $\W$,
     the oracle answers {yes} on every branch, and the algorithm rejects.
     
     \mypara{Running time.}
     Guessing $\W$ and computing $\alloc$ take polynomial time, and we make a single call to the \NP-oracle.
     Therefore $\MWGNTUCexist{\vr}\in \NP^{\NP}=\SigmaTwoP$.
    \end{proof}

We next show that most of these upper bounds are essentially tight.
The reductions isolate the source of difficulty in our setting: blocking combines proportional feasibility constraints with an optimization problem over committees, and the interaction can encode both existential and universal structure (e.g., ``there exists a blocking coalition'' versus ``for all coalitions no blocking exists'').
 As a result, core membership problems often land in classes such as \DP, and core non-emptiness problems can require $\ThetaTwoP$- or $\SigmaTwoP$-style access to an \NP-oracle.

\subsection{Complexity for the \TU\ model}\label{sub:complexity-TU}
We begin with the \TU\ model.
In this case, blocking reduces to a \emph{scalar} comparison: a coalition blocks a candidate utility function if it can achieve strictly larger \emph{total} score than the coalition's achievable score.
This makes \CE\ resemble checking a system of core inequalities, but the coalition values are themselves defined by an optimization problem over admissible committees.
Accordingly, the complexity depends strongly on the underlying rule: for \AV\ and \SAV, \totsepity\ yields substantially more tractability,
while for \CC\ and \PAV\ the induced optimization structure leads to hardness because computing a winning committee under both rules is already \NPhh~\cite{AS1999,lu2011budgeted,procaccia2008complexity,Aziz15}.

\begin{restatable}[]{theorem}{thmTUCE}\label{thm:TU-CE}
  \MWGTUCexist{\AV} and \MWGTUCexist{\SAV} are in $\PP$, while \MWGTUCexist{\CC} and \MWGTUCexist{\PAV} are \NPhh.
\end{restatable}
\begin{proof}
  The statement for \AV\ and \SAV\ follows directly from \cref{thm:AV-TU-Core-Non-Empty}, as the core is always non-empty for these rules.
  The \NPhh{ness} for \CC\ and \PAV\ comes from the fact that computing a corresponding winning committee is already \NPhh.
  For the sake of completeness, we give the corresponding reduction, which works for both.
  We reduce from the \NPcc\ problem~\textsc{RX3C}~\cite{RX3C}.

  \decprob{Restricted Exact Cover by 3-Sets (RX3C)}
  {A universe of $3\nn$ elements~$X=\{x_1,\dots,x_{3\nn}\}$, a family $\mathcal{S}=\{S_1,\ldots,S_{\mm}\}$ with $S_j\subseteq X$ and $|S_j|=3$ for all $j\in[\mm]$ such that every element appears in exactly three 3-sets in~$\mathcal{S}$, i.e., $\mm=3\nn$.}
  {Does $(X, \mathcal{S})$ admit an \emph{exact cover} (i.e., a subcollection~$\mathcal{S}'$ of cardinality~$\nn$ such that $\cup_{S_j\in \mathcal{S}'}S_j = X$)?}
  
  Let $I=(X,\mathcal{S})$ be an instance of \textsc{RX3C}. We create an instance $I'=(\ppp,\kk)$ of \MWGTUCexist{\CC} (resp.\ \MWGTUCexist{\PAV}). Note that the two instances are the same.

\mypara{Profile construction.} Let $\kk\coloneq\nn$.
  \begin{compactitem}[--]
    \item For each set $S_i\in\mathcal{S}$, create an alternative $\alt_i$.
    \item For each element $x_i\in X$, create a voter $v_i$. Their approvals are $\app(v_i)\coloneqq\{\alt_j\mid x_i\in S_j\}$.
  \end{compactitem}
Thus the total number of voters is $|\vvv|=3\nn$.

\mypara{Correctness.} We show that $I$ is a yes-instance of \textsc{RX3C} if and only if the core is non-empty (for both \CC\ and \PAV).

For the ``if'' direction, we prove the contrapositive that ``if $I$ is a no-instance, then the core is empty''.
Assume that $I$ is a no-instance, i.e., there does not exist an exact cover.
Then, no size-$\kk$ committee can be approved by all voters, and hence $\cval{\CC}(\vvv)< 3 \nn$.

We next argue that $\cval{\PAV}(\vvv)<3\nn$ as well.
Fix a committee~$\W$ with $|\W|=\kk=\nn$ and write $t_i \coloneqq |\app(v_i)\cap \W|$ for each voter $v_i$.
Since every alternative is approved by exactly three voters, the total number of approval incidences inside $\W$ is $\sum_{i=1}^{3\nn} t_i \;=\; 3|\W| \;=\; 3\nn$.
Because $I$ is a no-instance, $\W$ cannot cover all voters; let $d\ge 1$ be the number of uncovered voters (i.e., voters with $t_i=0$).
Then exactly $3\nn-d$ voters are covered, and assigning one approved winner to each covered voter uses $3\nn-d$ incidences.
Hence there remain precisely $d$ further incidences, each corresponding to an \emph{additional} approved winner for a voter who already has at least one.

Under \PAV, the marginal gain of giving a voter an additional approved winner beyond their first is at most $1/2$
(the second winner contributes $1/2$ and the third contributes $1/3$).
Therefore,
\[
\scoreTU{\PAV}(\vvv,\W)
\;\le\;
(3\nn-d)\cdot 1 \;+\; d\cdot \frac12
\;=\;
3\nn - \frac{d}{2}
\;<\;
3\nn.
\]
Since $\W$ was arbitrary, it follows that $\cval{\PAV}(\vvv)<3\nn$.
 
We utilize the Bondareva--Shapley Theorem~(cf.\ \cref{prop:BS}) and show that there exists a balanced vector that violates~\eqref{eq:BS_ineq}.
For each 3-set~$S_j\in\mathcal{S}$, let $\vvv_j \;\coloneqq\; \{v_i \in \vvv \mid x_i \in S_j\}$
be the corresponding size-$3$ coalition (equivalently, $\vvv_j$ is exactly the set of voters approving~$\alt_j$).
Define a weight vector $\wvector=(\weight_{\vvv'})_{\vvv'\subseteq \vvv}$ by
\[
\weight_{\vvv'} \coloneqq
\begin{cases}
\frac{1}{3}, & \text{if } \vvv'=\vvv_j \text{ for some } j\in[\mm],\\
0, & \text{otherwise.}
\end{cases}
\]
This vector is balanced: Each voter $v_i$ (corresponding to element $x_i$) belongs to exactly three coalitions
$\vvv_j$ because $x_i$ appears in exactly three 3-sets of $\mathcal{S}$, and each such coalition has weight $1/3$.
Therefore $\sum_{\vvv'\ni v_i}\weight_{\vvv'} = 3\cdot \frac13 = 1$, i.e., \eqref{eq:BS_balanced} holds.

Moreover, every coalition $\vvv_j$ with non-zero weight has size $3$ and can select exactly one alternative
(since $\share=\lfloor 3/3\rfloor=1$). Choosing $\alt_j$, all three members approve it, hence
\[
\cval{\CC}(\vvv_j)=3
\qquad\text{and}\qquad
\cval{\PAV}(\vvv_j)=3.
\]
Thus, for $\vr\in\{\CC,\PAV\}$ we obtain
\begin{align*}
\sum_{\vvv'\subseteq \vvv} \weight_{\vvv'}\cdot \cval{\vr}(\vvv')
\;=\;
\sum_{j=1}^{\mm} \frac13\cdot \cval{\vr}(\vvv_j)
\;=\;
\sum_{j=1}^{\mm} \frac13\cdot 3
\;=\; \mm
\;=\; 3\nn.
\end{align*}
Combining this with $\cval{\CC}(\vvv)<3\nn$ and $\cval{\PAV}(\vvv)<3\nn$ shown above, we get
\[
\sum_{\vvv'\subseteq \vvv} \weight_{\vvv'}\cdot \cval{\vr}(\vvv') \;>\; \cval{\vr}(\vvv)
\qquad\text{for }\vr\in\{\CC,\PAV\},
\]
so \eqref{eq:BS_ineq} is violated. By the Bondareva--Shapley Theorem~(\cref{prop:BS}), the core is empty.
This completes the contrapositive and hence the ``if'' direction.

For the ``only if'' direction, assume that $I$ is a yes-instance and let $\mathcal{S}'$ be an exact cover.
Define $\W\coloneqq\{\alt_j\mid S_j\in\mathcal{S}'\}$, so $|\W|=\nn$ and every voter approves exactly one alternative in~$\W$.
Hence $\scoreTU{\CC}(\vvv,\W)=|\vvv|$ and $\scoreTU{\PAV}(\vvv,\W)=|\vvv|$.
Since $\scoreTU{\CC}(\vvv,\W)\le |\vvv|$ for every committee and (by the argument in the ``if'' direction for \PAV) any size-$\nn$ committee that leaves some voter uncovered has \PAV-score strictly less than $|\vvv|$, it follows that
\[
\cval{\CC}(\vvv)=|\vvv|
\qquad\text{and}\qquad
\cval{\PAV}(\vvv)=|\vvv|.
\]
Consider the allocation $\alloc(v)=1$ for all $v\in\vvv$. Then $\sum_{v\in\vvv}\alloc(v)=\cval{\vr}(\vvv)$ for $\vr\in\{\CC,\PAV\}$.

It remains to show that no coalition $\vvv'\subseteq \vvv$ blocks $\alloc$.
Let $\vvv'\subseteq\vvv$ be arbitrary and let $\share\coloneqq \left\lfloor \frac{|\vvv'|\kk}{|\vvv|}\right\rfloor=\left\lfloor\frac{|\vvv'|}{3}\right\rfloor$.
For any witness committee $\W'$ with $|\W'|\le \share$, we have
\[
\scoreTU{\CC}(\vvv',\W') \le 3|\W'|,
\]
because each alternative is approved by exactly three voters.
Moreover,
\[
\scoreTU{\PAV}(\vvv',\W') \le 3|\W'|,
\]
because adding an alternative can increase each approving voter's \PAV-utility by at most $1$, and each alternative is approved by exactly three voters.
Therefore, for $\vr\in\{\CC,\PAV\}$,
\[
\scoreTU{\vr}(\vvv',\W') \le 3|\W'| \le 3\share \le |\vvv'| = \sum_{v\in\vvv'}\alloc(v).
\]
Taking the maximum over all feasible $\W'$ yields $\cval{\vr}(\vvv')\le \sum_{v\in\vvv'}\alloc(v)$, so $\vvv'$ cannot block.
Thus $\alloc\in\TUcore{\CC}$ and $\alloc\in\TUcore{\PAV}$.

This concludes the proof.
\end{proof}

We next turn from \CE\ to \CV. 
In the \TU\ setting, a blocking certificate consists of a coalition together with an admissible committee that achieves a larger total score, turning non-membership into a succinct witness.
This perspective yields the stated upper bounds and guides the hardness reductions below.

The hardness for \AV, and later for \SAV\ via a slight modification, is via reducing from the complement of the \NPcc\ problem \BicliqueL~(\Biclique)~\cite{GareyJohnsonBook}.
Given a graph $G=(U,E)$ and a non-negative integer~$h$,
\Biclique\ asks whether there exist two disjoint subsets~$U_1,U_2\subseteq U$ with $U_1 \cap U_2=\emptyset$ and $|U_1|=|U_2|=h$ such that $\{u,w\}\in E$ for all $u\in U_1$ and all $w\in U_2$. 
  
\newcommand{\vtxnum}{\ensuremath{\hat{n}}}
\newcommand{\richdum}{\ensuremath{o}}
\newcommand{\richdumset}{\ensuremath{O}}
\newcommand{\poordum}{\ensuremath{b}}
\newcommand{\poordumset}{\ensuremath{B}}
\begin{restatable}{theorem}{thmTUCVAV}\label{thm:TU-CV-AV}
   \MWGTUCverif{\AV} is \coNPcc. 
 \end{restatable}
 \begin{proof}
   \CoNP-membership follows from \cref{thm:complexity-upperbounds}\eqref{upper:TU-verif}.
  To show \coNP-hardness, we reduce from \Biclique. %
  Let $I=(G=(N,E),h)$ be an instance with $|N|=\nn$ and assume without loss of generality that $\nn\ge 2h$ and $\nn \ge 2$.
  We construct an instance $I'=(\ppp,\kk,\alloc)$ of \MWGTUCverif{\AV}.

  \mypara{Profile construction.} Let $L\coloneq \nn^2+1$ and $\kk\coloneq L+h$.
  \begin{compactitem}[--]
    \item For each vertex~$u_i\in U$, create a \myemph{vertex-alternative} $\alt_i$; let $\aaa_U\coloneq \{\alt_i\mid u_i\in U\}$.
    \item Create $\kk=L+h$ \myemph{dummy alternatives} $\aaa_D\coloneq\{d_1,\dots,d_{L+h}\}$.
    \item For each vertex~$u_i\in U$, create a \myemph{vertex-voter} $v_i$; let $\vvv_U\coloneq \{v_i\mid u_i\in U\}$.
    Their approvals are $\app(v_i)\coloneq \{\alt_j\mid \{u_i,u_j\}\in E\}$; in particular, $\alt_i\notin \app(v_i)$ since $G$ is simple.
    \item Create $d_1\coloneq L\cdot \nn^2-(\nn-h)$ \myemph{full voters} %
    $\vvv_{\richdumset}:=\{\richdum_1,\dots,\richdum_{d_1}\}$.
    Each full voter~$\richdum_i$ approves all dummy alternatives: $\app(\richdum_i)\coloneq \aaa_D$.
    \item Create $d_2\coloneq h\cdot \nn^2-h$ \myemph{empty voters}
    $\vvv_\poordumset\coloneq\{\poordum_1,\dots,\poordum_{d_2}\}$.
    Each empty voter~$\poordum_i$ approves nothing: $\app(b_i)\coloneq \emptyset$.
  \end{compactitem}
  The total number of voters is
  $|\vvv| = |\vvv_U|+|\vvv_\richdumset|+|\vvv_\poordumset|
    = \nn + (L\nn^2-(\nn-h)) + (h\nn^2-h)
    = \kk\cdot \nn^2$.

  \mypara{The utility function~$\alloc$.}
  Define $\varepsilon \coloneq \frac{\nn h-1}{d_1}$ (note that $\varepsilon<1$ since $\nn^2-1 > \nn-h$ and hence $d_1 > L=\nn^2+1$).
  Set  
  $\alloc(v_i)=h-\frac{1}{\nn}$, $\forall v_i\in \vvv_U$,
  $\alloc(\richdum_i)=L+h-\varepsilon$, $\forall \richdum_i \in \vvv_\richdumset$, and 
  $\alloc(\poordum_i)=0$, $\forall \poordum_i\in \vvv_\poordumset$.

  \smallskip
  This completes the construction of the instance.
  Before we show the correctness, we first observe that the utility function~$\alloc$ is feasible and optimal for~$\vvv$.

  \begin{restatable}[]{claim}{clefficiency}
    \label{cl:efficiency}
    $\sum_{v\in \vvv}\alloc(v)=\cval{\AV}(\vvv)$.
  \end{restatable}
  \begin{claimproof}{cl:efficiency}
    Each dummy alternative~$d\in C_D$ is approved by all~$d_1$ full voters and by no other voters, hence
    $|\VV(d)|=d_1=L\nn^2-(\nn-h)$.
    Each vertex-alternative~$\alt_j\in \aaa_U$ is approved only by some subset of the $\nn$ vertex voters, hence
    $|\VV(\alt_j)|\le \nn < d_1$.
    Therefore a winning committee under \AV\ is $\W^*=\aaa_D$ and
    \[
      \cval{\AV}(\vvv)=\kk\cdot d_1=(L+h)\cdot (L\nn^2-(\nn-h)).
    \]
    On the other hand,
    \[
      \sum_{v\in \vvv}\alloc(v)
      = \sum_{v_i\in \vvv_U}\left(h-\frac{1}{\nn}\right) + \sum_{\richdum_i \in \vvv_\richdumset}(L+h-\varepsilon)
      = (h\nn-1) + d_1(L+h) - d_1\varepsilon
      = d_1(L+h),
    \]
    because $d_1\varepsilon = h\nn-1$ by definition of $\varepsilon$.%
  \end{claimproof}

  \mypara{Correctness.} By \cref{cl:efficiency}, $\alloc$ is feasible for~$\vvv$,  so $\alloc\notin \TUcore{\AV}$ if and only if there exists a blocking coalition.

  For the ``if'' part (biclique implies blocking coalition), assume that $I$ is a yes-instance and let $U_1$ and $U_2$ witness a balanced biclique:
  $U_1\cap U_2=\emptyset$, $|U_1|=|U_2|=h$, and all edges between $U_1$ and $U_2$ are present.
  Let $\vvv'\coloneq \{v_i\mid u_i\in U_1\} \cup \vvv_\poordumset$. Then, $|\vvv'| = h+(h\nn^2-h)=h\nn^2$, hence the coalition share is
  $\share = \lfloor \frac{h\nn^2\cdot \kk}{\kk \nn^2} \rfloor = h$.
  Let $\W' \coloneq \{\alt_j\mid u_j\in U_2\}$, so $|\W'|=h$. Every voter in~$\{v_i\mid u_i\in U_1\}$ approves every alternative in~$\W'$, and hence $\score{\AV}(\vvv',\W') = h\cdot h = h^2$.
  Since empty voters have zero utility, 
  $\sum_{v_i\in \vvv'}\alloc(v_i) = \sum_{u_i \in U_1}(h-1/\nn) = h^2 - h/\nn < h^2$.
  Thus, $\vvv'$ blocks~$\alloc$ and $\alloc \notin \TUcore{\AV}$, i.e., $I'$ is a no-instance.

  For the ``only if'' part (blocking coalition implies biclique), assume that $I'$ is a no-instance and let $\vvv'$ be a blocking coalition and $\W'$ a witness: $|\W'|\le \share$ and $\score{\AV}(\vvv',\W') > \sum_{v\in \vvv'}\alloc(v)$.
  
  \begin{restatable}[]{claim}{clnorichdummy}
    \label{cl:no-richdummy}
    No blocking coalition contains any full voter, i.e.\ $\vvv'\cap \vvv_{\richdumset}=\emptyset$.
  \end{restatable}
  \begin{claimproof}{cl:no-richdummy}
    Towards a contradiction, suppose $x\coloneq |\vvv'\cap \vvv_{\richdumset}|\ge 1$.
    Each dummy alternative has score exactly $x$ in $\vvv_{\richdumset}$,
    while every vertex-alternative has score at most $|\vvv'\cap \vvv_U|\le \nn$.
    Let $t\coloneq \share = \lfloor \frac{|\vvv'|}{\nn^2}\rfloor$.
    
    If $x> \nn$, then any \AV-score-maximizing witness $\W'$ committee may be assumed to contain only dummy alternatives (as otherwise replacing one vertex-alternative with a dummy alternative would increase the score), so
    $\score{\AV}(\vvv',\W') \le t x$.
    
    Blocking would imply $t x > x(L+h-\varepsilon)$ and hence $t > L+h-\varepsilon$.
    Since $t$ is an integer, this yields $t\ge L+h=k$ and therefore $|\vvv'|\ge \nn^2k=|\vvv|$, so $\vvv'=\vvv$.
    But the grand coalition~$\vvv$ does not block~$\alloc$ (\cref{cl:efficiency}).

    If $1 \le x \le \nn$, then the total \AV-score of any committee of size at most $t$ is at most $t\cdot \nn$,
    because every alternative has \AV-score at most $\nn$ in $\vvv'$.
    Yet one full voter contributes utility
    $\alloc(\richdum_i)=L+h-\varepsilon > \nn^2 \ge t\cdot \nn$ (since $|\vvv'|\leq h\cdot\nn^2+2\nn-h$, due to $|\vvv'\cap\vvv_{\richdumset}|\leq\nn$, and therefore $t\leq h+1<\nn$).
    Hence $\vvv'$ cannot block~$\alloc$.

    In both cases we obtain a contradiction, so $\vvv'\cap \vvv_{\richdumset}=\emptyset$.%
  \end{claimproof}

    \begin{restatable}[]{claim}{claddemptydummy}
      \label{cl:add-emptydummy}
      If $\vvv'$ blocks~$\alloc$, then $\vvv'\cup \vvv_{\poordumset}$ also blocks~$\alloc$ (with the same witness committee).
    \end{restatable}
 \begin{claimproof}{cl:add-emptydummy}
    Empty voters approve no alternatives and have utility $0$,
    hence adding them does not change the \AV-score of~$\W'$ nor the utility sum~$\sum_{v\in \vvv'}\alloc(v)$,
    while it can only increase the share~$\share$. 
     \end{claimproof}

By Claims~\ref{cl:no-richdummy} and~\ref{cl:add-emptydummy}, without loss of generality, we may assume $\vvv_\poordumset\subseteq \vvv'$ and $\vvv'\cap \vvv_{\richdumset}=\emptyset$.
Let $h'\coloneq |\vvv'\cap \vvv_U|$ be the number of vertex-voters in $\vvv'$. We aim to show that $h'\ge h$ and that the vertex-voters in the blocking coalition and vertex-alternatives in the witness correspond to a balanced biclique of size at least~$h$ on each side.
First,
$|\vvv'| = |\vvv_{\poordumset}| + h' = (h\nn^2-h)+h'$.
Since $0\le h'\le \nn < \nn^2$, it follows that if $h' < h$, then  $\lfloor \frac{|S|}{\nn^2}\rfloor= h-1$;
otherwise  $\lfloor \frac{|S|}{\nn^2}\rfloor= h$. 

If $h'<h$, then $|\W'|\le h-1$ and every alternative has \AV-score at most $h'$ in $\vvv'$, hence
$\score{\AV}(\vvv', \W')\le (h-1)h'$.
But the coalition utility is
$\sum_{v_i\in \vvv'}\alloc(v_i)=h'\left(h-\frac{1}{\nn}\right)>(h-1)h'$,
a contradiction.
Therefore $h' \ge h$ and the budget forces~$|\W'|\le h$. 

Now note that for every size-at-most-$h$ committee~$\W''$ we have
$\score{\AV}(\vvv', \W'')\le hh'$ since each chosen alternative has score at most $h'$.
Since $\vvv'$ is blocking we get

{\centering
 $\score{\AV}(\vvv', \W') > h'(h-\frac{1}{\nn}) = h'h-\frac{h'}{\nn}$.
 \par}

The left-hand side is an integer and the right-hand side is strictly less than $h'h$, so
$\score{\AV}(\vvv', \W') \ge h'h$.
Thus $\score{\AV}(\vvv', \W')=h'h$ and $|\W'|=h$,
and consequently every alternative in $\W'$ has score exactly $h'$ in $\vvv'$, i.e.\ it is approved by \emph{all} $h'$ vertex voters in $\vvv'$.

Let $U_1$ be any subset of size-$h$ of the vertices corresponding to $\vvv'\cap \vvv_U$, and let $U_2$ be the set of size-$h$ of vertices corresponding to $\W'$.
Then every vertex in $U_1$ is adjacent to every vertex in $U_2$.
Moreover $U_1\cap U_2=\emptyset$ because a voter~$v_i$ does not approve $\alt_i$ (no self-loop),
so if $\alt_i\in \W'$ and $v_i\in \vvv'$, then $\alt_i$ would not be approved by all voters in $\vvv'\cap \vvv_U$,
contradicting that $\alt_i$ has score $h'$.
Hence $(U_1,U_2)$ witnesses that $I$ is a yes-instance of \Biclique.

We have shown that $I$ is a yes-instance if and only if $I'$ is a no-instance.
Therefore \MWGTUCverif{\AV} is \coNPhh.
Together with \coNP-containment in \cref{thm:complexity-upperbounds}\eqref{upper:TU-verif}, the problem is \coNPcc.
\end{proof}
The following result can be obtained using a similar reduction to \cref{thm:TU-CV-AV}. Intuitively, we circumvent the fact that the utility that a voter $v$ obtains from an approved alternative is $\frac{1}{|\app(v)|}$ rather than $1$ by making sure that all voters that approve at least one alternative approve the exact same number of alternatives by adding dummy alternatives that are approved by one voter each.
\begin{restatable}{theorem}{thmTUCVSAV}\label{thm:TU-CV-SAV}
	\MWGTUCverif{\SAV} is \coNPcc. 
\end{restatable}
	\begin{proof}
		\newcommand{\privtxdum}[2]{\ensuremath{p^{#1}_{#2}}}
		\newcommand{\privvtxset}{\ensuremath{P}}
		\CoNP-membership follows from \cref{thm:complexity-upperbounds}\eqref{upper:TU-verif}.
		To show \coNP-hardness, we reduce from \Biclique. %
		Let $I=(G=(N,E),h)$ be an instance with $|N|=\nn$ and assume without loss of generality that $\nn\ge 2h$ and $\nn \ge 2$.
		We construct an instance $I'=(\ppp,\kk,\alloc)$ of \MWGTUCverif{\SAV}.
		
		\mypara{Profile construction.} Let $L\coloneq \nn^2+1$ and $\kk\coloneq L+h$.
		\begin{compactitem}[--]
			\item For each vertex~$u_i\in U$, create a \myemph{vertex-alternative} $\alt_i$; let $\aaa_U\coloneq \{\alt_i\mid u_i\in U\}$.
			\item For each vertex~$u_i\in U$, create $\kk-\deg(u_i)=L+h-\deg(u_i)$ \myemph{private alternatives} $\privtxdum{i}{1},\ldots,\privtxdum{i}{\kk-\deg(u_i)}$; let $\privvtxset\coloneq \{\privtxdum{i}{j}\mid u_i\in U, j\in[\kk-\deg(u_i)]\}$.
			\item Create $\kk=L+h$ \myemph{dummy alternatives} $\aaa_D\coloneq\{d_1,\dots,d_{L+h}\}$.
			\item For each vertex~$u_i\in U$, create a \myemph{vertex-voter} $v_i$; let $\vvv_U\coloneq \{v_i\mid u_i\in U\}$.
			Their approvals are $\app(v_i)\coloneq \{\alt_j\mid \{u_i,u_j\}\in E\}\cup\{\privtxdum{i}{1},\ldots,\privtxdum{i}{\kk-\deg(u_i)}\}$; in particular, $\alt_i\notin \app(v_i)$ since $G$ is simple.
			\item Create $d_1\coloneq L\cdot \nn^2-(\nn-h)$ \myemph{full voters} %
			$\vvv_{\richdumset}:=\{\richdum_1,\dots,\richdum_{d_1}\}$.
			Each full voter~$\richdum_i$ approves all dummy alternatives: $\app(\richdum_i)\coloneq \aaa_D$.
			\item Create $d_2\coloneq h\cdot \nn^2-h$ \myemph{empty voters}
			$\vvv_\poordumset\coloneq\{\poordum_1,\dots,\poordum_{d_2}\}$.
			Each empty voter~$\poordum_i$ approves nothing: $\app(b_i)\coloneq \emptyset$.
		\end{compactitem}
		The total number of voters is
		$|\vvv| = |\vvv_U|+|\vvv_\richdumset|+|\vvv_\poordumset|
		= \nn + (L\nn^2-(\nn-h)) + (h\nn^2-h)
		= \kk\cdot \nn^2$.
		
		\mypara{The utility function~$\alloc$.}
		Define $\varepsilon \coloneq \frac{\nn h-1}{d_1}$ (note that $\varepsilon<1$ since $\nn^2-1 > \nn-h$ and hence $d_1 > L=\nn^2+1$).
		Set  
		$\alloc(v_i)=\frac{h-\frac{1}{\nn}}{\kk}$, $\forall v_i\in \vvv_U$,
		$\alloc(\richdum_i)=\frac{L+h-\varepsilon}{\kk}$, $\forall \richdum_i \in \vvv_\richdumset$, and 
		$\alloc(\poordum_i)=0$, $\forall \poordum_i\in \vvv_\poordumset$.
		
		\smallskip
		This completes the construction of the instance. An attentive reader may have noticed that the construction of the profile and utility function is very similar to the one in \cref{thm:TU-CV-AV}. The only differences lie in the fact that we added the private alternatives and divided the utility of each voter by $\kk$. Due to this the remainder of the proof will follow the same steps as the proof of \cref{thm:TU-CV-AV}.
		
		We can immediately observe the following from the construction:
		\begin{observation}\label{obs:TUCVAVTUCVSAV}
			Each voter approves either $0$ or $\kk$ alternatives. Therefore it holds that $\cval{\SAV}(\vvv')=\frac{\cval{\AV}(\vvv')}{\kk}$ and $\scoreTU{\SAV}(\vvv',\W')=\frac{\scoreTU{\AV}(\vvv',\W')}{\kk}$ for all $\vvv'\subseteq\vvv$ and $\W'\subseteq\aaa$.
		\end{observation}
		Next, we observe that the utility function~$\alloc$ is feasible and optimal for~$\vvv$.
		
		\begin{restatable}[]{claim}{SAVclefficiency}
				\label{cl:SAVefficiency}
				$\sum_{v\in \vvv}\alloc(v)=\cval{\SAV}(\vvv)$.
			\end{restatable}
			\begin{claimproof}{cl:SAVefficiency}
				Each dummy alternative~$d\in C_D$ is approved by all~$d_1$ full voters and by no other voters, hence
				$|\VV(d)|=d_1=L\nn^2-(\nn-h)$.
				Each vertex-alternative~$\alt_j\in \aaa_U$ is approved only by some subset of the $\nn$ vertex voters, hence
				$|\VV(\alt_j)|\le \nn < d_1$. Each private alternative~$\privtxdum{i}{j}$ is approved only by one voter, hence $|\VV(\privtxdum{i}{j})|=1$.
				Therefore a winning committee under \SAV\ is $\W^*=\aaa_D$ and by \cref{obs:TUCVAVTUCVSAV}
				\[
				\cval{\SAV}(\vvv)=\frac{\cval{\AV}(\vvv)}{\kk}=\frac{\kk\cdot d_1}{\kk}=d_1=(L\nn^2-(\nn-h)).
				\]
				On the other hand,
				\[
				\sum_{v\in \vvv}\alloc(v)
				= \sum_{v_i\in \vvv_U}\left(\frac{h-\frac{1}{\nn}}{\kk}\right) + \sum_{\richdum_i \in \vvv_\richdumset}\frac{(L+h-\varepsilon)}{\kk}
				= \frac{(h\nn-1) + d_1(L+h) - d_1\varepsilon}{\kk}
				= d_1,
				\]
				because $d_1\varepsilon = h\nn-1$ by definition of $\varepsilon$.%
			\end{claimproof}
			
			\mypara{Correctness.} By \cref{cl:SAVefficiency}, $\alloc$ is feasible for~$\vvv$,  so $\alloc\notin \TUcore{\SAV}$ if and only if there exists a blocking coalition.
			
			For the ``if'' part (biclique implies blocking coalition), assume that $I$ is a yes-instance and let $U_1$ and $U_2$ witness a balanced biclique:
			$U_1\cap U_2=\emptyset$, $|U_1|=|U_2|=h$, and all edges between $U_1$ and $U_2$ are present.
			Let $\vvv'\coloneq \{v_i\mid u_i\in U_1\} \cup \vvv_\poordumset$. Then, $|\vvv'| = h+(h\nn^2-h)=h\nn^2$, hence the coalition share is
			$\share = \lfloor \frac{h\nn^2\cdot \kk}{\kk \nn^2} \rfloor = h$.
			Let $\W' \coloneq \{\alt_j\mid u_j\in U_2\}$, so $|\W'|=h$. Every voter in~$\{v_i\mid u_i\in U_1\}$ approves every alternative in~$\W'$, and hence $\score{\SAV}(\vvv',\W')=\frac{\score{\AV}(\vvv',\W')}{\kk} = \frac{h^2}{\kk}$.
			Since empty voters have zero utility, 
			$\sum_{v_i\in \vvv'}\alloc(v_i) = \sum_{u_i \in U_1}\frac{h-1/\nn}{\kk} = \frac{h^2 - h/\nn}{\kk} < \frac{h^2}{\kk}$.
			Thus, $\vvv'$ blocks~$\alloc$ and $\alloc \notin \TUcore{\SAV}$, i.e., $I'$ is a no-instance.
			
			For the ``only if'' part (blocking coalition implies biclique), assume that $I'$ is a no-instance and let $\vvv'$ be a blocking coalition and $\W'$ a witness: $|\W'|\le \share$ and $\score{\SAV}(\vvv',\W') > \sum_{v\in \vvv'}\alloc(v)$.
			
			\begin{restatable}[]{claim}{SAVclnorichdummy}
					\label{cl:SAVno-richdummy}
					No blocking coalition contains any full voter, i.e.\ $\vvv'\cap \vvv_{\richdumset}=\emptyset$.
				\end{restatable}
				\begin{claimproof}{cl:SAVno-richdummy}
					Towards a contradiction, suppose $x\coloneq |\vvv'\cap \vvv_{\richdumset}|\ge 1$.
					Each dummy alternative has score exactly $\frac{x}{\kk}$ in $\vvv_{\richdumset}$,
					while every vertex-alternative has score at most $\frac{|\vvv'\cap \vvv_U|}{\kk}\le \frac{\nn}{\kk}$ and every private alternative has score at most $\frac{1}{\kk}$.
					Let $t\coloneq \share = \lfloor \frac{|\vvv'|}{\nn^2}\rfloor$.
					
					If $x> \nn$, then any \SAV-score-maximizing witness $\W'$ committee may be assumed to contain only dummy alternatives (as otherwise replacing one vertex-alternative or private alternative with a dummy alternative would increase the score), so
					$\score{\SAV}(\vvv',\W') \le \frac{t x}{\kk}$.
					
					Blocking would imply $\frac{t x}{\kk} > \frac{x(L+h-\varepsilon)}{\kk}$ and hence $t > L+h-\varepsilon$.
					Since $t$ is an integer, this yields $t\ge L+h=k$ and therefore $|\vvv'|\ge \nn^2k=|\vvv|$, so $\vvv'=\vvv$.
					But the grand coalition~$\vvv$ does not block~$\alloc$ (\cref{cl:SAVefficiency}).
					
					If $1 \le x \le \nn$, then the total \SAV-score of any committee of size at most $t$ is at most $\frac{t\cdot \nn}{\kk}$,
					because every alternative has \SAV-score at most $\frac{\nn}{\kk}$ in $\vvv'$.
					Yet one full voter contributes utility
					$\alloc(\richdum_i)=\frac{L+h-\varepsilon}{\kk} > \frac{\nn^2}{\kk} \ge \frac{t\cdot \nn}{\kk}$ (since $|\vvv'|\leq h\cdot\nn^2+2\nn-h$, due to $|\vvv'\cap\vvv_{\richdumset}|\leq\nn$, and therefore $t\leq h+1<\nn$).
					Hence $\vvv'$ cannot block~$\alloc$.
					
					In both cases we obtain a contradiction, so $\vvv'\cap \vvv_{\richdumset}=\emptyset$.%
				\end{claimproof}

				\begin{restatable}[]{claim}{SAVcladdemptydummy}
						\label{cl:SAVadd-emptydummy}
						If $\vvv'$ blocks~$\alloc$, then $\vvv'\cup \vvv_{\poordumset}$ also blocks~$\alloc$ (with the same witness committee).
					\end{restatable}
					\begin{claimproof}{cl:SAVadd-emptydummy}
						Empty voters approve no alternatives and have utility $0$,
						hence adding them does not change the \SAV-score nor the utility sum~$\sum_{v\in \vvv'}\alloc(v)$,
						while it can only increase the share~$\share$. 
					\end{claimproof}
					
					By Claims~\ref{cl:SAVno-richdummy} and~\ref{cl:SAVadd-emptydummy}, without loss of generality, we may assume $\vvv_\poordumset\subseteq \vvv'$ and $\vvv'\cap \vvv_{\richdumset}=\emptyset$.
					Let $h'\coloneq |\vvv'\cap \vvv_U|$ be the number of vertex-voters in $\vvv'$. We aim to show that $h'\ge h$ and that the vertex-voters in the blocking coalition and vertex-alternatives in the witness correspond to a balanced biclique of size at least~$h$ on each side.
					First,
					$|\vvv'| = |\vvv_{\poordumset}| + h' = (h\nn^2-h)+h'$.
					Since $0\le h'\le \nn < \nn^2$, it follows that if $h' < h$, then  $\lfloor \frac{|S|}{\nn^2}\rfloor= h-1$;
					otherwise  $\lfloor \frac{|S|}{\nn^2}\rfloor= h$. 
					
					If $h'<h$, then $|\W'|\le h-1$ and every alternative has \SAV-score at most $\frac{h'}{\kk}$ in $\vvv'$, hence
					$\score{\SAV}(\vvv', \W')\le \frac{(h-1)h'}{\kk}$.
					But the coalition utility is
					$\sum_{v_i\in \vvv'}\alloc(v_i)=\frac{h'\left(h-\frac{1}{\nn}\right)}{\kk}>\frac{(h-1)h'}{\kk}$,
					a contradiction.
					Therefore $h' \ge h$ and the budget forces~$|\W'|\le h$. 
					
					Now note that for every size-at-most-$h$ committee~$\W''$ we have
					$\score{\SAV}(\vvv', \W'')\le \frac{hh'}{\kk}$ since each chosen alternative has score at most $\frac{h'}{\kk}$.
					Since $\vvv'$ is blocking we get
					
					{\centering
						$\score{\SAV}(\vvv', \W') > \frac{h'(h-\frac{1}{\nn})}{\kk} = \frac{h'h-\frac{h'}{\nn}}{\kk}$.
						\par}
					
					If we multiply everything with $\kk$, then the left-hand side is an integer and the right-hand side is strictly less than $h'h$, so
					$\score{\SAV}(\vvv', \W') \ge \frac{h'h}{\kk}$.
					Thus $\score{\SAV}(\vvv', \W')=\frac{h'h}{\kk}$ and $|\W'|=h$,
					and consequently every alternative in $\W'$ has score exactly $h'$ in $\vvv'$, i.e.\ it is approved by \emph{all} $h'$ vertex voters in $\vvv'$.

					Let $U_1$ be any subset of size-$h$ of the vertices corresponding to $\vvv'\cap \vvv_U$, and let $U_2$ be the set of size-$h$ of vertices corresponding to $\W'$.
					Then every vertex in $U_1$ is adjacent to every vertex in $U_2$.
					Moreover $U_1\cap U_2=\emptyset$ because a voter~$v_i$ does not approve $\alt_i$ (no self-loop),
					so if $\alt_i\in \W'$ and $v_i\in \vvv'$, then $\alt_i$ would not be approved by all voters in $\vvv'\cap \vvv_U$,
					contradicting that $\alt_i$ has score $h'$.
					Hence $(U_1,U_2)$ witnesses that $I$ is a yes-instance of \Biclique.

					We have shown that $I$ is a yes-instance if and only if $I'$ is a no-instance.
					Therefore \MWGTUCverif{\SAV} is \coNPhh.
					Together with \coNP-containment in \cref{thm:complexity-upperbounds}\eqref{upper:TU-verif}, the problem is \coNPcc.
				\end{proof}
We next turn to \CC\ in the \TU\ model.
Unlike \AV\ and \SAV, core membership verification for \CC\ combines
\begin{inparaenum}
  \item an \NP-type witness of a winning committee and 
  \item a \coNP-type requirement that no blocking coalition exists;
\end{inparaenum}
these two constraints naturally place the problem in the class \DP.
We show that this upper bound is tight by a reduction from a \DP-hard problem.
\begin{restatable}[]{theorem}{thmTUCVCC}\label{thm:TU-CV-CC}
\MWGTUCverif{\CC} is \DPcc.
\end{restatable}
\begin{proof}
\newcommand{\varalt}{\ensuremath{{\alt}}}
\newcommand{\negvaralt}{\ensuremath{{\overline{\alt}}}}
\newcommand{\clausealt}{\ensuremath{\alt_\varphi}}
\newcommand{\varvot}{\ensuremath{v^x}}
\newcommand{\clausevot}{\ensuremath{v^\varphi}}
\newcommand{\satassign}{\ensuremath{\sigma}}
\newcommand{\dummyvoter}{\ensuremath{\mathsf{v_d}}}

\DP-membership follows from \cref{thm:complexity-upperbounds}\eqref{upper:TU-verif}.
Therefore it suffices to show \DPhh-hardness. We reduce from \textsc{RX3C-UNSAT}.\footnote{This problem is \DP-complete. Membership is immediate: exact cover (\textsc{RX3C}) is in \NP\ and unsatisfiability~(\textsc{UNSAT}) is in \coNP.
  Hardness follows from \textsc{SAT-UNSAT}\cite{Papadimitriou88facets} by mapping the first formula to an equivalent \textsc{RX3C} instance (via a standard NP-completeness reduction) and keeping the second formula unchanged.}

\smallskip

\decprob{\textsc{RX3C-UNSAT}}
{An \textsc{RX3C}-instance $I_1=(X_1,\mathcal{S})$ and a 3SAT-instance $I_2=(X_2,\varphi)$ such that :
  \begin{inparaenum}
    \item $|X_1|=3\nn_1$,
    \item $\mathcal{S}=\{S_1,\ldots,S_{3\nn_1}\}$ with $S_j\subseteq X_1$ and $|S_j|=3$ for all $j\in[3\nn_1]$,
    \item each element of $X_1$ is contained in exactly three sets of $\mathcal{S}$, and
	\item $\varphi$ is a 3CNF formula containing only variables in~$X_2$.
  \end{inparaenum}}{Is it true that $I_1$ has an \emph{exact cover} (i.e., a subcollection $\mathcal{S}'\subseteq \mathcal{S}$ such that each $x\in X_1$ belongs to exactly one set of $\mathcal{S}'$) and $I_2$ has no satisfying truth assignment?}

Let $I=(I_1,I_2)$ be an instance of \textsc{RX3C-UNSAT}, where $|X_1|=3\nn_1$, $|X_2|=\nn_2$, and $|\varphi|=\mm$, and each element of $X_1$ occurs in exactly three members in $\mathcal{S}$. W.l.o.g.\ assume that $\mm+2$ is divisible by $3$ and that $\mm>1$.

\mypara{Profile construction.}
Set $\kk\coloneqq 2\nn_1+\nn_2+1$ and $L\coloneqq \frac{\mm+2}{3}$.
\begin{compactitem}[--]
  \item For each set~$S_j\in\mathcal{S}$ create a \myemph{set1-alternative}~$\setaltone_j$ and a \myemph{set2-alternative}~$\setalttwo_j$; let
$\aaa_{\mathcal{S}_1}\coloneqq\{\setaltone_1,\ldots,\setaltone_{3\nn_1}\}$ and
$\aaa_{\mathcal{S}_2}\coloneqq\{\setalttwo_1,\ldots,\setalttwo_{3\nn_1}\}$.
\item For each variable~$x_i\in X_2$ create two \myemph{literal-alternatives} $\varalt_{i}$ and $\negvaralt_{i}$; let
$\aaa_X\coloneqq\{\varalt_1,\negvaralt_{1},\ldots,\varalt_{\nn_2},\negvaralt_{\nn_2}\}$.
\item Create a clause alternative~$\clausealt$ and a dummy alternative~$\dummyalt$.

\item For each element~$x_i\in X_1$ create $L$ \myemph{element1-voters}~$\votone_{i,1},\ldots,\votone_{i,L}$ and $L$ \myemph{element2-voters} $\vottwo_{i,1},\ldots,\vottwo_{i,L}$; let $\vvv^1\coloneqq\{\votone_{1,1},\ldots,\votone_{1,L},\ldots,\votone_{3\nn_1,1},\ldots,\votone_{3\nn_1,L}\}$,
$\vvv^2\coloneqq\{\vottwo_{1,1},\ldots,\vottwo_{1,L},\ldots,\vottwo_{3\nn_1,1},\ldots,\vottwo_{3\nn_1,L}\}$.
Their approvals are $\app(\votone_{i,\ell})\coloneqq\{\setaltone_j\mid x_i\in S_j\}$ and $\app(\vottwo_{i,\ell})\coloneqq\{\setalttwo_j\mid x_i\in S_j\}$ for all $\ell\in[L]$.

\item For each variable~$x_i\in X_2$ create $3L$ \myemph{variable-voters} $\varvot_{i,1},\ldots,\varvot_{i,3L}$; let
$\vvv^X\coloneqq\{\varvot_{1,1},\ldots,\varvot_{1,3L}$, $\ldots,\varvot_{\nn_2,1}$, $\ldots,\varvot_{\nn_2,3L}\}$.
Their approvals are $\app(\varvot_{i,\ell})=\{\varalt_{i},\negvaralt_{i}\}$ for $\ell\in[3L]$.

\item For each clause $C_j\in\varphi$ create a \myemph{clause-voter}~$\clausevot_j$; let $\vvv^\varphi\coloneqq\{\clausevot_1,\ldots,\clausevot_{\mm}\}$. Their approvals are
$\app(\clausevot_j)\coloneqq\{\varalt_{i}\mid x_i\in C_j\}\cup\{\negvaralt_{i}\mid \overline{x}_i\in C_j\}\cup\{\clausealt\}$.
\item Create a dummy voter $\dummyvoter$ with $\app(\dummyvoter)\coloneqq\{\dummyalt\}$.
\end{compactitem}
Thus the total number of voters is $|\vvv| =2\cdot (3\nn_1)\cdot L + (3L)\nn_2 + \mm + 1
=(\mm+2)\cdot (2\nn_1+\nn_2)+\mm+1$.

\mypara{The utility function $\alloc$.}
Set $\alloc(v)=1$ for all $v\in\vvv\setminus\{\dummyvoter\}$ and $\alloc(\dummyvoter)=0$.

\mypara{Correctness.}
    We show that $I$ is a yes-instance of \textsc{RX3C-UNSAT}
    if and only if $\alloc\in\TUcore{\CC}$.

For the ``\textbf{if}'' part,  we prove the contra-positive. Assume $I$ is a no-instance of \textsc{RX3C-UNSAT}, meaning that $I_1$ has no exact cover or $I_2$ is satisfiable. 
We distinguish two cases.

\smallskip
\myparait{Case 1: $I_1$ has no exact cover.}
We show that $\score{\CC}(\vvv)=\cval{\CC}(\vvv)<\sum_{v\in\vvv}\alloc(v)=|\vvv|-1$, hence $\alloc$ is infeasible.

Let $\W$ be an arbitrary size-$\kk$ committee. Define
\[
\W_1\coloneqq \W\cap \aaa_{\mathcal{S}_1},\qquad
\W_2\coloneqq \W\cap \aaa_{\mathcal{S}_2},\qquad
\W_3\coloneqq \W\setminus (\aaa_{\mathcal{S}_1}\cup \aaa_{\mathcal{S}_2}).
\]
Since voters in $\vvv^1$ approve only alternatives in $\aaa_{\mathcal{S}_1}$, we have
$\scoreTU{\CC}(\vvv^1,\W)=\scoreTU{\CC}(\vvv^1,\W_1)$, and analogously
$\scoreTU{\CC}(\vvv^2,\W)=\scoreTU{\CC}(\vvv^2,\W_2)$ and
$\scoreTU{\CC}(\vvv\setminus(\vvv^1\cup\vvv^2),\W)=\scoreTU{\CC}(\vvv\setminus(\vvv^1\cup\vvv^2),\W_3)$.

We claim that if $|\W_1|\le \nn_1$, then $\scoreTU{\CC}(\vvv^1,\W_1)\le |\vvv^1|-L$.
Indeed, the sets corresponding to~$\W_1$ cover at most $3|\W_1|\le 3\nn_1$ elements of $X_1$.
Since $|X_1| = 3\nn_1$ but $I_1$ has no exact cover, some element~$x_i\in X_1$ is not covered by these sets.
Then none of the $L$ voters $\votone_{i,1},\ldots,\votone_{i,L}$ is covered by $\W_1$, and hence at least $L$ voters in $\vvv^1$ are uncovered.
Therefore $\scoreTU{\CC}(\vvv^1,\W_1)\le |\vvv^1|-L$. The same argument gives:
if $|\W_2|\le \nn_1$, then $\scoreTU{\CC}(\vvv^2,\W_2)\le |\vvv^2|-L$.

Next, note that the dummy voter~$\dummyvoter$ is covered by $\W_3$ if and only if $\dummyalt\in \W_3$, but no other voter approves~$\dummyalt$.
Moreover, the $\nn_2$ variable voter groups are disjoint in the sense that for each $i\in[\nn_2]$ the voters
$\{\varvot_{i,1},\ldots,\varvot_{i,3L}\}$ approve only $\{\varalt_i,\negvaralt_{i}\}$, and different variables have disjoint approved sets.
Thus, in order to cover all variable voters in $\vvv^X$, $\W_3$ must contain at least one literal-alternative for each variable, i.e., at least $\nn_2$ alternatives from $\aaa_X$.
In particular, if $|\W_3|\le \nn_2$, then $\W_3$ cannot contain $\dummyalt$ and one literal for each variable at the same time, so at least one voter in $\vvv\setminus(\vvv^1\cup\vvv^2)$ is uncovered. Hence
\[
\text{if } |\W_3|\le \nn_2, \text{ then }
\scoreTU{\CC}\big(\vvv\setminus(\vvv^1\cup\vvv^2),\W_3\big)\le \big|\vvv\setminus(\vvv^1\cup\vvv^2)\big|-1.
\]
Now suppose, towards a contradiction, that $\scoreTU{\CC}(\vvv,\W)\ge |\vvv|-1$.
Then the deficit of uncovered voters over all of $\vvv$ is at most $1$.
In particular, we cannot have $|\W_1|\le \nn_1$ (otherwise $\vvv^1$ alone loses at least $L\ge 1$), so $|\W_1|\ge \nn_1+1$.
Similarly, $|\W_2|\ge \nn_1+1$.
Finally, we must have $|\W_3|\ge \nn_2$ (otherwise $\vvv^X$ loses at least $3L\ge 1$), and in fact $|\W_3|\ge \nn_2$ is necessary just to cover all variable voters.
Therefore,
\[
|\W|
=|\W_1|+|\W_2|+|\W_3|
\ge (\nn_1+1)+(\nn_1+1)+\nn_2
=2\nn_1+\nn_2+2
>\kk,
\]
a contradiction. Hence no size-$\kk$ committee can cover at least~$|\vvv|-1$ voters, so $\cval{\CC}(\vvv)\le |\vvv|-2<|\vvv|-1=\sum_{v\in\vvv}\alloc(v)$ and $\alloc$ is not feasible for~$\vvv$.

\smallskip
\myparait{Case 2: $I_2$ has a satisfying truth assignment and $I_1$ has an exact cover.}
We show $\score{\CC}(\vvv)=|\vvv|=\cval{\CC}(\vvv)>\sum_{v\in\vvv}\alloc(v)=|\vvv|-1$, hence $\alloc$ is infeasible for~$\vvv$ as well. 

Let $\mathcal{S}'$ be an exact cover of $I_1$. Define
\[
\W_1\coloneqq \{\setaltone_{j}\mid S_j\in\mathcal{S}'\},
\qquad
\W_2\coloneqq \{\setalttwo_{j}\mid S_j\in\mathcal{S}'\}.
\]
Then $|\W_1|=|\W_2|=\nn_1$ and $\W_1$ (resp.\ $\W_2$) covers all voters in $\vvv^1$ (resp.\ $\vvv^2$).
Let $\satassign$ be a satisfying assignment of $\varphi$, and define
\[
\W_3\coloneqq \{\varalt_{i}\mid \satassign(x_i) = \true\}\ \cup\ \{\negvaralt_i\mid \satassign(x_i) = \false\}.
\]
Let $\W\coloneqq\W_1\cup\W_2\cup\W_3\cup\{\dummyalt\}$, then $|\W|=\nn_1+\nn_1+\nn_2+1=\kk$ and $W$ covers all voters in $\vvv$, hence $\cval{\CC}(\vvv)=|\vvv|>\sum_{v\in\vvv}\alloc(v)$.
This finishes the contra-positive and thus the ``if'' part.

For the ``\textbf{only if}'' part, assume that $I$ is a yes-instance of \textsc{RX3C-UNSAT}, i.e., $I_1$ has an exact cover and $\varphi$ is unsatisfiable.
We first show feasibility of $\alloc$, i.e., $\sum_{v\in \vvv}\alloc(v)=\cval{\CC}(\vvv)$.

Let $\mathcal{S}'$ be an exact cover of $I_1$ and define $\W_1,\W_2$ as above (so $|\W_1|=|\W_2|=\nn_1$ and they cover $\vvv^1,\vvv^2$).
Let
\[
\W_3\coloneqq \{\varalt_1,\ldots,\varalt_{\nn_2}\}\cup\{\clausealt\}.
\]
Then $|\W_3|=\nn_2+1$, and $\W_3$ covers all variable voters (each group~$\{\varvot_{i,1},\ldots,\varvot_{i,3L}\}$ is covered by $\varalt_i$) and all clause voters (every clause voter approves $\clausealt$).
Set
\[
\W^\star\coloneqq \W_1\cup \W_2\cup \W_3.
\]
Then $|\W^\star|=\nn_1+\nn_1+(\nn_2+1)=\kk$ and $\W^\star$ covers all voters except $\dummyvoter$.
Hence $\cval{\CC}(\vvv)\ge |\vvv|-1=\sum_{v\in\vvv}\alloc(v)$.

For the matching upper bound, let $\W$ be any size-$\kk$ committee, and decompose it into $\W_1,\W_2,\W_3$ as in Case~1.
If $|\W_1|<\nn_1$ or $|\W_2|<\nn_1$, then $\vvv^1$ or $\vvv^2$ loses at least $L\ge 1$ voters, hence $\scoreTU{\CC}(\vvv,\W)\le |\vvv|-1$.
So suppose $|\W_1|\ge \nn_1$ and $|\W_2|\ge \nn_1$.
Then $|\W_3|\le \kk-2\nn_1=\nn_2+1$.
If $\dummyalt\in \W_3$, then $\W_3$ contains at most $\nn_2$ other alternatives from $\aaa_X\cup\{\clausealt\}$.
To cover all $\nn_2$ variable voter groups one needs at least $\nn_2$ literal-alternatives, leaving no room for $\clausealt$.
In that situation, the literals in $\W_3$ select at most one of $\{\varalt_i,\negvaralt_{i}\}$ for each $i$ (otherwise fewer variables are covered), and thus correspond to a truth assignment of $X_2$; since $\varphi$ is unsatisfiable, at least one clause voter is not covered by these literals. Hence $\W$ fails to cover at least one voter among $\vvv^\varphi$.
If instead $\dummyalt\notin \W_3$, then $\dummyvoter$ is uncovered.
In both cases, $\W$ leaves at least one voter uncovered and thus $\scoreTU{\CC}(\vvv,\W)\le |\vvv|-1$.
Together with the lower bound from $\W^\star$, we obtain $\cval{\CC}(\vvv)=|\vvv|-1=\sum_{v\in\vvv}\alloc(v)$, so $\alloc$ is feasible.

It remains to show that $\alloc$ is not blocked.

\begin{claim}\label{clm:CCTU_normalform}
  Assume $I$ is a yes-instance. If there exists a blocking voter coalition~$\vvv'\subseteq\vvv$, then there exists a blocking voter coalition~$\hat{\vvv}$ with a witness $\hat{\W}$ such that
\begin{compactenum}[(i)]
\item\label{itm:nf-dummy} $\dummyvoter\in \hat{\vvv}$,
\item\label{itm:nf-covered} $\scoreTU{\CC}(\hat{\vvv},\hat{\W})=|\hat{\vvv}|$ (i.e., every voter in $\hat{\vvv}$ is covered by $\hat{\W}$),
\item\label{itm:nf-incl} $\vvv^X\cup \vvv^1 \cup \vvv^2\subseteq \hat{\vvv}$.
\end{compactenum}
\end{claim}

\begin{claimproof}{clm:CCTU_normalform}
  Let $\vvv'$ be blocking with witness $\W'$.

  Statement~\eqref{itm:nf-dummy}: %
  If $\dummyvoter\notin\vvv'$, then $\sum_{v\in\vvv'}\alloc(v)=|\vvv'|$, while always
  $\scoreTU{\CC}(\vvv',\W')\le |\vvv'|$, contradicting that $\vvv'$ blocks $\alloc$.
  Hence $\dummyvoter\in\vvv'$.
  
  Statement~\eqref{itm:nf-covered}: By the first statement, since $\dummyvoter\in\vvv'$, we have $\sum_{v\in\vvv'}\alloc(v)=|\vvv'|-1$. Therefore,
  \[
    |\vvv'|-1\ =\ \sum_{v\in\vvv'}\alloc(v)\ <\ \scoreTU{\CC}(\vvv',\W')\ \le\ |\vvv'|.
  \]
  As the \CC-score is integral, it follows that $\scoreTU{\CC}(\vvv',\W')=|\vvv'|$, i.e., $\W'$ covers every voter in $\vvv'$.

  Statement~\eqref{itm:nf-incl}: We now construct $(\hat{\vvv},\hat{\W})$ from $(\vvv',\W')$ by successive operations that preserve blocking and maintain full coverage.

  For the admissibility bound we use the floor variant:
a witness $\W'$ for $\vvv'$ must satisfy
\[
|\W'|\ \le\ \left\lfloor \frac{\kk|\vvv'|}{|\vvv|}\right\rfloor.
\]
We will repeatedly use that $|\vvv|=(\mm+2)\kk-1$ and $3L=\mm+2$, hence
\begin{align}
\frac{|\vvv|}{\kk} = (\mm+2)-\frac{1}{\kk} < 3L
\label{eq:floor-quota} %
\end{align}

\smallskip
\myparait{Step 1: Close $\vvv'$ under copies (without changing the witness).}
We iteratively add missing copies of already present voters:
\begin{compactitem}[--]
\item If for some $i\in[\nn_2]$ and some $\ell_1,\ell_2\in[3L]$ we have $\varvot_{i,\ell_1}\in\vvv'$ but $\varvot_{i,\ell_2}\notin\vvv'$, then add $\varvot_{i,\ell_2}$ to $\vvv'$.
\item If for some $j\in[3\nn_1]$ and some $\ell_1,\ell_2\in[L]$ we have $\votone_{j,\ell_1}\in\vvv'$ but $\votone_{j,\ell_2}\notin\vvv'$, then add $\votone_{j,\ell_2}$ to $\vvv'$.
\item If for some $j\in[3\nn_1]$ and some $\ell_1,\ell_2\in[L]$ we have $\vottwo_{j,\ell_1}\in\vvv'$ but $\vottwo_{j,\ell_2}\notin\vvv'$, then add $\vottwo_{j,\ell_2}$ to $\vvv'$.
\end{compactitem}
Each such addition preserves blocking with the same witness $\W'$:
Since each voter in~$\vvv'$ is covered by~$\W'$ (Statement~\eqref{itm:nf-covered})
and the added voter has exactly the same approval set as a voter already in~$\vvv'$,
he is covered by $\W'$, and since $\alloc$ assigns $1$ to all non-dummy voters, both $\scoreTU{\CC}(\cdot,W')$ and $\sum \alloc(\cdot)$ increase by exactly $1$.
The proportional bound $|\W'|\le \share$ remains valid because the right-hand side increases when we add voters.
After finitely many steps, we obtain a blocking coalition (still denoted $\vvv'$) such that for every $i\in[\nn_2]$ either the entire voter group 
$\{\varvot_{i,1},\ldots,\varvot_{i,3L}\}$ is contained in $\vvv'$ or it is disjoint from $\vvv'$, and analogously for each element voter group in $\vvv^1$ and $\vvv^2$.

\myparait{Step 2: Add~$\vvv^X$.}
If $\vvv^X\subseteq\vvv'$, we are done. Otherwise let $U\subseteq[\nn_2]$ be the set of indices $i$ whose entire block is missing,
so $|\vvv^X\setminus\vvv'|=|U|\cdot 3L$. Define
\[
\hat{\vvv}\ \coloneqq\ \vvv'\cup \bigcup_{i\in U}\{\varvot_{i,1},\ldots,\varvot_{i,3L}\}.
\]
Let $\hat{\W}$ be obtained from $\W'$ by, for each $i\in U$ with $\W'\cap\{\varalt_i,\negvaralt_{i}\}=\emptyset$,
adding one of $\varalt_i,\negvaralt_{i}$.
Then every newly added voter is covered by construction, and every voter in $\vvv'$ remains covered, so by the second statement, 
\[
\scoreTU{\CC}(\hat{\vvv},\hat{\W})=|\hat{\vvv}|.
\]
Moreover, $|\hat{\W}|\le |\W'|+|U|$.

It remains to check admissibility for $(\hat{\vvv},\hat{\W})$.
Since $|\hat{\vvv}|=|\vvv'|+|U|\cdot 3L$, we have
\[
\frac{\kk|\hat{\vvv}|}{|\vvv|}
=\frac{\kk|\vvv'|}{|\vvv|} + |U|\cdot \frac{\kk\cdot 3L}{|\vvv|}.
\]
By \eqref{eq:floor-quota}, $\frac{\kk\cdot 3L}{|\vvv|}>1$, hence
\[
\left\lfloor \frac{\kk|\hat{\vvv}|}{|\vvv|}\right\rfloor
\ \ge\
\left\lfloor \frac{\kk|\vvv'|}{|\vvv|}\right\rfloor + |U|.
\]
Using $|\W'|\le \share$ and $|\hat{\W}|\le |\W'|+|U|$, we conclude
$|\hat{\W}|\le \lfloor \kk|\hat{\vvv}|/|\vvv|\rfloor$.
Thus $\hat{\vvv}=\vvv'\cup\vvv^X$ is blocking.

\myparait{Step 3: Add $\vvv^1$ and $\vvv^2$ via an exact-cover exchange.}
We show the step for $\vvv^1$; the argument for $\vvv^2$ is identical.
By Step~1, we may assume that for each $j\in[3\nn_1]$ either the whole voter group~$\{\votone_{j,1},\ldots,\votone_{j,L}\}$ is contained in $\vvv'$ or disjoint from $\vvv'$.

Let $\W'$ witness that $\vvv'$ is blocking. Set $g_1$ such that $|\vvv'\cap \vvv^1|=g_1\cdot L$ (number of present element-blocks),
and let $g_2\coloneqq |\W'\cap \aaa_{\mathcal{S}_1}|$.
By Statement~\eqref{itm:nf-covered}, every voter in $\vvv'\cap\vvv^1$ is covered by $\W'$, and only alternatives in $\aaa_{\mathcal{S}_1}$ can cover them.
Each alternative in $\aaa_{\mathcal{S}_1}$ corresponds to a 3-set and hence can cover voters from at most three element-voter groups.
Therefore $g_1\le 3g_2$.

Let $\W_x\subseteq \aaa_{\mathcal{S}_1}$ be the $\nn_1$ set1-alternatives corresponding to an exact cover of $I_1$.
Define
\[
\hat{\vvv}\ \coloneqq\ \vvv'\cup \vvv^1,
\qquad
\hat{\W}\ \coloneqq\ (\W'\setminus \aaa_{\mathcal{S}_1})\cup \W_x.
\]
Then $\hat{\W}$ covers all voters in $\vvv^1$ (exact cover) and does not affect coverage of voters outside $\vvv^1$,
since no voter outside $\vvv^1$ approves alternatives in $\aaa_{\mathcal{S}_1}$.
Hence $\scoreTU{\CC}(\hat{\vvv},\hat{\W})=|\hat{\vvv}|$.

For admissibility, note that $|\hat{\W}|=|\W'|-g_2+\nn_1$ and $|\hat{\vvv}|=|\vvv'|+(3\nn_1-g_1)L$.
Using $g_1\le 3g_2$ we obtain
\[
|\hat{\vvv}|
=|\vvv'|+(3\nn_1-g_1)L
\ \ge\
|\vvv'|+3(\nn_1-g_2)L.
\]
Now apply \eqref{eq:floor-quota} with $|U|=\nn_1-g_2$:
adding $3L$ voters increases the floor budget by at least $1$, hence adding $3(\nn_1-g_2)L$ voters increases it by at least $\nn_1-g_2$.
Formally,
\[
\left\lfloor \frac{\kk|\hat{\vvv}|}{|\vvv|}\right\rfloor
\ \ge\
\left\lfloor \frac{\kk|\vvv'|}{|\vvv|}\right\rfloor + (\nn_1-g_2).
\]
Using $|\W'|\le \share$ and $|\hat{\W}|=|\W'|+(\nn_1-g_2)$, we obtain
$|\hat{\W}|\le \lfloor \kk|\hat{\vvv}|/|\vvv|\rfloor$.
Thus $\vvv'\cup \vvv^1$ is blocking.

Analogously, we can show that $\vvv'\cup \vvv^2$ is blocking with an admissible committee.

\end{claimproof}
    
Now suppose, towards a contradiction, that a blocking coalition exists.
By \cref{clm:CCTU_normalform} there exists a blocking coalition $\hat{\vvv}$ containing $\dummyvoter$ and all voters in $\vvv^X\cup\vvv^1\cup\vvv^2$, with witness $\hat{\W}$ covering all voters in $\hat{\vvv}$.
Then $\hat{\W}$ must contain at least
$\nn_1$ alternatives from $\aaa_{\mathcal{S}_1}$ (to cover $\vvv^1$),
$\nn_1$ from $\aaa_{\mathcal{S}_2}$ (to cover $\vvv^2$),
$\nn_2$ from $\aaa_X$ (to cover $\vvv^X$),
and $\dummyalt$ (to cover $\dummyvoter$).
Thus $|\hat{\W}|\ge 2\nn_1+\nn_2+1=\kk$.
But also $|\hat{\W}|\le \kk\cdot |\hat{\vvv}|/|\vvv|\le \kk$, hence $|\hat{\W}|=\kk$, which implies $|\hat{\vvv}|=|\vvv|$ and thus $\hat{\vvv}=\vvv$.
Therefore $\vvv$ would be blocking, contradicting $\cval{\CC}(\vvv)=\sum_{v\in\vvv}\alloc(v)$ shown above. Hence no blocking coalition exists and $\alloc\in\TUcore{\CC}$. %
\end{proof}

Finally, the \NPhh{ness} for \PAV\ comes from the fact that computing a winning committee under \PAV\ is \NPhh~\cite{Aziz15}. 

\begin{restatable}[]{proposition}{propTUCVPAV}\label{prop:TU-CV-PAV}
  \MWGTUCverif{\PAV} is \NPhh.
\end{restatable}
\begin{proof}
    We reduce from \textsc{RX3C}.
    Let $I=(X,\mathcal{S})$ be an \textsc{RX3C} instance with $X=\{x_1,\ldots, x_{3\nn}\}$; the definition of \textsc{RX3C} can be found in the proof of \cref{thm:TU-CE}.
    We create an instance~$I'=(\ppp,\kk,\alloc)$ of \MWGTUCverif{\PAV} as follows.
    
    \mypara{Profile construction.} Let $\kk\coloneq\nn$.
    \begin{compactitem}[--]
      \item For each set $S_j\in\mathcal{S}$, create an alternative $\alt_j$.
      \item For each element $x_i\in X$, create a voter $v_i$. Their approvals are $\app(v_i)\coloneqq\{\alt_j\mid x_i\in S_j\}$.
	\end{compactitem}
    Thus the total number of voters is $|\vvv|=3\nn$. 
  
    \mypara{The utility function~$\alloc$.}
    $\alloc(v_i)=1$ for all $v_i\in\vvv$.
    
    \mypara{Correctness.}
    We show that $I$ is a yes-instance of \textsc{RX3C} if and only if $\alloc$ is in the \TUcore{\PAV}.

    For the ``if'' direction, suppose $\alloc\in \TUcore{\PAV}$.
    By definition, there must exist a winning committee that realizes~$\alloc$.
    Let $\W$ be such a committee, i.e., $|\W|=\kk = \nn$ and $\scoreTU{\PAV}(\vvv,\W)=\sum_{i=1}^{3\nn}\alloc(v_i)=3\nn=|\vvv|$.
    As each alternative~$\alt_j\in\aaa$ is approved by exactly three voters, it holds that each voter must approve exactly one alternative in $\W$; otherwise the score would be less than $|\vvv|$ due to the definition of the \PAV-score.
    Then, it must also hold that the sets corresponding to alternatives in $\W$ form an exact cover:
    Each element is covered by them exactly once. This concludes the ``if'' direction. 

    For the ``only if'' direction, suppose that $I$ is a yes-instance and let $\mathcal{S}'$ be an exact cover for~$I$.
    We show that $\alloc$ is realized by a size-$\kk$ committee and no coalition (including the grand coalition) should be able to block~$\alloc$.  

    For the realizability, define $\W\coloneqq\{\alt_j\mid S_j\in\mathcal{S}'\}$.
    Then, $|\W|=\nn$.
    Moreover, $\scoreTU{\PAV}(\vvv,\W)=|\vvv|$ since $\mathcal{S}'$ covers each element exactly once.

    It remains to show that no coalition is blocking. 
    Suppose for the sake of contradiction, that $\vvv'\subseteq \vvv$ is a coalition that is \TUblocking{\PAV}~$\alloc$ and $\W'$ is a witness.
    By the admissibility of~$\W'$, we infer
    \begin{align}
      |\W'| \le \share =\lfloor \frac{|\vvv'|}{3} \rfloor \le \frac{|\vvv'|}{3}.\label{eq:PV-TU-CV-admissible}
    \end{align}
    By the score, we infer $\scoreTU{\PAV}(\vvv',\W') > \sum_{v_i\in \vvv'}\alloc(v_i) = |\vvv'|$.
    Since each alternative is approved by exactly three voters, it follows that
    $3|\W'| \ge \scoreTU{\PAV}(\vvv',\W') > |\vvv'|$, a contradiction to~\eqref{eq:PV-TU-CV-admissible}.
    Therefore $\alloc\in \TUcore{\PAV}$. %
    This concludes the proof.
\end{proof}

\subsection{Complexity for the \NTU\ model}\label{sub:complexity-NTU}

We now consider the \NTU\ model, where coalitions cannot redistribute utility.
Here, blocking requires a \emph{component-wise} strict improvement for every coalition member, which makes core membership verification inherently more combinatorial: a witness must specify an admissible committee whose induced per-voter scores dominate the proposed utilities on the entire coalition.

By \cref{thm:CCNTUCE}, we immediately obtain the following polynomial result.
\begin{theorem}\label{thm:NTU-CE}
  \MWGNTUCexist{\CC} is in \PP.
\end{theorem}

For the core membership problem, we show a somewhat interesting contrast to the \TU\ model. It becomes \DPcc\ for \AV\ and hence \SAV\ and \PAV\ (see \cref{obs:AVPAVequiv}), but is \NPcc\ for \CC.
We prove the \DPhh{ness} by a reduction from the following \DPhh\ problem, denoted \XTCDP.\footnote{%
  \XTCDP\ is \DPhh: Since the underlying \textsc{RX3C} (short for \textsc{Restricted Exact Cover by 3-Sets}) decision problem (asking whether $I_1$ has an exact cover) is \NP-complete, there exists a polynomial-time
  reduction~$r$ from \textsc{SAT} to \textsc{RX3C}.
  Hence $(\varphi_1,\varphi_2)\mapsto (r(\varphi_1),r(\varphi_2))$
  yields a polynomial-time reduction from the \DPcc\ problem \textsc{SAT-UNSAT}\cite{Papadimitriou88facets} to \XTCDP.}

\decprob{\XTCDP}
{Two \textsc{RX3C} instances $I_1=(X_1,\mathcal{S}_1)$ and $I_2=(X_2,\mathcal{S}_2)$ such that for a common $\nn\in\mathds{N}$ and each $z\in[2]$:
  \begin{inparaenum}
    \item $|X_z|=3\nn$,
    \item $\mathcal{S}_z=\{S^z_1,\ldots,S^z_{3\nn}\}$ with $S^z_j\subseteq X_z$ and $|S^z_j|=3$ for all $j\in[3\nn]$, and
    \item each element of $X_z$ is contained in exactly three sets of $\mathcal{S}_z$.
  \end{inparaenum}}
{Is it true that $I_1$ has an \emph{exact cover} (i.e., a subcollection $\mathcal{S}'_1\subseteq \mathcal{S}_1$ of cardinality~$\nn$ such that each $x\in X_1$
belongs to exactly one set of $\mathcal{S}'_1$) while $I_2$ does not?}

Our reduction extends a hardness reduction by Brill et al.~\cite{BrillGPSW24},
which shows that verifying whether a committee is in the core is \coNPhh.
\begin{restatable}{theorem}{thmNTUCVAVPAV}\label{thm:NTU-CV-AV-PAV}
  \MWGNTUCverif{\AV} and hence \MWGNTUCverif{\SAV} and \MWGNTUCverif{\PAV} are \DPcc.
 \end{restatable}

 \begin{proof}
   \DP-membership follows from \cref{thm:complexity-upperbounds}\eqref{upper:NTU-verif}.
By \cref{obs:AVPAVequiv}, it suffices to show \DP-hardness for \MWGNTUCverif{\AV}.
As mentioned, we reduce from \XTCDP.

Let $I=(I_1,I_2)$ be an instance of \XTCDP, where for~$z\in [2]$,  $I_z=(X_z,\mathcal{S}_z)$, 
$|X_z|=3\nn$, $\mathcal{S}_z=\{S^z_1,\ldots,S^z_{3\nn}\}$, each $|S^z_j|=3$, and each element of $X_z$ occurs in exactly three sets.
We construct an instance $I'=(\ppp,\kk,\alloc)$ of \MWGNTUCverif{\AV}.

\mypara{Profile construction.} Set $\kk\coloneqq 2\nn$. 
\begin{compactenum}[(i)]
  \item For each $S^1_j\in\mathcal{S}_1$ create a \myemph{set1-alternative}~$\setaltone_j$ and let
$\aaa_{\mathcal{S}_1}\coloneqq\{\setaltone_1,\ldots,\setaltone_{3\nn}\}$.
For each $S^2_j\in\mathcal{S}_2$ create a \myemph{set2-alternative}~$\setalttwo_j$ and let
$\aaa_{\mathcal{S}_2}\coloneqq\{\setalttwo_1,\ldots,\setalttwo_{3\nn}\}$.
Create private alternatives $\aaa_P\coloneqq\{p_1,\ldots,p_{\nn-1}\}$ and one dummy alternative $\alt_d$.
Thus $\aaa=\aaa_{\mathcal{S}_1}\cup\aaa_{\mathcal{S}_2}\cup\aaa_P\cup\{\alt_d\}$.

\item For each $x^1_i\in X_1$ create an \myemph{element1-voter}~$\votone_i$ with approvals
$\app(\votone_i)\coloneqq\{\setaltone_j\mid x^1_i\in S^1_j\}$, and let $\vvv_{X_1}\coloneqq\{\votone_i\mid x^1_i\in X_1\}$.
For each $x^2_i\in X_2$ create an \myemph{element2-voter}~$\vottwo_i$ with approvals
$\app(\vottwo_i)\coloneqq\{\setalttwo_j\mid x^2_i\in S^2_j\}$, and let $\vvv_{X_2}\coloneqq\{\vottwo_i\mid x^2_i\in X_2\}$.
Add two dummy voters: $d_1$ with $\app(d_1)=\aaa_{\mathcal{S}_1}$ and $d_2$ with $\app(d_2)=\aaa_{\mathcal{S}_2}\cup\aaa_P$.
Hence $|\vvv|=|\vvv_{X_1}|+|\vvv_{X_2}|+2=6\nn+2$.
\end{compactenum}

\mypara{The utility function $\alloc$.}
Set $\alloc(\votone_i)=1$ for all $\votone_i\in\vvv_{X_1}$, $\alloc(\vottwo_i)=0$ for all $\vottwo_i\in\vvv_{X_2}$,
and $\alloc(d_1)=\nn$, $\alloc(d_2)=\nn-1$.

\mypara{Correctness.}
We show that $I$ is a yes-instance of \XTCDP\ if and only if $\alloc\in\NTUcore{\AV}$.

For the ``if'' part, assume $\alloc\in\NTUcore{\AV}$.
Then, $\alloc$ is feasible for~$\vvv$ (i.e., $\alloc\in\NTUutil{\AV}(\vvv)$),
so there exists a size-$\kk$ committee $\W$ with 
$\scoreNTU{\AV}(v,\W)=\alloc(v)$ for all $v\in\vvv$.
From $\alloc(d_1)=\nn$ and $\app(d_1)=\aaa_{\mathcal{S}_1}$ it follows that $|\W\cap\aaa_{\mathcal{S}_1}|=\nn$.
Moreover, for each element1-voter~$\votone_i$ we have $\alloc(\votone_i)=1$ and $\app(\votone_i)\subseteq\aaa_{\mathcal{S}_1}$, hence
$|\W\cap\app(\votone_i)|=1$. Therefore the $\nn$ chosen set1-alternatives cover every element of $X_1$ exactly once, so $I_1$ has an exact cover.

It remains to show that $I_2$ has no exact cover.
Suppose, for contradiction, that $I_2$ admits an exact cover
$\mathcal{S}'\subseteq\mathcal{S}_2$ of size $\nn$ and set $\W'\coloneqq\{\setalttwo_j\mid S^2_j\in\mathcal{S}'\}$, and $\vvv'=\vvv_{X_2}\cup \{d_2\}$.
We claim that $\vvv'$ is blocking with witness~$\W'$.
Clearly, $|\W'|=\nn$. Each element2-voter~$\vottwo_i$ satisfies $\scoreNTU{\AV}(\vottwo_i,\W')=1>\alloc(\vottwo_i)=0$,
and also $\scoreNTU{\AV}(d_2,\W')=\nn>\alloc(d_2)=\nn-1$.
Since $|\W'|=\nn=\kk\cdot (3\nn+1)/(6\nn+2)= \share$, it is admissible for~$\vvv_{X_2}\cup\{d_2\}$.
So, $\vvv'$ is \NTUblocking{\AV} against~$\alloc$,
contradicting $\alloc\in\NTUcore{\AV}$. Hence $I_2$ is a no-instance, and $I$ is a yes-instance of \XTCDP.

For the ``only if'' part,  assume $I$ is a yes-instance%
, i.e., $I_1$ has an exact cover and $I_2$ has none.
We show $\alloc\in\NTUcore{\AV}$ by establishing feasibility and ruling out blocking coalitions.

Let $\mathcal{S}'\subseteq\mathcal{S}_1$ be an exact cover of size~$\nn$ and define
$\W_S\coloneqq\{\setaltone_j\mid S^1_j\in\mathcal{S}'\}$ and $\W\coloneqq \W_S\cup\aaa_P\cup\{\alt_d\}$.
Then $|\W|=\nn+(\nn-1)+1=2\nn=\kk$.
By construction, $\scoreNTU{\AV}(\votone_i,\W)=1$, $\scoreNTU{\AV}(d_1,\W)=\nn$, $\scoreNTU{\AV}(\vottwo_i,\W)=0$,
and $\scoreNTU{\AV}(d_2,\W)=\nn-1$, hence $\scoreNTU{\AV}(v,\W)=\alloc(v)$ for all $v\in\vvv$.
Thus $\alloc\in\NTUutil{\AV}(\vvv)$.

Now we show that no coalition is blocking.
Suppose for contradiction that some coalition $\vvv'\subseteq\vvv$ is \NTUblocking{\AV} against~$\alloc$ with witness committee~$\W'$.

\begin{restatable}[]{claim}{clmAVNTUverifXone}\label{clm:AVNTUverifXone}
  If $\vvv'$ \NTUblocks{\AV} $\alloc$, then $\vvv'\setminus (\vvv_{X_1}\cup\{d_1\})$ \NTUblocks{\AV} $\alloc$ as well.
\end{restatable}
\begin{claimproof}{clm:AVNTUverifXone}
Let $\vvv'$ be a blocking coalition with witness committee $\W'$.
Write $\W_1\coloneqq \W'\cap \aaa_{\mathcal{S}_1}$ and $\hat{\W}\coloneqq \W'\setminus \aaa_{\mathcal{S}_1}$, and define
$\hat{\vvv}\coloneqq \vvv'\setminus(\vvv_{X_1}\cup\{d_1\})$.
Every voter in $\hat{\vvv}$ approves no set1-alternatives, hence $\scoreNTU{\AV}(v,\hat{\W})=\scoreNTU{\AV}(v,\W')$ for all $v\in\hat{\vvv}$,
so strict improvement over $\alloc$ is preserved for $\hat{\vvv}$.

It remains to show that $\hat{\vvv}\neq\emptyset$ and that $\hat{\W}$ is admissible for $\hat{\vvv}$.
We distinguish two cases.

\mypara{Case A: $d_1\in\vvv'$.}
Blocking requires $\scoreNTU{\AV}(d_1,\W') > \nn = \alloc(d_1)$, hence
$|\W_1|=|\W'\cap\aaa_{\mathcal{S}_1}|\ge \nn+1$.
Since $\W'$ is admissible for $\vvv'$, we have $|\W'|\le \share$, implying
$|\vvv'|/|\W'|\ge |\vvv|/\kk=(3\nn+1)/\nn$.
Together with $|\W'|\ge |\W_1|\ge \nn+1$, this yields $|\vvv'|>(3\nn+1)$, and thus
$\hat{\vvv}=\vvv'\setminus(\vvv_{X_1}\cup\{d_1\})\neq\emptyset$.
To see that $\hat{\W}$ is admissible for $\hat{\vvv}$, write $q\coloneqq \kk/|\vvv|=\nn/(3\nn+1)$ and note that
$|\hat{\W}|=|\W'|-|\W_1|$ and $|\hat{\vvv}|\ge |\vvv'|- (3\nn+1)$.
Using $|\W'|\le q|\vvv'|$ (since $\W'$ is admissible for~$\vvv'$) and $|\W_1|\ge \nn+1\ge q(3\nn+1)+1$, we obtain
$|\hat{\W}|\le q|\vvv'|- q(3\nn+1) -  1 \le q|\hat{\vvv}| - 1 \le \lfloor \frac{|\hat{\vvv}| \kk}{|\vvv|} \rfloor$,
so $\hat{\W}$ is admissible for~$\hat{\vvv}$.

\mypara{Case B: $d_1\notin\vvv'$.}
Let $T\coloneqq \vvv'\cap \vvv_{X_1}$.
Assume $T\neq\emptyset$ as otherwise $\hat{\vvv}=\vvv'$ and $\hat{\W}=\W'$, so the claim is immediate.

Since $\alloc(\votone_i)=1$ for all $\votone_i\in\vvv_{X_1}$ and $\vvv'$ \NTUblocks{\AV} $\alloc$, we have
$\scoreNTU{\AV}(\votone_i,\W')>\alloc(\votone_i)=1$ for every $\votone_i\in T$, and hence
$\scoreNTU{\AV}(\votone_i,\W')\ge 2$.
As $\votone_i$ approves only set1-alternatives, this implies $|\W_1\cap \app(\votone_i)|\ge 2$ for all $\votone_i\in T$.
Double counting approval incidences between $T$ and $\W_1$ yields
$2|T|\le \sum_{\votone_i\in T}|\W_1\cap \app(\votone_i)|\le 3|\W_1|$,
because each set1-alternative is approved by exactly three voters in $\vvv_{X_1}$.
Thus,
\begin{align}
  |\W_1|\ge \frac{2}{3}|T|.\label{eq:W1-T}
\end{align}
Similarly to Case A, we show that $\hat{\vvv}$ is non-empty and $\hat{\W}$ is admissible for~$\hat{\vvv}$.
Suppose towards a contradiction that $\hat{\vvv}=\emptyset$. Then $\vvv'=T\subseteq \vvv_{X_1}$. 
Hence,
\begin{align}
  |\W'|\ge |\W_1| \stackrel{\eqref{eq:W1-T}}{\ge}  \frac{2}{3}|T| = \frac{2}{3}|\vvv'|.\label{eq:W'-V'}
\end{align}
On the other hand, admissibility of $\W'$ for $\vvv'$ gives
$|\W'|\le \share = \lfloor \nn/(3\nn+1)\,|\vvv'| \rfloor < \tfrac{1}{3}|\vvv'|$,
a contradiction to \eqref{eq:W'-V'}. Therefore $\hat{\vvv}\neq\emptyset$.

Now, we consider admissibility of~$\hat{\W}$. 
Let $q\coloneqq \kk/|\vvv|=\nn/(3\nn+1)$.
From admissibility of $\W'$ for~$\vvv'$ we have $|\W'|\le q|\vvv'|$.
Moreover, $|\hat{\W}|=|\W'|-|\W_1|$ and $|\hat{\vvv}|=|\vvv'|-|T|$.
Using \eqref{eq:W1-T}, $q<\tfrac{1}{3}$, and $|T|\ge 1$, we obtain $|\W_1| > q|T|$ and thus
$|\hat{\W}| < q|\vvv'|-q|T|=q|\hat{\vvv}| = \frac{|\hat{\vvv}| \kk}{|\vvv|}$.
Since $|\hat{\W}|$ is integral, this implies that $|\hat{\W}| \le  \lfloor \frac{|\hat{\vvv}| \kk}{|\vvv|} \rfloor$, so $\hat{\W}$ is admissible for $\hat{\vvv}$.

In all cases, $\hat{\vvv}$ is nonempty and \NTUblocks{\AV} $\alloc$ witnessed by $\hat{\W}$, and
$\hat{\vvv}\cap(\vvv_{X_1}\cup\{d_1\})=\emptyset$ by construction.%
\end{claimproof}

\begin{restatable}[]{claim}{clmAVNTUverifdummytwo}\label{clm:AVNTUverifdummytwo}
If $\vvv'$ \NTUblocks{\AV} $\alloc$ and $\vvv'\cap(\vvv_{X_1}\cup\{d_1\})=\emptyset$, then $d_2\in\vvv'$.
\end{restatable}
\begin{claimproof}{clm:AVNTUverifdummytwo}
Assume $\vvv'$ \NTUblocks{\AV} $\alloc$ and $\vvv'\cap(\vvv_{X_1}\cup\{d_1\})=\emptyset$.
Suppose for contradiction that $d_2\notin\vvv'$. Then $\vvv'\subseteq\vvv_{X_2}$ and, since $\alloc(\vottwo_i)=0$,
blocking implies $\scoreNTU{\AV}(\vottwo_i,\W')\ge 1$ for all $\vottwo_i\in\vvv'$.
Each set2-alternative is approved by exactly three element2-voters, so $|\W'|\ge |\vvv'|/3$.
But admissibility implies $|\W'|\le \share=\lfloor\frac{\nn}{3\nn+1}|\vvv'|\rfloor<\frac{|\vvv'|}{3}$,
a contradiction. Hence $d_2\in\vvv'$.
\end{claimproof}

By \cref{clm:AVNTUverifXone} we may assume $\vvv'\cap(\vvv_{X_1}\cup\{d_1\})=\emptyset$, hence
$\vvv'\subseteq \vvv_{X_2}\cup\{d_2\}$; note that this implies that $|\vvv'|\le 3\nn+1$.
By \cref{clm:AVNTUverifdummytwo}, we have $d_2\in\vvv'$.

Since $\alloc(d_2)=\nn-1$, blocking implies $\scoreNTU{\AV}(d_2,\W')\ge \nn$.
As $d_2$ approves exactly $\aaa_{\mathcal{S}_2}\cup\aaa_P$, this forces $|\W'|\ge \nn$.
Moreover, admissibility gives $|\W'|\le \share\le \kk\cdot|\vvv'|/|\vvv|\le \kk\cdot(3\nn+1)/(6\nn+2)=\nn$,
so $|\W'|=\nn$ and thus $|\vvv'|=3\nn+1$, i.e., $\vvv'=\vvv_{X_2}\cup\{d_2\}$.

Now, for each $\vottwo_i\in\vvv_{X_2}$ blocking requires $\scoreNTU{\AV}(\vottwo_i,\W')\ge 1$.
Each alternative in $\aaa_{\mathcal{S}_2}$ corresponds to a 3-set, hence is approved by exactly three element2-voters.
Since $|\W'| =\nn$, $\W'$ can ``cover'' at most $3\nn$ element2-voters, and to improve all $3\nn$ element2-voters
it must cover all of them. This implies that the $\nn$ selected 3-sets cover $X_2$ and, since their total size is $3\nn$,
they are pairwise disjoint; hence they form an exact cover for $I_2$, contradicting that $I_2$ is a no-instance. %
Thus no blocking coalition exists and $\alloc\in\NTUcore{\AV}$.

This establishes \DP-hardness for \MWGNTUCverif{\AV}. The statement for \SAV\ and \PAV\ follows from \cref{obs:AVPAVequiv}.
\end{proof}

Finally, we show an arguably interesting phenomenon: verifying \CC-core membership is \NPcc.
The \NPhh{ness} follows from the fact that computing a \CC-winning committee is known to be \NPhh. 
Membership in \NP\ follows because a committee~$\W$ witnessing feasibility of $\alloc$ serves as a polynomial-size certificate; given $\W$, core membership can be checked in polynomial time by verifying JR; see \cref{prop:CCJR}.
\begin{restatable}[]{theorem}{thmNTUCVCC}\label{thm:NTU-CV-CC}
  \MWGNTUCverif{\CC} is \NPcc.
\end{restatable}
\begin{proof}
 We show \NPhh{ness} by a reduction from the \NPcc\ problem \textsc{Set Cover}.

 \decprob{Set Cover}
 {A universe $U=[\nn]$, a family $\mathcal{S}=\{S_1,\ldots,S_{\mm}\}$ with $S_j\subseteq U$ for all $j\in[\mm]$, and an integer $h$.}
 {Is there a subfamily of size $h$ whose union is $U$?}

  Given an instance $I=([\nn],\mathcal{S},h)$ we construct an instance
  $I'=(\ppp,\kk,\alloc)$ of \MWGNTUCverif{\CC} as follows.
  Let $\vvv=\{v_1,\ldots,v_{\nn}\}$ and $\aaa=\{\alt_1,\ldots,\alt_{\mm}\}$, where $\alt_j$ corresponds to~$S_j$.
  For each voter $v_i\in\vvv$, set $\app(v_i) = \{\alt_j \mid i \in S_j\}$.
  Finally, set $\kk=h$ and define $\alloc(v_i)=1$ for all $v_i\in\vvv$.

  We claim that $I$ is a yes-instance if and only if $I'$ is a yes-instance.

  Assume that  $I$ is a yes-instance, and  let $S_{j_1},\ldots,S_{j_h}$ be a set cover.
  Let $\W=\{\alt_{j_1},\ldots,\alt_{j_h}\}$.
  Then for every $i\in[\nn]$ we have $\W\cap \app_i\neq\emptyset$, hence
  $\scoreNTU{\CC}(v_i,\W)=1=\alloc(v_i)$.
  Since no voter can obtain utility strictly larger than~$1$ under \CC, the profile $\alloc$ cannot be blocked.
  Thus $\alloc\in\NTUcore{\CC}$ and $I'$ is a yes-instance.

  Assume that $I'$ is a yes-instance, and let $\W$ be a committee with $|\W|=\kk=h$ such that
  $\scoreNTU{\CC}(v_i,\W)=\alloc(v_i)=1$ for all $v_i\in\vvv$.
  Equivalently, $\W\cap \app_i\neq\emptyset$ for all $i\in[\nn]$.
  Therefore the $h$ sets corresponding to the alternatives in $\W$ cover $U$, and $I$ is a yes-instance.

   For \NP-membership, we use a committee $\W$ as a certificate.
   Let $\W$ be a committee with $|\W|=\kk$. We can verify in polynomial time that
   $\alloc(v_i)=\scoreNTU{\CC}(v_i,\W)$ for all $v_i\in\vvv$.
   By \cref{prop:CCJR}, it then suffices to check whether $\W$ satisfies JR, which can be done as follows.
   Let $\vvv_0\coloneqq\{v_i\in\vvv \mid \W\cap \app_i=\emptyset\}$ be the set of unrepresented voters.
   Then $\W$ violates JR if and only if there exists an alternative $\alt\in\aaa$ such that
   $|\VV(\alt)\cap \vvv_0|\ \ge\ \lceil |\vvv|/\kk\rceil$; recall that $\VV(\alt)\coloneqq\{v_i\in\vvv \mid \alt\in\app_i\}$.
   This condition can be checked in polynomial time by scanning all alternatives and counting approvals inside $\vvv_0$.
   Hence \MWGNTUCverif{\CC}\ lies in \NP, completing the proof.
\end{proof}

Note that \MWGNTUCverif{\CC} is polynomial-time solvable when a witnessing committee~$\W$ is included in the input, since the task reduces to verifying $\alloc(v_i)=\scoreNTU{\CC}(v_i,\W)$ for all $v_i$ and checking JR for $\W$.
We conclude by discussing the broader implications and open directions.

\section{Conclusion}\label{sec:conclude} 
We introduced \emph{multi-winner voting games}, a cooperative-game framework for multi-winner approval voting in which coalitions are constrained by proportional entitlements and outcomes are rule-induced utility functions.
The framework supports both \NTU\ and \TU\ interpretations, separating instability that is inherent to the committee feasibility constraints from instability that disappears once coalitions can redistribute utility internally.

    Our results demonstrate that transferability fundamentally reshapes the stability landscape in predictable ways.
    The stark contrast---\AV\ and \SAV\ are well-behaved under \TU\ but open under \NTU, \CC\ is the reverse, \PAV\ yields no positive existence guarantee under \TU, and remains open under \NTU---reflects a structural principle identified in the introduction.
    Under \TU, feasibility ties to winning committees, making rule-specific properties actionable; under \NTU, this link is severed.
    The Bondareva--Shapley approach~\cite{bondareva1963,shapley1967} applies broadly, to any rule where the total committee score is the sum of individual voter scores and each alternative's contribution is independent of which others are selected.
    For any such \totsep\ rule, fractional committees cannot outperform integral ones,
    so the \TU-core is always non-empty. If proportional entitlements are transferable, such rules thus provide robust stability.
    If entitlements are rigid, rule choice matters less than utility structure: \CC's binary scores enable the JR\ characterization, while richer score ranges resist known techniques.

On the computational side, we studied \CE\ and \CV\ for both utility models and four rules.
Even when stable outcomes exist and can be constructed, \CV\ can be challenging because blocking quantifies over exponentially many coalitions and admissible committees.
Our complexity classifications delineate tractable cases (notably \MWGTUCexist{\AV}, \MWGTUCexist{\SAV}, and \MWGNTUCexist{\CC}) from settings where \CE\ and/or \CV\ become computationally hard.

Three concrete open problems warrant attention.
First, the open non-emptiness problem for the \NTU-core under \AV/\SAV/\PAV\ remains a central challenge.
For \TU-\AV\ and \TU-\SAV, efficiency constrains core elements to derive from winning committees, which anchors both the Bondareva–Shapley non-emptiness proof and the combinatorial construction in \cref{alg:TU-core-totsep}---the latter of which required a non-trivial surplus-distribution argument even within this restricted search space.
For \NTU-\AV, \NTU-\SAV, and \NTU-\PAV, core-stable outcomes need not arise from winning committees (cf.\ \cref{ex1}),
so the search space expands to all $\binom{m}{\kk}$ committees in the worst case. %
The JR equivalence that resolves \NTU-\CC\ offers no leverage here, as it relies on the binary structure of \CC-utilities.
Among potential approaches, Sperner's lemma and related combinatorial-topological tools are natural candidates,
as they operate in discrete settings without requiring convexification.
Second, the \TU\ model suggests connections to market-based interpretations of proportional entitlements (e.g., prices or transferable compensation) and raises the question of which notions of fairness or distributional constraints are compatible with core stability.
Third, it would be useful to develop approximation and relaxation frameworks that interpolate between \NTU\ and \TU\ (e.g., limited transfers, bounded side payments, or randomized outcomes) and to understand whether they recover existence or improve computational tractability.
Finally, extending the framework beyond approval utilities---for instance to participatory budgeting or richer preference languages---may help identify the general principles determining when proportional entitlements yield stable outcomes.

\newpage{}

\newpage
\bibliographystyle{ACM-Reference-Format}
\bibliography{bib}

\end{document}